\pdfoutput=1
\documentclass[acmsmall, screen]{acmart}
\setcopyright{rightsretained}
\acmJournal{PACMPL}
\acmYear{2020} \acmVolume{4} \acmNumber{OOPSLA} \acmArticle{140} \acmMonth{11} \acmPrice{}\acmDOI{10.1145/3428208}
\usepackage{booktabs}   %
\usepackage{subcaption} %
\usepackage{listings, tikz, tikz-cd, enumitem, amsthm, bcprules, enumitem,
  pgfplots, multirow, pifont, mathtools, extarrows, stmaryrd, wrapfig}
\usepackage[tight-spacing=true]{siunitx}

\usetikzlibrary{shapes.misc, positioning, shapes.geometric, backgrounds,
  patterns, fit, calc}
\usepackage[linewidth=0.5pt]{mdframed}
\usepackage[linesnumbered,ruled]{algorithm2e}
\usepackage[linewidth=0.5pt]{mdframed}

\newcounter{magicrownumbers}

\tikzset{>=latex}

\pgfplotstableread[col sep = comma]{data/hmm.csv}\hmmdata
\pgfplotstableread[col sep = comma]{data/psi-mc.csv}\psimc
\pgfplotstableread[col sep = comma]{data/dice-mc.csv}\dicemc
\pgfplotstableread[col sep = comma]{data/webppl-mc.csv}\webmc
\pgfplotstableread[col sep = comma]{data/syminf-encryption.csv}\syminfenc
\pgfplotstableread[col sep = comma]{data/psi-encryption.csv}\psiinfenc
\pgfplotstableread[col sep = comma]{data/webppl-enc.csv}\webpplenc
\pgfplotstableread[col sep = comma]{data/modppl-diamond.csv}\syminfdiamond
\pgfplotstableread[col sep = comma]{data/modppl-ladder.csv}\syminfladder
\pgfplotstableread[col sep = comma]{data/psi-networking.csv}\psiinfnetworking
\pgfplotstableread[col sep = comma]{data/psi-diamond.csv}\psidiamond
\pgfplotstableread[col sep = comma]{data/sym-diamond.csv}\symdiamond
\pgfplotstableread[col sep = comma]{data/syminf-encryption-error.csv}\symencerror
\pgfplotstableread[col sep = comma]{data/psi-encryption-error.csv}\psiencerror
\pgfplotstableread[col sep = comma]{data/webppl-enc-error.csv}\webpplencerror
\pgfplotstableread[col sep = comma]{data/webppl-diamond.csv}\webppldiamond
\pgfplotstableread[col sep = comma]{data/webppl-ladder.csv}\webpplladder

\usetikzlibrary{patterns,calc,backgrounds}

\tikzstyle{nnf}=[
  >=stealth,font=\small,auto,scale=0.7,every node/.style={scale=0.7}
]
\tikzstyle{extnode}=[
  draw,circle,inner sep=2pt,fill=white
]

\tikzstyle{leafnode}=[
  draw,fill=gray!20,inner sep=3.5pt
]
\tikzstyle{constnode}=[
  draw,fill=white,inner sep=3.5pt
]
\tikzstyle{label}=[
  fill=white,inner sep=2.5pt
]

\tikzstyle{acarrow}=[
    decoration={markings,mark=at position 1 with {\arrow[scale=0.6]{>}}},
    postaction={decorate},
    shorten >=0.4pt,
    >=latex,
    line width=0.1
]

\tikzstyle{bnarrow}=[
    decoration={markings,mark=at position 1 with {\arrow[scale=1.5]{>}}},
    postaction={decorate},
    shorten >=0.7pt,
    >=latex,
    line width=0.3
]
\tikzstyle{bayesnet}=[
  >=latex, thick, auto
]
\tikzstyle{bnnode}=[
  draw,ellipse,minimum size=7mm,inner sep=1pt,font=\small
]
\tikzstyle{cpt}=[
  font=\footnotesize
]

\tikzstyle{graph}=[
  >=stealth,font=\small,auto,scale=1,every node/.style={scale=1}
]
\tikzstyle{node}=[
  draw,circle,inner sep=3pt,fill=white
]

\tikzstyle{bdd}=[
  >=latex, thick, >=stealth, font=\small,auto,scale=0.9,every node/.style={scale=0.9}
]
\tikzstyle{bddnode}=[
  draw,circle,inner sep=0pt,fill=white,minimum size=5.5mm
]

\tikzstyle{highedge}=[
    line width=0.9
]
\tikzstyle{lowedge}=[
    line width=0.9,dotted
]
\tikzstyle{bddterminal}=[
  draw,fill=gray!20,inner sep=2.5pt, font=\small
]

\lstdefinestyle{compact}{
  \ttfamily\tiny
}

\usepackage{pifont}%
\newcommand{\xmark}{\ding{55}}%

\makeatletter
\newcommand{\xMapsto}[2][]{\ext@arrow 0599{\Mapstofill@}{#1}{#2}}
\def\Mapstofill@{\arrowfill@{\Mapstochar\Relbar}\Relbar\Rightarrow}
\makeatother

\newlength{\pardefault}
\newcommand{\Llet}[2]{ {\texttt{let}~#1~\texttt{in}~#2} } %
\newcommand{\Lobs}[1]{ {\texttt{observe}~#1}}
\newcommand{\Lflip}[1]{ {\texttt{flip}~#1}}
\newcommand{\Lfst}[1]{ {\texttt{fst}~#1}}
\newcommand{\Lsnd}[1]{ {\texttt{snd}~#1}}
\newcommand{\Lite}[3]{ {\texttt{if}~#1~\texttt{then}~#2~\texttt{else}~#3}}

\newcommand{\dice}[0]{\texttt{Dice}}

\newcommand{\bigO}[0] {\mathcal{O}}

\newcommand{\lx}[0] {\mathbf{x}}

\newcommand{\lf}[0] {\mathbf{f}}

\newcommand{\tx}[0] {\texttt{x}}
\newcommand{\ty}[0] {\texttt{y}}
\newcommand{\tz}[0] {\texttt{z}}

\newcommand{\comp}[0] {\rightsquigarrow}

\newcommand{\xbroadand}[1]{{\underset{#1}{\land}}}
\newcommand{\xpointor}[1]{{\overset{{.}}{\underset{#1}{\lor}}}}

\newcommand{\xpointphi}[0]{{\overset{{.}}{\varphi}}}
\newcommand{\tup}[0] {\xpointphi}

\newcommand{\defeq}[0] {\triangleq}
\newcommand{\obs}[0]{\gamma}
\newcommand{\bool}[0]{ \mathbf{Bool} }
\newcommand{\wmc}[0]{ \mathtt{WMC} }
\newcommand{\form}[0]{ F }

\newcommand{\true}[0]{ \mathtt{T} }

\newcommand{\false}[0]{ \mathtt{F} }

\newcommand{\prog}[0]{ \texttt{p}}

\newcommand{\te}[0]{ \texttt{e}}  %

\newcommand{\mods}[0]{ \mathtt{Mods}} 
\newcommand{\dbracket}[1]{ \left\llbracket { #1 } \right\rrbracket}

\def\Pr{\mathop{\rm Pr}\nolimits}

\begin{document}

\title[Scaling Exact Inference for Discrete Probabilistic Programs]{Scaling Exact Inference for Discrete Probabilistic Programs}
\author{Steven Holtzen}
\affiliation{
  \institution{University of California, Los Angeles}            %
}
\email{sholtzen@cs.ucla.edu}          %

\author{Guy Van den Broeck}
\affiliation{
  \institution{University of California, Los Angeles}            %
}
\email{guyvdb@cs.ucla.edu}          %

\author{Todd Millstein}
\affiliation{
  \institution{University of California, Los Angeles}            %
}
\email{todd@cs.ucla.edu}          %

\begin{abstract}
Probabilistic programming languages (PPLs) are an expressive means of representing and reasoning about probabilistic models. The computational challenge of {\em probabilistic inference} remains the primary roadblock for applying PPLs in practice. Inference is fundamentally hard, so there is no one-size-fits all solution. In this work, we target scalable inference for an important class of probabilistic programs: those whose probability distributions are {\em discrete}. Discrete distributions are common in many fields, including text analysis, network verification, artificial intelligence, and graph analysis, but they prove to be challenging for existing PPLs.

We develop a domain-specific probabilistic programming language called \dice{} that features a new approach to exact discrete probabilistic program inference. \dice{} exploits program structure in order to {\em factorize} inference, enabling us to perform exact inference on probabilistic programs with hundreds of thousands of random variables. Our key technical contribution is a new reduction from discrete probabilistic programs to \emph{weighted model counting} (WMC). This reduction separates the structure of the distribution from its parameters, enabling logical reasoning tools to exploit that structure for probabilistic inference. We (1) show how to compositionally reduce \dice{} inference to WMC, (2) prove this compilation correct with respect to a denotational semantics, (3) empirically demonstrate the performance benefits over prior approaches, and (4) analyze the types of structure that allow \dice{} to scale to large probabilistic programs.
\end{abstract}


\maketitle

\section{Introduction}
\label{sec:intro}
The primary analysis task
in probabilistic programming languages is \emph{probabilistic inference}, computing the probability that an
event occurs according to the distribution defined by the program. Probabilistic
inference generalizes many well-known program analysis tasks, such as
reachability, and hence inference for a sufficiently expressive language is an extremely
hard program analysis task. The key to scaling inference is to strategically make
assumptions about the structure of programs and place restrictions on which
programs can be written, while retaining a useful and expressive language.

In this paper, we focus on scaling inference for an important class of
probabilistic programs: those whose probability distributions are {\em
discrete}. Most PPLs today focus on handling continuous random variables. In
the continuous setting one usually desires approximate inference techniques, such as forms
of sampling~\citep{wingate2013automated, kucukelbir2015automatic,
jordan1999introduction, bingham2019pyro, dillon2017tensorflow,
carpenter2016stan, nori2014r2, chaganty2013efficiently}. However, handling
continuous variables typically requires making strong assumptions about the structure of
the program: many of these inference techniques have strict
differentiability requirements that preclude their application to programs with
discrete random variables. For instance, momentum-based sampling algorithms like
HMC and NUTS~\citep{hoffman2014no} and many variational
approximations~\citep{kucukelbir2017automatic} are restricted to continuous latent random
variables and almost-everywhere differentiability of the posterior distribution.
Yet many application domains are naturally discrete: for example mixture models,
networks and graphs, ranking and voting, and text. This key deficiency in some of the most popular
PPLs has led to a recent rise in interest in handling
discreteness in probabilistic programs~\citep{obermeyer2019tensor,
gorinova2019automatic, zhou2019divide}.

In this work we focus entirely on the challenge of designing a fast and efficient
discrete probabilistic program inference engine. We describe \dice{}, a domain-specific language for representing discrete
probabilistic programs, along with a new algorithm for exact inference for such
programs. \dice{} extends a simple first-order, non-recursive functional language
with support for making discrete probabilistic choices. It also provides
first-class {\em observations}, which enables \dice{} to support Bayesian
reasoning in the presence of evidence. %

Discrete programs are not a new challenge, and there are existing PPLs that
support exact inference for discrete probabilistic
programs~\citep{narayanan2016probabilistic, gehr2016psi, Sankaranarayanan2013,
Albarghouthi2017_1, goodman2014design, wang2018pmaf, Claret2013, Pfeffer2007,
bingham2019pyro, geldenhuys2012probabilistic}. However, we identify several compelling
example programs from text analysis, network verification, and discrete
graphical models on which existing methods fail. The reason that they fail is
that the existing methods do not find and automatically exploit the necessary factorizations and structure.
\dice{}'s inference algorithm is inspired by
techniques for exact inference on discrete graphical models, which leverage the
graphical structure to factorize the inference computation. For example, a common
property is {\em conditional independence}: if a variable {\tt z} is conditionally
independent of {\tt x} given {\tt y}, then {\tt y} acts as a kind of {\em
interface} between {\tt x} and {\tt z} that allows inference to be split into
two separate analyses. This kind of structure abounds in typical probabilistic
programs. For example, a function call is conditionally independent of the
calling context given the actual argument value. \dice{}'s inference algorithm
automatically identifies and exploits these independences in order to factorize
inference. This enables \dice{} to scale to extremely large discrete probabilistic
programs: our experiments in Section~\ref{sec:experiments} show \dice{}
performing exact inference on a real-world probabilistic program that is 1.9MB large.

At its core, \dice{} builds on the \emph{knowledge compilation} approach to
probabilistic inference~\cite{Chavira2008, chavira2005compiling, Darwiche09,
  Fierens2015, chavira2006compiling}.
We show how to compile \dice{} programs to \emph{weighted Boolean formulas}
(WBF) and then perform exact inference via \emph{weighted model counting} (WMC)
on those formulas. We use binary decision diagrams (BDDs) to represent these
formulas. We show that during compilation, these BDDs naturally identify and exploit conditional independence and other
forms of structure, thereby avoiding exponential explosion for many classes of
interesting programs. Further, BDDs support efficient WMC,
linear in the size of the BDD.

Employing knowledge compilation for probabilistic inference in \dice{} requires
us to generalize the prior approaches in several ways. First, in order to support
logical compilation of traditional programming constructs such as conditionals,
local variables, and arbitrarily nested tuples, we develop novel compilation
rules that compositionally associate \dice{} programs with weighted Boolean
formulas. A key challenge here is supporting arbitrary observations. To do
this, a \dice{} program, as well as each \dice{} function, is compiled to {\em
two} BDDs. Intuitively, one BDD represents all possible executions of the
program, ignoring observations, and the other BDD represents all executions that
satisfy the program's observations. We show how to use WMC on these formulas to
perform exact Bayesian inference, with arbitrary observations throughout the
program. Second, \dice{} compiles functions {\em modularly}: each function is
compiled to a BDD once, and we exploit efficient BDD composition operations to
reuse this BDD at each call site. This technique produces the same final BDD
that would otherwise be produced, but it allows us to amortize the costly BDD
construction phase across all callers, which we demonstrate can provide
orders-of-magnitude speedups.

In sum, this paper presents the following technical contributions:
\begin{itemize}[leftmargin=*]
 \item We describe the \dice{} language and illustrate its utility through three motivating examples (Section~\ref{sec:motivation}).

 \item We formalize \dice{}'s semantics (Section~\ref{sec:language}) and its compilation to weighted Boolean formulas (Section~\ref{sec:compilation}).  We prove that the compilation rules are correct with respect to the denotational semantics:
 the probability distribution represented by a compiled \dice{}
program is equivalent to that of the original program.

\item We empirically compare \dice{}'s performance to that of prior PPLs with
exact inference (Section~\ref{sec:experiments}). We describe new and
challenging benchmark probabilistic programs from cryptography,
network analysis, and discrete Bayesian networks, and show that \dice{} scales
to orders-of-magnitude larger programs than existing probabilistic programming
languages, and is competitive with specialized Bayesian network inference
engines on certain tasks.

\item We analyze some of the benefits of \dice{}'s compilation strategy in
Section~\ref{sec:analysis}. First we note that \dice{} inference is
\textsf{PSPACE}-hard. Then we characterize cases where \dice{} scales
efficiently, and which types of structure it exploits in the distribution. We illustrate where to find that structure in the program code as well as the compiled  BDD form. We use these results
to provide a technical comparison with prior exact inference algorithms.
\end{itemize}
\dice{} is available at \url{https://github.com/SHoltzen/dice}.

\section{An Overview of \dice{}}
\label{sec:motivation}
This section overviews the \dice{} language and its inference algorithm.  First we use a simple example program to
show how \dice{} exploits program structure to perform inference in
a \emph{factorized} manner. Then we use an example from network verification to show how \dice{} exploits the modular structure of
functions. Finally we use a cryptanalysis example to illustrate how inference in \dice{} is augmented to support Bayesian inference in the presence of \emph{evidence}.

\subsection{Factorizing Inference}

Probabilistic programming languages (PPLs) endow traditional programming
languages with probabilistic operations that enable the construction of
probability distributions~\citep{Huang2016, Albarghouthi2017_1, Claret2013,
 borgstrom2011measure, Kozen1979}, and \dice{} is no exception. Specifically, \dice{}
extends a first-order functional language that supports non-recursive functions
and a form of bounded iteration. Despite its simplicity, this language can
express a wide variety of statistical models, and exact probabilistic inference
in \dice{} is fundamentally hard. In addition to its standalone usage, we
anticipate \dice{} being used as a core language for discrete inference inside
other probabilistic programming systems.

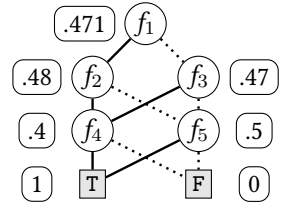
\begin{figure}
\begin{subfigure}[b]{0.5\linewidth}
 \begin{lstlisting}[mathescape=true, basicstyle=\ttfamily\small]
let x = flip$_1$ 0.1 in
let y = if x then flip$_2$ 0.2 else 
    flip$_3$ 0.3 in
let z = if y then flip$_4$ 0.4 else 
    flip$_5$ 0.5 in z
\end{lstlisting}
     \caption{Example \dice{} program.}
     \label{fig:cond_ex}
   \end{subfigure}~
   \begin{subfigure}[b]{0.45\linewidth}
     \centering
    \begin{tikzpicture}
      \def\lvl{20pt}
    \node (f1) at (0, 0) [bddnode] {$f_1$};
    \node[draw, rounded corners] at ($(f1) + (-23bp, 0bp)$)  {.471};

    \node (f2) at ($(f1) + (-20bp, -\lvl)$) [bddnode] {$f_2$};
    \node[draw, rounded corners] at ($(f2) + (-21bp, 0bp)$)  {.48};

    \node (f3) at ($(f1) + (20bp, -\lvl)$) [bddnode] {$f_3$};
    \node[draw, rounded corners] at ($(f3) + (21bp, 0bp)$)  {.47};

    \node (f4) at ($(f2) + (0bp, -\lvl)$) [bddnode] {$f_4$};
    \node[draw, rounded corners] at ($(f4) + (-21bp, 0bp)$)  {.4};

    \node (f5) at ($(f3) + (0bp, -\lvl)$) [bddnode] {$f_5$};
    \node[draw, rounded corners] at ($(f5) + (21bp, 0bp)$)  {.5};
    
    \node (true) at ($(f4) + (0bp, -\lvl)$) [bddterminal] {$\true$};
    \node[draw, rounded corners] at ($(true) + (-21bp, 0bp)$)  {1};

    \node (false) at ($(f5) + (0bp, -\lvl)$) [bddterminal] {$\false$};
    \node[draw, rounded corners] at ($(false) + (21bp, 0bp)$)  {0};

    \begin{scope}[on background layer]
      \draw [highedge] (f1) -- (f2);
      \draw [lowedge] (f1) -- (f3);
      \draw [highedge] (f2) -- (f4);
      \draw [lowedge] (f2) -- (f5);
      \draw [highedge] (f3) -- (f4);
      \draw [lowedge] (f3) -- (f5);

      \draw [highedge] (f4) -- (true);
      \draw [highedge] (f5) -- (true);
      \draw [lowedge] (f4) -- (false);
      \draw [lowedge] (f5) -- (false);
    \end{scope}
    \end{tikzpicture}
    \caption{
      Compiled BDD with weighted model counts.
    }
    \label{fig:cond_bdd}
  \end{subfigure}
  \caption{Illustration of compiling a \dice{} program that exploits
    factorization.}
  \label{fig:motiv_init}
\end{figure}

We begin with a simple motivating example that highlights the challenge of
performing inference efficiently and how \dice{} meets this challenge. Consider
the example \dice{} program in Figure~\ref{fig:cond_ex}. The syntax is standard
except for the introduction of a probabilistic expression $\Lflip{\theta}$,
which flips a coin that returns true with probability $\theta$ and false with
probability $1-\theta$. The subscript on each \texttt{flip} is not part of the
syntax but rather used to refer to them uniquely in our discussion.

The goal of probabilistic inference is to produce a program's output probability
distribution, so in Figure~\ref{fig:cond_ex} we desire the probability that {\tt
z} is true and the probability that {\tt z} is false. Consider computing the
probability that {\tt z} is true, which we denote $\Pr(\tz=\true)$. The most
straightforward way to compute this quantity is via \emph{path enumeration}: we
can consider all possible assignments to all \texttt{flip}s and sum the
probability of all assignments under which $\tz = \true$. A number of existing
PPLs directly implement path enumeration to perform
inference~\citep{Sankaranarayanan2013, Albarghouthi2017_1,
geldenhuys2012probabilistic, filieri2013reliability}. Concretely this would
involve computing the following sum of products:
\begin{align}
  \underbrace{0.1}_{\tx = \true} \cdot \underbrace{0.2}_{\ty = \true} \cdot \underbrace{0.4}_{\tz = \true} ~~+~~
  \underbrace{0.1}_{\tx = \true} \cdot \underbrace{0.8}_{\ty = \false} \cdot \underbrace{0.5}_{\tz = \true} ~~+~~
  \underbrace{0.9}_{\tx = \false} \cdot \underbrace{0.3}_{\ty = \true} \cdot \underbrace{0.4}_{\tz = \true} ~~+~~
  \underbrace{0.9}_{\tx = \false} \cdot \underbrace{0.7}_{\ty = \false} \cdot \underbrace{0.5}_{\tz = \true}
  \label{eq:exhaustive}
\end{align} %

In this work we focus on the problem of scaling inference, so we ask: how does
exhaustive enumeration scale as this program grows in size? In this case we grow
the program by adding one additional layer to the chain of \texttt{flip}s that
depends on the previous. With this growing pattern, the number of terms that a
path enumeration must explore grows exponentially in the number of layers, so
clearly exhaustive enumeration does not scale on this simple example. Despite
its apparent simplicity, many existing inference algorithms cannot scale to
large instances of this example; see Figure~\ref{fig:motiv_scale} in Section~\ref{sec:experiments}.
However, the sum in
Equation~\ref{eq:exhaustive} has redundant computation, and thus can be
factorized as
\begin{align}
  \underbrace{0.1}_{\tx = \true} \cdot
  \Big(\underbrace{0.2}_{\ty = \true} \cdot \underbrace{0.4}_{\tz = \true} +
  \underbrace{0.8}_{\ty = \false} \cdot \underbrace{0.5}_{\tz = \true} \Big) + 
  \underbrace{0.9}_{\tx = \false} \cdot
  \Big(\underbrace{0.3}_{\ty = \true} \cdot \underbrace{0.4}_{\tz = \true} +
  \underbrace{0.7}_{\ty = \false} \cdot \underbrace{0.5}_{\tz = \true} \Big).
  \label{eq:factor}
\end{align}
Such factorizations are abundant in this example, and in many others. \dice{}
exploits these factorizations to scale, and in Section~\ref{sec:experiments} we
show that \dice{} scales to orders of magnitude larger programs than existing
methods in part by exploiting these forms of factorization.
Such factorizations are extremely common in
probabilistic models, and finding and exploiting them is an essential strategy
for scaling exact inference algorithms, for example for graphical
models~\citep{Chavira2008, Darwiche09, koller2009probabilistic,
boutilier1996context, Pearl88b}.

\paragraph{Factorized inference in \dice{}} Inference in \dice{} is designed to
find and exploit factorizations like the one shown above. The key insight is to
separate the logical representation of the state space of the program from the
probabilities, which allows \dice{} to identify factorizations implied by the
structure of the program that are otherwise difficult to detect. This separation
is achieved by compiling each program to a {\em weighted Boolean formula}:
\begin{definition}
  \label{def:wbf}
Let $\varphi$ be a Boolean formula over variables $X$, let $L$ be the set of all
\emph{literals} (assignments to variables) over $X$, and $w : L \rightarrow
\mathbb{R}$ be a \emph{weight function} that associates a real-valued weight with each literal
$L$. The pair $(\varphi, w)$ is a \emph{weighted Boolean formula} (WBF).
\end{definition}

To compile the program in Figure~\ref{fig:cond_ex} into a WBF, we introduce one
Boolean variable $f_i$ for each expression
\texttt{flip$_i$~$\theta$} in the program.  
Our goal is for the resulting boolean formula over these variables to
represent all possible {\tt flip} valuations that cause {\tt z} to be true, so
one choice of WBF is $\varphi_{ex} = f_1f_2f_4 \lor f_1\bar{f_2}f_5 \lor
\bar{f_1}f_3f_4 \lor
\bar{f_1}\bar{f_3}f_5$.  Separately, the weight function represents the specific probabilities for each expression
\texttt{flip$_i$~$\theta$} from the program: the weight of $f_i$ is
$\theta$ if $f_i$ is true and $1-\theta$ otherwise. 

Once the program is associated with a WBF, we can perform probabilistic inference via a
{\em weighted model count} (WMC). Formally, for a formula $\varphi$ over
variables $X$, a sentence $\omega$ is a \emph{model} of $\varphi$ if it is a
conjunction of literals, contains every variable in $X$, and $\omega \models
\varphi$. We denote the set of all models of $\varphi$ as $\mods(\varphi)$.
The \emph{weight of a model}, denoted $w(\omega)$, is the product of the
weights of each literal $w(\omega) \defeq \prod_{l \in \omega} w(l)$.
Then, the following defines the WMC task:
\begin{definition}
  \label{def:wmc}
  Let $(\varphi, w)$ be a weighted Boolean formula. The \emph{weighted
model count} ($\wmc$) of $(\varphi, w)$ is the sum of the weights of each model,
$\wmc(\varphi, w) \triangleq
\sum_{\omega \in \mods(\varphi)} w(\omega)$.
\end{definition}

What has been achieved? So far, not much! The $\wmc$ task is known to be
\#P-hard for arbitrary Boolean formulas. Indeed, our formula $\varphi_{ex}$
above is isomorphic to the structure of Equation~\ref{eq:exhaustive}, so the WMC
calculation over it will be essentially equivalent. However, it has been
observed in the AI literature that certain representations of Boolean formulas
--- such as binary decision diagrams (BDDs) --- both exploit the structure of a
formula to minimize its representation and support \emph{linear time weighted
model counting}, and as such are useful compilation targets~\citep{Chavira2008,
Darwiche2002, Bryant86}. The field of compiling Boolean
formulas to representations that support tractable weighted model counting is
broadly known as \emph{knowledge compilation}, and \emph{inference via knowledge
compilation} is currently the state-of-the-art inference algorithm for certain
kinds of discrete Bayesian networks~\citep{Chavira2008} and probabilistic logic
programs~\citep{Fierens2015}.

\dice{} utilizes the insights of knowledge compilation to perform factorized
inference. First, the generated formula $\varphi$ in a compiled WBF is
represented as a BDD; Figure~\ref{fig:cond_bdd} shows the compiled BDD for the
program in Figure~\ref{fig:cond_ex}. A solid edge denotes the case where the
parent variable is true and a dotted edge denotes the case where the parent
variable is false. This BDD is logically equivalent to $\varphi_{ex}$ but the
BDD's construction process exploits the program's conditional independence to
efficiently produce a compact canonical representation. Specifically, there is a
single subtree for $f_4$, which is shared by both the path coming from $f_2$ and
the path coming from $f_3$, and similarly for $f_5$. These shared sub-trees are
induced by conditional independence: fixing \texttt{y} to the value true ---
and hence guaranteeing that a path to $f_4$ is taken in the BDD --- screens off
the effect of \texttt{x} on \texttt{z}, and hence reduces both the size of the
final BDD and the cost of constructing it. The BDD automatically finds and
exploits such factorization opportunities by caching and reusing repetitious
logical sub-functions.

\dice{} performs inference on the original probabilistic program via $\wmc$ once
the program is compiled to a BDD. Crucially, it does so without exhaustively
enumerating all paths or models. By virtue of the shared sub-functions, the BDD
in Figure~\ref{fig:cond_bdd} directly describes how to compute the WMC in the
factorized manner. Observe that each node is annotated with the weighted model
count, which is computed in linear time in a single bottom-up pass of the BDD.
For instance, the WMC at node $f_2$ is given by taking the weighted sum of the
WMC of its children, $0.2 \times 0.4 + 0.8 \times 0.5$. Finally, the sum taken
at the root of the BDD (the node $f_1$) is exactly the factorized sum in
Equation~\ref{eq:factor}.

\subsection{Leveraging Functional Abstraction}
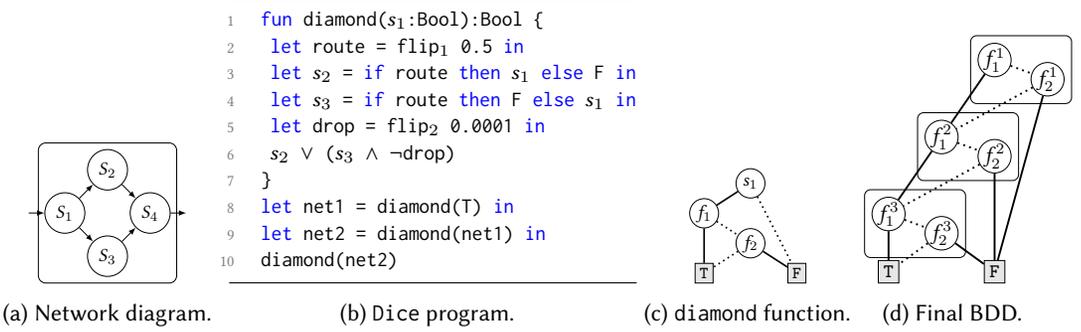
\begin{figure*}
  \centering
  \begin{subfigure}[b]{0.2\linewidth}
    \centering
    \scalebox{0.7}{
    \begin{tikzpicture}[node distance=0.3cm]
      \node[] (init) {};
      \node[draw, circle, right=of init] (S1) {$S_1$};
      \node[right= of S1] (C) {};
      \node[draw, circle, above=of C] (S2) {$S_2$};
      \node[draw, circle, below=of C] (S3) {$S_3$};
      \node[draw, circle, right=of C] (S4) {$S_4$};
      \node[right=of S4] (final) {};

      \node[fit=(S1)(S2)(S3)(S4), draw, rounded corners] {};

      \draw[->] (init) -- (S1);
      \draw[->] (S1) -- (S2);
      \draw[->] (S1) -- (S3);
      \draw[->] (S2) -- (S4);
      \draw[->] (S3) -- (S4);
      \draw[->] (S4) -- (final);
    \end{tikzpicture}
    }
    \caption{Network diagram.}
    \label{fig:ex_net_a}
  \end{subfigure}
  ~
  \begin{subfigure}[b]{0.4\linewidth}
    \centering

\begin{lstlisting}[mathescape=true,
    basicstyle=\smaller\ttfamily,
]
fun diamond($s_1$:Bool):Bool {
 let route = flip$_1$ 0.5 in
 let $s_2$ = if route then $s_1$ else F in
 let $s_3$ = if route then F else $s_1$ in
 let drop = flip$_2$ 0.0001 in
 $s_2$ $\lor$ ($s_3$ $\land$ $\neg$drop)
}
let net1 = diamond(T) in
let net2 = diamond(net1) in 
diamond(net2)
\end{lstlisting}
    \caption{
      \dice{} program.
    }
    \label{fig:ex_net_b}
  \end{subfigure}~~
    \begin{subfigure}[b]{0.2\linewidth}
      \centering
    \scalebox{0.7}{
    \begin{tikzpicture}
      
    \def\lvl{16pt}
    \node (inc) at (0bp,0bp) [bddnode] {$s_1$};

    \node (f1) at ($(inc) + (-25bp, -\lvl)$) [bddnode] {$f_1$};
    \node (f2) at ($(f1) + (25bp, -\lvl)$) [bddnode] {$f_2$};

    \node (vt) at ($(f2) + (-25bp, -\lvl)$) [bddterminal] {$\true$};

    \node (vf) at ($(f2) + (25bp, -\lvl)$) [bddterminal] {$\false$};
    \begin{scope}[on background layer]
      \draw [highedge] (inc) -- (f1);
      \draw [lowedge] (inc) -- (vf);

      \draw [lowedge] (f1) -- (f2);
      \draw [highedge] (f1) -- (vt);

      \draw [highedge] (f2) -- (vf);
      \draw [lowedge] (f2) -- (vt);
    \end{scope}

  \end{tikzpicture}
  }
    \caption{
      \texttt{diamond} function.
    }
    \label{fig:ex_net_c}
  \end{subfigure}~
  \begin{subfigure}[b]{0.18\linewidth}
    \centering
    \scalebox{0.8}{
    \begin{tikzpicture}
      
    \def\lvl{18pt}
    \node (f11) at (0bp,0bp) [bddnode] {$f_1^1$};
    \node (f21) at ($(f11) + (25bp, -0.5*\lvl)$) [bddnode] {$f_2^1$};

    \node (f12) at ($(f11) + (-25bp, -2*\lvl)$) [bddnode] {$f_1^2$};
    \node (f22) at ($(f12) + (25bp, -0.5*\lvl)$) [bddnode] {$f_2^2$};

    \node (f13) at ($(f12) + (-25bp, -2*\lvl)$) [bddnode] {$f_1^3$};
    \node (f23) at ($(f13) + (25bp, -0.5*\lvl)$) [bddnode] {$f_2^3$};

    \node (vt) at ($(f23) + (-25bp, -\lvl)$) [bddterminal] {$\true$};

    \node (vf) at ($(f23) + (25bp, -\lvl)$) [bddterminal] {$\false$};

    \node[fit=(f11)(f21), rounded corners=1mm, draw] {};
    \node[fit=(f12)(f22), rounded corners=1mm, draw] {};
    \node[fit=(f13)(f23), rounded corners=1mm, draw] {};

    \begin{scope}[on background layer]
      \draw [lowedge] (f11) -- (f21);
      \draw [highedge] (f11) -- (f12);
      
      \draw [lowedge] (f12) -- (f22);
      \draw [highedge] (f12) -- (f13);

      \draw[lowedge] (f21) -- (f12);
      \draw[highedge] (f21) -- (vf);

      \draw [highedge] (f23) -- (vf);
      \draw [lowedge] (f23) -- (vt);

      \draw [lowedge] (f13) -- (f23);
      \draw [highedge] (f13) -- (vt);

      \draw [highedge] (f22) -- (vf);
      \draw [lowedge] (f22) -- (f13);
    \end{scope}

  \end{tikzpicture}
  }
    \caption{
      Final BDD.
    }
    \label{fig:ex_net_d}
  \end{subfigure}
 
  \caption{A sub-network, its description as a probabilistic program, a
    compiled function, and the final BDD.}
  \label{fig:ex_net}
\end{figure*}

The previous section highlights how \dice{} exploits factorization that comes
from conditional independences in the program. One common source of such
independences is functional abstraction: the behavior of a function call is
independent of the calling context, given the actual argument. \dice{} inference
as described above automatically exploits this structure as part of the BDD
construction. In addition, \dice{} exploits functional abstraction in an
orthogonal manner by modularly compiling a BDD for each function once and then
reusing this BDD at each call site, thereby amortizing the cost of the BDD
construction across all callers.

To illustrate the benefits of functional abstraction, we adapt an example from recent work in probabilistic verification of computer networks via probabilistic programs~\cite{gehr2018bayonet}. 
Figure~\ref{fig:ex_net_a} shows a ``diamond'' network that contains four
servers, labeled $S_i$. The network's behavior is naturally probabilistic, to
account for dynamics such as load balancing and congestion. In this case, server
$S_1$ forwards an incoming packet to either $S_2$ or $S_3$, each with
probability $50\%$. In turn, those servers forward packets received from $S_1$
to $S_4$, except that $S_3$ has a $0.1\%$ chance of dropping such a packet.
The {\tt diamond} function in Figure~\ref{fig:ex_net_b} defines the behavior of
this network as a probabilistic program in \dice{}. The argument boolean ${s_1}$
represents the existence of an incoming packet to $S_1$ from the left, and the
function returns a boolean indicating whether a packet was 
delivered to $S_4$.

As mentioned above,
\dice{} compiles functions modularly, so \dice{} first compiles the \texttt{diamond} function to a BDD, shown in Figure~\ref{fig:ex_net_c}.
The variable $s_1$ represents the unknown input to the function, and the $f_i$
variables represent the {\tt flip}s in the function body, as in our previous
example. Next \dice{} will create the BDD for the ``main'' expression in lines
8--10 of Figure~\ref{fig:ex_net_b}. During this process, the BDD for the {\tt
diamond} function is reused at each call site using standard BDD composition
operations like
conjunction (Section~\ref{sec:compilation} describes this in more detail).  
The final BDD for the program is shown in Figure~\ref{fig:ex_net_d}, where each variable $f_i^j$ represents the $i$th {\tt flip} in the $j$th call to {\tt diamond}.

The final BDD automatically identifies and exploits functional abstraction.  For example, the structure of the BDD makes it clear that the third call to {\tt diamond} depends only on the output of the second call to {\tt diamond}, rather than the particular execution path taken to produce that output.  As a result, even though there are three sub-networks, and therefore $2^6$ possible joint assignments 
to \texttt{flip}s, the BDD only has 8 nodes.  More generally, this BDD will grow linearly in the number of composed {\tt diamond} calls, though the number of possible executions grows exponentially.  Hence functional abstraction both produces smaller BDDs, which leads to faster WMC computation, and reduces BDD compilation time by compiling each function once.  We show in Section~\ref{sec:experiments} that these capabilities provide orders of magnitude speedups in inference.

\subsection{Bayesian Inference \& Observations}
\label{sec:motiv_bayes}
\begin{wrapfigure}{r}{0.5\textwidth}
   \centering
\begin{lstlisting}[mathescape=true]
fun EncryptChar(key:int, c:char):Bool {
 let randomChar = ChooseChar() in
 let ciphertext = (randomChar+key)%
 let fail = flip 0.0001 in
 if fail then true else 
   observe ciphertext == c
}
let k = UniformInt(0, 25) in 
let _ = EncryptChar(k, 'H') in 
$\cdots$ // encrypt $n$ total characters
in k
\end{lstlisting}
  \caption{A frequency analyzer for a noisy Caesar cipher.}
  \label{fig:caesar}
\end{wrapfigure}

Bayesian inference is a general and popular technique for reasoning about the
probability of events in the presence of {\em evidence}. \dice{}, similar to
other PPLs, supports Bayesian reasoning through an {\tt observe} expression.
Specifically, the expression ``{\tt observe~e}'' represents evidence (or an {\em
observation}) that {\tt e} is true; the expression always evaluates to true, but it has the side effect that executions on which {\tt e} is not true are defined to have 0 probability.

\dice{} supports first-class observations, including inside of functions. An
example is shown in Figure~\ref{fig:caesar}, which shows another rich class of discrete
probabilistic inference problems that come from \emph{text analysis}. For this
problem the goal is to decrypt a given
piece of ciphertext by inferring the most likely encryption key. We assume that
the plaintext was encrypted using a \emph{Caesar cipher}, which simply shifts
characters by a fixed but unknown constant, so the encryption key is an integer between 0
and 25 (e.g., with key 2, ``abc'' becomes ``cde'').

The task of decrypting encrypted ciphertext can be cast as a probabilistic
inference task by using \emph{frequency analysis}~\cite{katz1996handbook}. In
the English language each letter has a certain probability of being used: for
instance, the frequency of letter ``E'' is 12.02\%. In Figure~\ref{fig:caesar},
the function \texttt{EncryptChar} is a {\em generative model} for how each
letter in the ciphertext was created. The function takes as an argument the
encryption {\tt key} as well as a received ciphertext character {\tt c}. First a
plaintext character {\tt randomChar} is chosen according to its empirical
distribution (the {\tt ChooseChar} function is not shown but straightforward).
Then this character is encrypted with the given key and we {\tt observe}
that the ciphertext is the actual ciphertext character {\tt c} that we received.
To make the inference problem more challenging and realistic, we assume that
there is a chance that the encryptor mistakenly forgets to encrypt a character,
in which case we do not perform the observation.
Initially, the key (\texttt{k}) is assumed to be uniformly random (line 6).
After invoking \texttt{EncryptChar} once for each received ciphertext character (lines 7--8), the posterior
distribution on the key is returned.

The interaction of probabilistic inference with observations is subtle.
Observations have a non-local and ``backwards'' effect on the probability distribution, which must
be carefully preserved when performing inference.  In our example, the observation inside of {\tt EncryptChar} affects the posterior distribution of its argument key.  These non-local effects are
the bane of sampling-based inference algorithms: observations can impose complex
constraints --- such as the need in our example for \texttt{ChooseChar} to draw the right character --- that make it challenging for sampling algorithms to find sufficiently many
valid samples (we highlight this challenge in Section~\ref{sec:experiments}).

The WBF compilation strategy outlined in the previous section is inadequate for
capturing the semantics of the \texttt{EncryptChar} function: this function
always returns \texttt{true}, so its compiled BDD would be trivial. Clearly this
is incorrect, since the \texttt{EncryptChar} function has an additional,
implicit effect on the program, by making certain encryption keys more or less
likely to be the correct one. To handle observations, we augment our compilation
strategy to produce a second logical formula, which we call the \emph{accepting
formula} and denote $\obs$. The accepting formula represents all possible
assignments to {\tt flip}s that cause all {\tt observe}s in the program to be
satisfied. Together the formulas $\varphi$ and $\obs$ capture the meaning of the
program: we can compute the posterior distribution on \texttt{k} by computing
weighted model counts of the form
$\wmc(\varphi \land \obs,w ) / \wmc(\obs, w)$ for each value of \texttt{k}. Note that $\obs$ serves two
roles: it constrains $\varphi$ to only those executions that satisfy the
observations, and its weighted model count computes the normalizing constant for
the final probability distribution.

\section{The \dice{} Language}
\label{sec:language}

\dice{} is a first-order functional language augmented with 
constructs for probabilistic programming. This section describes the language
formally, providing its syntax and compositional semantics.

\subsection{Syntax}
The core syntax of
\dice{} is given in Figure~\ref{fig:grammar}. We enforce an A-normal
form via the usage of atomic expressions (\texttt{aexp})~\cite{DBLP:conf/pldi/FlanaganSDF93}, which simplifies the semantics and
compilation rules. A program is a sequence of functions followed by the "main"
expression. Each function is non-recursive and can only call functions that
precede it. The language supports booleans, tuples, and typical operations over
those types. In addition to this core syntax our \dice{} implementation supports
convenient syntactic sugar for logical operations ($\land$, $\lor$, and $\neg$), statically
bounded loops, bounded-size integers, and arbitrary function arity, as we describe in
Section~\ref{sec:extensions}. We utilize these extensions in our examples freely.

\dice{} supports two probabilistic expressions. First, the expression {\tt flip
$\theta$}, where $\theta$ is a real number between 0 and 1, denotes the
distribution that has the value true with probability $\theta$ and false with
probability $1-\theta$. Second, the expression {\tt observe~e} enables Bayesian
reasoning by incorporating {\em evidence}. Specifically, {\tt observe~e}
represents the {\em observation} that {\tt e} has the value true. Semantically,
executions on which {\tt e} does not have the value true are defined to have 0
probability, which has the effect of implicitly increasing the probabilities of
other executions. We define the expression {\tt observe e} to always evaluate to true.

\begin{figure}
  \centering
\begin{lstlisting}[mathescape=true]
$\tau$ ::= $\mathbf{Bool}$ | $\tau_1 \times \tau_2$
$v$ ::= T | F | ($v$, $v$)
aexp ::= $x$ | $v$
e ::=  aexp | fst aexp | snd aexp | (aexp, aexp) | let $x$ = e in e | flip $\theta$ 
       | if aexp then e else e | observe aexp | $f$(aexp)
func ::= fun $f$($x$:$\tau$)$:\tau$ { e }
p ::= e | func p
\end{lstlisting}
\caption{Syntax for the core \dice{} language.  The metavariable $f$ ranges over function names, $x$ over variable names, and $\theta$ over real numbers in the range $[0,1]$.}
  \label{fig:grammar}
\end{figure}

\subsection{Semantics}
\begin{figure}
  \centering
  \begin{mdframed}
  \begin{align*}
    \dbracket{v_1}(v) \defeq \big(\delta(v_1)\big)(v)
    \quad\quad \dbracket{\Lfst{(v_1, v_2)}}(v) \defeq \big(\delta(v_1)\big)(v)
    \quad\quad \dbracket{\Lsnd{(v_1, v_2)}}(v) \defeq \big(\delta(v_2)\big)(v)
  \end{align*}
  \begin{align*}
    \dbracket{\Lite{v_g}{\te_1}{\te_2}}(v) \defeq
                                         \begin{cases}
                                         \dbracket{\te_1}(v) ~& \text{if } v_g = \true\\  
                                         \dbracket{\te_2}(v) ~& \text{if } v_g = \false\\
                                         0 \quad& \text{otherwise}
                                         \end{cases} \quad\quad
    \dbracket{\Lflip{\theta}}(v) \defeq& \begin{cases}
                                    \theta ~& \text{if }v = \true\\
                                    1-\theta ~& \text{if }v=\false\\
                                    0 ~& \text{otherwise}
                                  \end{cases}
  \end{align*}
  \begin{align*}
    \dbracket{\Lobs{v_1}}(v) \defeq
                          \begin{cases}
                            1 ~& \text{if } v_1 = \true \text{ and } v = \true,\\
                            0 ~& \text{otherwise}\\
                          \end{cases} \qquad\quad
      \dbracket{f(v_1)}(v) \defeq& \Big(\big(T(f)\big)(v_1)\Big)(v)
  \end{align*}
  \begin{align*}
     \dbracket{\Llet{x = \te_1}{\te_2}}(v)\defeq
    \sum_{v'}\dbracket{\te_1}(v') \times \dbracket{\te_2[x \mapsto v']}(v) 
  \end{align*}
\end{mdframed} 
\caption{Semantics for \dice{} expressions. The function $\delta(v)$ is a
probability distribution that assigns a probability of 1 to the value $v$ and 0
to all other values. The implicit context $T$ maps function names to their
semantics. }
  \label{fig:semantics}
\end{figure}

The semantics of \dice{} programs is largely standard.  We overview the semantics and highlight its key aspects and design choices.
We begin with the semantics of \dice{} expressions, which are naturally
represented as a probability distribution on values. Formally, let $V$ be the
set of all \dice{} values. Then, a \emph{discrete probability distribution} on
$V$ is a function $\Pr : V \rightarrow [0,1]$ such that $\sum_{v \in V} \Pr(v) =
1$.

Figure~\ref{fig:semantics} provides the semantics for \dice{} expressions.
First we will discuss the semantics of expressions without function calls and observations,
which are deferred to
Sections~\ref{sec:semantics_func} and \ref{sec:semantics_obs} respectively.
The \emph{semantic function} $\dbracket{\cdot}$ maps expressions to unnormalized
probability distributions (i.e., distributions that do not necessarily sum to 1).
The semantics of values and tuple access are straightforward. For example, the
semantics of the expression {\tt fst~($\false$,$\true$)} is the probability
distribution that assigns probability 1 to $\false$ and 0 to all other values.
The semantics for conditionals follows from its usual semantics. In well-formed (i.e., closed)
programs the conditional guard $v_g$ is always a value, because the language uses A-normal form. Hence, the semantics of \texttt{if} selects either the \emph{then}-branch or
\emph{else}-branch's semantics depending on the value of $v_g$. For
completeness of the semantics, we define the semantics of \texttt{if} to be the
always-zero function if the argument is not a Boolean. 

The semantics
for flips produces the corresponding Bernoulli distribution. For example, the
expression {\tt flip 0.8} denotes the
distribution that assigns $\true$ the probability 0.8, $\false$ the probability 0.2, and all other values the probability 0.

The most interesting case in the semantics is for {\tt let}, as it shows the
path enumeration that is
required when sequencing probabilistic expressions. Consider the example:
\begin{align}
  \Llet{x = \Lflip{0.1}}{\Lflip{0.4} \lor x} \tag{\textsc{ExLet}}\label{eq:let}
\end{align}
To compute the probability that (\ref{eq:let}) results in some value $v$, we must consider all possible ways in which that value could result, based on all possible values $v'$ for $x$. Concretely, to
evaluate $\dbracket{\Llet{x = \Lflip{0.1}}{\Lflip{0.4} \lor x}}(\true)$, the
following sum is computed:
  $\dbracket{\Lflip{0.1}}(\true) \times \dbracket{\Lflip{0.4 \lor x[x \mapsto \true]}}(\true)
  + \dbracket{\Lflip{0.1}}(\false) \times \dbracket{\Lflip{0.4 \lor x[x \mapsto \false]}}(\true)
  = 0.1 \times 1.0 + 0.9 * 0.4 = 0.46$.

\subsubsection{Functions and Programs} 
\label{sec:semantics_func}
\dice{} supports non-recursive functions.
We generalize the semantics of expressions to functions in a natural way.  Specifically, the semantics of a function
{\tt f} is a {\em conditional probability distribution}, which is a function
from each value $v$ to a probability distribution for~{\tt f($v$)}.
 Formally, the semantics of a function $\dbracket{{\tt func}} : V
\rightarrow V \rightarrow [0,1]$ is defined as follows:
\begin{align}
  \dbracket{\texttt{fun}~f(x:\tau):\tau' \{\te\}}(v) \defeq
  \dbracket{\te[x \mapsto v]}
\end{align}

We can now give a semantics to function calls.  To do so, we extend the semantics judgment to include a {\em function table} $T$, which is a finite map from function names to their conditional probability distributions.  Formally our semantics judgment for expressions now has the form 
  $\dbracket{\te}^T : V \rightarrow [0,1]$, and similarly for the semantics of function definitions above, but we leave $T$ implicit when it is clear from the context.  Figure~\ref{fig:semantics} provides the semantics of a function call:  the function's conditional probability distribution is found in $T$, and the probability distribution associated with the actual argument $v$ is retrieved.

Finally, we define the semantics of programs $\dbracket{{\tt p}}^T : V
\rightarrow [0,1]$. Intuitively, each function is given a semantics in the
context of the prior functions, and then the semantics of the program is defined
as the semantics of the ``main'' expression. We formalize this semantics
inductively via the following two rules, where $\bullet$ denotes the empty
sequence and $\eta({\tt func})$ denotes the name of the function {\tt func}:
\begin{align}
  \dbracket{\bullet~ \te}^T \defeq
  \dbracket{\te}^T \qquad \qquad
  \dbracket{{\tt func}~{\tt p}}^T \defeq
   \dbracket{\tt p}^{T \cup \big\{\eta({\tt func}) \mapsto \dbracket{\tt func}^T \big\}}.
\end{align}

\subsubsection{Observations \& Bayesian Conditioning}
\label{sec:semantics_obs}
Observations complicate the goal of associating a probability distribution with
each program expression. Our semantics of {\tt observe} in
Figure~\ref{fig:semantics} follows prior work by assigning probability 0 to a
failed observation~\cite{borgstrom2011measure, Kozen1979, Claret2013, Huang2016,
nori2014r2}. Now consider the following example program:
\begin{align*}
  \texttt{let x = flip 0.6 in let y = flip 0.3 in let \_ = observe x $\lor $ y in x}
  \tag{ObsProg}\label{eq:obsprog}
\end{align*}
Because the {\tt observe} expression is falsified when both {\tt x} and {\tt y} are false, that scenario has probability 0.  Hence
according to our semantics
$\dbracket{\textsc{ObsProg}}(\true) = 0.6$ and $\dbracket{\textsc{ObsProg}}(\false) =
0.12$.  As a result the meaning of this program is not a valid probability distribution.

The standard approach to handling this issue is to treat the semantics as
producing an {\em unnormalized} distribution, which need not sum to 1 and which
is normalized at the very end to produce a valid probability distribution for
the entire program. Here we explore the subtle properties of this unnormalized
distribution, which will serve a crucial purpose later during our compilation
strategy.
Let $\dbracket{\te}_A$ denote the normalizing constant
and $\dbracket{\te}_D$ denote the normalized
distribution for an expression. These two quantities can be straightforwardly
computed from the unnormalized semantics in Figure~\ref{fig:semantics}:
\begin{align}
  \dbracket{\te}_A \defeq \sum_{v} \dbracket{\te}(v), \qquad\qquad
  \dbracket{\te}_D(v) \defeq \frac{1}{\dbracket{\te}_A} \dbracket{\te}(v).
\end{align}
For instance, in the above example $\dbracket{\textsc{ObsProg}}_A = 0.12+0.6=0.72$, 
$\dbracket{\textsc{ObsProg}}_D(\true) = 0.6 / 0.72 \approx 0.83$, and 
$\dbracket{\textsc{ObsProg}}_D(\false) = 0.12 / 0.72 \approx 0.17$. In the event
that $\dbracket{\te}_A = 0$, the distributional semantics is also defined to be
zero.

By construction, $\dbracket{\cdot}_D$ always yields a
probability distribution (or the always-zero function in the event that the
accepting semantics is zero), so we call it the \emph{distributional semantics}.
This is the quantity that we need in order to answer inference queries. What does $\dbracket{\cdot}_A$ represent? Typically it is not given a
meaning but rather simply considered to be an arbitrary normalizing constant that is only
computed for the entire program. And indeed, the normalizing constant is
irrelevant for the purposes of performing global inference: the probabilities in
the unnormalized semantics can be scaled arbitrarily without changing
$\dbracket{\cdot}_D$. This ``normalize at the end'' mode of operation is
standard for many PPLs that use an unnormalized semantics~\citep{Fierens2015,
Claret2013}.

However, when reasoning about partial programs, the distributional semantics
alone is not sufficient. For example, consider these two functions:
\begin{align}
&\texttt{fun f(x:Bool):Bool \{ let y = x $\vee$ flip(0.5) in let z = observe y in y \}}\\
&\texttt{fun g(x:Bool):Bool \{ true \}}
\end{align}
Because the observation in {\tt f} requires {\tt y} to be true, the two
functions have the identical distributional semantics: they both return true
with probability 1, regardless of the argument \texttt{x}. However, these two
functions are not equivalent! Specifically, the observation in \texttt{f} has
the effect of changing the probability distribution of the argument {\tt x} when
the function is called. Concretely,
\begin{align*}
&\dbracket{\texttt{let x = flip 0.1 in let obs = f(x) in x}}_D(\true) = 0.1/0.55 \\
&\dbracket{\texttt{let x = flip 0.1 in let obs = g(x) in x}}_D(\true) = 0.1
\end{align*}

The quantity $\dbracket{\cdot}_A$ carries exactly the information needed to
distinguish these functions. Specifically, $\dbracket{\te}_A$ represents the
probability that {\tt e} has an {\em accepting} execution, which satisfies all
observations, so we call it the \emph{accepting semantics}. In the above
example, $\dbracket{{\tt g}(\false)}_A = 1$ but $\dbracket{{\tt f}(\false)}_A =
0.5$: the function call {\tt f($\false$)} will succeed only half of the time.
This quantity allows us to precisely compute the effect of the observation on
any caller.

In summary, the semantics in Figure~\ref{fig:semantics} computes an unnormalized
distribution. However, since the normalizing constant is exactly the accepting
probability, the semantics has the effect of computing two key quantities on
each program fragment, both of which are necessary to characterize its meaning:
its normalized probability distribution and its probability of accepting.
Later this accepting semantics will be explicitly represented during compilation
as the accepting formula.

\section{Probabilistic Inference for \dice{}}
\label{sec:compilation}

This section formalizes our approach to probablistic inference in \dice{} via
reduction to {\em weighted model counting} (WMC). In this style, a probabilistic
model is compiled to a {\em weighted Boolean formula} (WBF) such that WMC
queries on the WBF exactly correspond to inference queries on the original
model. This approach has been successfully used to perform exact inference in
discrete Bayesian networks as well as probabilistic
databases and logic programs~\citep{Chavira2008, Fierens2015, VdBFTDB17}. However, to our knowledge it has not
been previously applied to a probablistic programming language with traditional
programming language constructs, functions, and first-class observations.

The bulk of this section formalizes our novel algorithm for compiling \dice{}
programs to WBF. We describe this compilation in stages: first on the Boolean
sub-language, then with the addition of tuples, and finally with the addition of
functions. We also state and prove a correctness theorem, which formally relates
WMC queries over a program's compiled WBF to the semantics from the previous
section. Finally we illustrate how we use BDDs to represent WBFs, which enables
the approach to automatically perform factorized inference.

\subsection{Compiling Boolean \dice{} Expressions}
\label{sec:comp_bool}
\begin{figure}
  \centering
  \begin{mdframed}
    \begin{align*}
     \frac{}
      {\true \rightsquigarrow (\true, \true, \emptyset)}~\textsc{(C-True)}
      \quad\qquad
      \frac{}
      {\false \rightsquigarrow (\false, \true, \emptyset)}~\textsc{(C-False)}
      \quad\qquad
      \frac{}
      { x \comp (\lx, \true, \emptyset)}~\textsc{(C-Ident)}
    \end{align*}
    \begin{align*}
      \frac{\text{fresh}~\lf}
      {\texttt{flip}~\theta \comp \Big(\lf, \true, (\lf \mapsto \theta, \true, \overline{\lf} \mapsto
      1-\theta)\Big)}~\textsc{(C-Flip)}
      \qquad\quad
      \frac{{\tt aexp} \comp (\varphi, \true, \emptyset)}
      {\Lobs{\tt aexp} \comp (\true, \varphi, \emptyset)}~\textsc{(C-Obs)}
    \end{align*}

  \infrule[C-Ite]{
    {\tt aexp} \comp (\varphi_g, \true, \emptyset) \andalso
    \te_T \comp (\varphi_T, \obs_T, w_T) \andalso
    \te_E \comp (\varphi_E, \obs_E, w_E)
  }{\Lite{\tt aexp}{\te_T}{\te_E} \comp
    \Big(\big((\varphi_g \land \varphi_{T}\big) \lor
    \big((\overline{\varphi}_g \land \varphi_{E}\big) ,
    \big((\varphi_g \land \obs_T\big) \lor
    \big((\overline{\varphi}_g \land \obs_E\big) ,
    w_T \cup w_E
    \Big)}
  
  \infrule[C-Let]{\te_1 \comp (\varphi_1, \obs_1, w_1) \andalso
    \te_2 \comp (\varphi_2, \obs_2, w_2)}
  {\Llet{x =\te_1}{\te_2}
    \comp 
 \big(\varphi_2[\lx \mapsto \varphi_1], \obs_1 \land \obs_2[\lx \mapsto \varphi_1],
    w_1 \cup w_2\big)}
\end{mdframed}
\caption{Compiling Boolean expressions to WBFs.}
\label{fig:comp_bool}
\end{figure}

The formal compilation judgment for Boolean \dice{} expressions has the form
$\te \rightsquigarrow (\varphi, \obs, w)$, where $\varphi$ and $\obs$ are
logical formulas and $w$ is a weight function (recall Definition~\ref{def:wbf}).
This judgment form will be extended later to accommodate other language
features. We call $\varphi$ the {\em unnormalized formula}: it represents all
possible assignments to variables and {\tt flip}s for which $\te$ evaluates to
true, ignoring observations. We call $\obs$ the {\em accepting formula}: it
represents all possible assignments to variables and {\tt flip}s that cause all
observations in $\te$ to succeed. Before showing the formal rules, we present
two examples to build intuition on the compilation to WBF and how it is used to
perform inference.

\begin{example}
  \label{ex:let}
The expression (\ref{eq:let}) from the previous section
compiles to the unnormalized formula $\varphi = f_1 \lor f_2$, where
$f_1$ and $f_2$ are Boolean variables associated with $\Lflip{0.1}$ and
$\Lflip{0.4}$ respectively.  Since there are no observations, $\obs = \true$ for this example.  The weight function $w$ assigns weights to the literals of $f_1$ and $f_2$ that
correspond with their probabilities in (\ref{eq:let}).  Then we have that $\dbracket{\ref{eq:let}}(\true) = \wmc(\varphi, w) = 0.46$ and $\dbracket{\ref{eq:let}}(\false) = \wmc(\overline{\varphi}, w) = 0.54$.
\end{example}

\begin{example}
The program (\ref{eq:obsprog}) from the previous section compiles to the unnormalized formula $\varphi = f_1$ and the accepting formula $\obs = f_1 \lor f_2$, where
$f_1$ corresponds with $\Lflip{0.6}$ and $f_2$ with $\Lflip{0.3}$. Hence the formula $\varphi \land \obs$ is true if and only if the
program evaluates to $\true$ and satisfies all observations, and similarly $\overline{\varphi}
\land \obs$ is true if and only if the program evaluates to $\false$ and satisfies all observations.  Then, with the appropriate weight function $w$, we can perform
Bayesian inference on (\ref{eq:obsprog}) via two weighted model counts:
$\dbracket{(\textsc{\ref{eq:obsprog}})}_D(\true) = \wmc(\varphi \land \obs, w) /
\wmc(\obs, w) \approx 0.83$ and $\dbracket{(\textsc{\ref{eq:obsprog}})}_D(\false) = \wmc(\overline{\varphi} \land \obs, w) /
\wmc(\obs, w) \approx 0.17$.
\end{example}

The formal compilation rules are shown in Figure~\ref{fig:comp_bool}. 
The above examples show how \emph{closed} programs are compiled, but expressions can also have free variables
in them.  The rule \textsc{C-Ident} handles a free variable $x$ simply by introducing a corresponding Boolean variable
$\lx$.   To illustrate the rule
\textsc{C-Flip}, $\Lflip{0.4} \rightsquigarrow (f,
\true, w)$ where $w$ maps $f$ to $0.4$ and $\bar{f}$ to $0.6$, and $f$ is a
fresh Boolean variable. Hence $\wmc(f \land \true, w) = 0.4 =
\dbracket{\Lflip{0.4}}(\true)$ and $\wmc(\bar{f}, w) = 0.6 =
\dbracket{\Lflip{0.4}}(\false)$.

The rule \textsc{C-Obs} handles \texttt{observe}s. Since an expression's
unnormalized formula ignores observations, the unnormalized formula for
$\Lobs{\tt aexp}$ is simply $\true$. The metavariable {\tt aexp} ranges over
values and identifiers and hence compiles to an accepting formula of $\true$ and
an empty weight function ({\tt aexp} stands for \emph {atomic expression}). Finally, the unnormalized formula of {\tt aexp}
becomes the accepting formula of $\Lobs{\tt aexp}$, in order to capture all ways
that the observation is satisfied.

The rule \textsc{C-Ite} encodes the usual logical semantics of conditionals.
Finally, the \textsc{C-Let} rule shows how to represent expression sequencing.
The {\em logical substitution} $\varphi_1[\lx \mapsto \varphi_2]$ replaces all
occurrences of $\lx$ in $\varphi_1$ with the formula $\varphi_2$. For the
accepting formula, the expression $\Llet{x=\te_1}{\te_2}$ only accepts if both
expressions accept, so we simply conjoin their accepting formulas. To illustrate
the rule, we show the derivation through the rules for our example
(\ref{eq:let}), assuming the obvious rule for compiling logical disjunction
(which is syntactic sugar for a conditional expression):

\infrule[ExLetCompilation]{
  \dfrac{\text{fresh}~f_1}{\Lflip{0.1} \rightsquigarrow (f_1, \true, w_1)}
  \andalso
  \dfrac{
    \dfrac{}{x \comp (\lx, \true, \emptyset)}
    \quad
    \dfrac{\text{fresh}~f_2}{\Lflip{0.4} \comp (f_2, \true, w_2)}
  }{\Lflip{0.4} \lor x \rightsquigarrow (f_2 \lor \lx, \true, w_2)}
}
{\Llet{x = \Lflip{0.1}}{\Lflip{0.4} \lor x} \comp (f_2 \lor \lx[\lx \mapsto
  f_1], \true, w_1 \cup w_2)}
This compilation matches Example~\ref{ex:let} above and shows how logical
substitution captures expression sequencing. The union of two weight functions,
denoted $w_1 \cup w_2$, is simply the union of the two maps $w_1$ and $w_2$;
this is well-defined because no two subexpressions can share \texttt{flip}s, so
there can be no conflicts.

The statement of correctness for Boolean \dice{} expressions connects our compilation rules to the formal semantics from the previous section:

\begin{lemma}[Boolean Expression Correctness]
  \label{thm:bool_correct}
  Let $\te$ be a Boolean \dice{} expression with free variables $x_1,\ldots,x_n$ and suppose $\te \rightsquigarrow (\varphi, \obs, w)$.  Then for any Boolean values $v_1,\ldots,v_n$:
\begin{itemize}[leftmargin=*]
\item $\dbracket{\te[x_i \mapsto v_i]}_A
= \wmc(\obs[\lx_i \mapsto v_i], w)$

\item for any Boolean value $v$, $\dbracket{\te[x_i \mapsto v_i]}_D(v)
= \frac{\wmc(((\varphi \Leftrightarrow v) \land \obs)[\lx_i \mapsto v_i], w)}{\wmc(\obs[\lx_i \mapsto v_i], w)}$.
\end{itemize}
\end{lemma}

As in the earlier definition of the distributional semantics, in the event that a division by zero occurs in the above lemma, the
result is defined to be zero. This lemma implies that we can answer inference queries on the
original expression via two WMC queries on the compiled WBF. The following key
lemma directly implies the one above; a stronger version of it is proven in
Appendix~\ref{sec:pf_expr_correct}:

\begin{lemma}
  \label{lem:bool_cor}
  Let $\te$ be a Boolean \dice{} expression with free variables $x_1,\ldots,x_n$
  and suppose $\te \rightsquigarrow (\varphi, \obs, w)$.  Then for any Boolean
  values $v_1,\ldots,v_n$ and Boolean value $v$, 
$$\dbracket{\te[x_i \mapsto v_i]}(v)
= \wmc(((\varphi \Leftrightarrow v) \land \obs)[\lx_i \mapsto v_i], w).$$
\end{lemma}

\subsection{Tuples \& Typed Compilation}
\label{sec:comp_tup}

Next we extend our compilation rules to support arbitrarily nested tuples. The
primary purpose of tuples is to empower \dice{} functions by enabling multiple
arguments and return values. Intuitively, this involves generalizing the
compilation target from a single Boolean formula $\varphi$ to tuples of Boolean
formulas.  Formally, this
extension requires that we generalize the compilation judgment, which now has
the following form:
\begin{align*}
\Gamma \vdash \te : \tau \comp (\xpointphi, \obs, w).
\end{align*}
First, our compilation is now {\em typed}: $\Gamma$ is the usual type
environment for free variables and $\tau$ is the type of $\te$. The types are
necessary to determine how to properly encode program variables in the compiled
logical formulas. Second, compilation produces a collection of Boolean formulas,
one per occurrence of the type $\mathbf{Bool}$ in $\tau$. The new metavariable
$\xpointphi$ is defined inductively as either a Boolean formula $\varphi$ or a
pair of the form $(\xpointphi_1, \xpointphi_2)$.

As a concrete example of compiling a program that contains tuples:
\begin{align*}
   \{\} \vdash \Llet{x=\Lflip{0.2}}{(x, \true)} : \bool \times \bool \rightsquigarrow \Big((f_1, \true), \true, [f_1 \mapsto 0.2, \bar{f}_1 \mapsto 0.8]\Big).
\end{align*}
Here, the resulting compiled formula $\xpointphi$ is a pair of Boolean formulas
$(f_1, \true)$. %

\begin{figure}
  \centering
  \begin{mdframed}
 \begin{align*}
 \frac{\Gamma(x) = \tau}
      {\Gamma \vdash x : \tau \comp (\form_\tau(x), \true, \emptyset)}~\textsc{(C-Ident)} \quad
  \frac{\Gamma(x_1) = \tau_1 \quad \Gamma(x_2) = \tau_2}
   {\Gamma \vdash (x_1, x_2) : \tau_1 \times \tau_2 \comp
   ((\form_{\tau_1}(x_1), \form_{\tau_2}(x_2)), \true, \emptyset)}~\textsc{(C-Tup)}
 \end{align*}
 \begin{align*}
 \frac{\Gamma(x) = \tau_1 \times \tau_2}
  {\Gamma \vdash \Lfst{x} : \tau_1 \rightsquigarrow (\form_{\tau_1}(x_l), \true,\emptyset)}~\textsc{(C-Fst)}
  \quad
 \frac{\Gamma(x) = \tau_1 \times \tau_2}
  {\Gamma \vdash \Lsnd{x} : \tau_2 \rightsquigarrow (\form_{\tau_2}(x_r), \true,\emptyset)}~\textsc{(C-Snd)}
\end{align*}
    \infrule[C-Ite]{
      \Gamma \vdash {\tt aexp} : \bool \comp (\varphi_g, \true, \emptyset) \andalso
    \Gamma \vdash \te_T : \tau \comp (\xpointphi_T, \obs_T, w_T) \andalso
    \Gamma \vdash \te_E : \tau \comp (\xpointphi_E, \obs_E, w_E)
  }{\Gamma \vdash \Lite{\tt aexp}{\te_T}{\te_E} : \tau \comp 
    \Big(\big((\varphi_g \xbroadand{\tau} \xpointphi_{T}\big) \xpointor{\tau}
    \big((\overline{\varphi}_g \xbroadand{\tau} \xpointphi_{E}\big),
    \big((\varphi_g \land \obs_T\big) \lor
    \big((\overline{\varphi}_g \land \obs_E\big) ,
    w_T \cup w_E
    \Big)
  }

  \infrule[C-Let]{\Gamma \vdash \te_1 : \tau_1 \comp (\xpointphi_1, \obs_1, w_1) \andalso
    \Gamma \cup \{x : \tau_1\} \vdash \te_2 : \tau_2 \comp (\xpointphi_2, \obs_2, w_2)}
  {\Gamma \vdash \Llet{x:\tau_1 =\te_1}{\te_2} : \tau_2
    \comp 
 \big(\xpointphi_2[\lx \xmapsto{\tau} \xpointphi_1], \obs_1 \land \obs_2[\lx \xmapsto{\tau} \xpointphi_1],
    w_1 \cup w_2\big)}
\end{mdframed}
\caption{Typed compilation for tuples. These assume, without loss of generality
but for simplicity, that \texttt{fst}, \texttt{snd}, and tuple construction are
only ever performed with identifiers as arguments.}
\label{fig:comp_typed}
\end{figure}

Figure~\ref{fig:comp_typed} shows the new rules for compiling tuples and also
presents updated versions of the rules from Figure~\ref{fig:comp_bool}, other
than the Boolean-specific rules.
The extended compilation for tuples is structurally very similar to Boolean
compilation, but requires generalizing the Boolean operations in a natural way
to accommodate tuples (Appendix~\ref{app:notation} summarizes this new notation).
The new version of \textsc{C-Ident} uses the
\emph{form function}
$\form_\tau(x)$, which constructs the logical representation of a variable $x$ based on its type $\tau$.  It is defined inductively as $\form_\bool(x) \defeq \lx$ and $\form_{\tau_1
\times \tau_2}(x) \defeq (\form_{\tau_1}(x_l), \form_{\tau_2}(x_r))$. Note the
subscripts $x_l$ and $x_r$ that lexically distinguish the left and right
elements. This function also allows us to naturally define the compilation for
tuple creation as well as {\tt fst} and {\tt snd} in Figure~\ref{fig:comp_typed}.

The \textsc{C-Ite} rule shows how we generalize the compilation of conditionals to accommodate tuples.  The rule requires that we conjoin a Boolean expression $\varphi_g$
(the compiled guard) with a potential tuple of formulas (the compiled then and else
branches). To do this, we define \emph{broadcasted conjunction}, denoted
$\varphi_g \xbroadand{\tau} \xpointphi$, as conjoining $\varphi_g$ with all the
Boolean expressions in the tuple $\xpointphi$. Formally, we define it as $\varphi_a
\xbroadand{\bool} \varphi_b \defeq \varphi_a \land \varphi_b$ and $\varphi_a
\xbroadand{\tau_1 \times \tau_2} (\xpointphi_{b1}, \xpointphi_{b2}) \defeq \big(\varphi_a
\xbroadand{\tau_1} \xpointphi_{b1},~ \varphi_a \xbroadand{\tau_2} \xpointphi_{b2}\big)$.
In addition to broadcasted conjunction, \textsc{C-Ite} also requires
\emph{point-wise disjunction}, denoted $\xpointphi_1 \xpointor{\tau} \xpointphi_2$.
Point-wise disjunction is defined inductively as 
$\varphi_1 \xpointor{\bool} \varphi_2 \defeq \varphi_1 \lor \varphi_2$ and
 $(\xpointphi_{11}, \xpointphi_{12}) \xpointor{\tau_1 \times \tau_2}
      (\xpointphi_{21}, \xpointphi_{22}) \defeq (\xpointphi_{11} \xpointor{\tau_1}
      \xpointphi_{21}, \xpointphi_{12} \xpointor{\tau_2} \xpointphi_{22})$.

Finally, to generalize the compilation of {\tt let} expressions, in the
\textsc{C-Let} rule we employ \emph{typed substitution} $\xpointphi_2[\lx
\xmapsto{\tau_1} \xpointphi_1]$ to substitute the compiled version of $\te_1$ into
the compiled version of $\te_2$. We define typed substitution inductively as
follows:
\begin{align*}
  \varphi_2[\lx \xmapsto{\bool} \varphi_1] \defeq \varphi_2[\lx \mapsto
\varphi_1],
\quad
\varphi_2[\lx \xmapsto{\tau_{a} \times \tau_{b}} (\xpointphi_a,
\xpointphi_b)] \defeq \varphi_2[\lx_l \xmapsto{\tau_{a}} \xpointphi_a][\lx_r \xmapsto{\tau_{b}}
\xpointphi_b],
\end{align*}
\begin{align*}
(\xpointphi_1, \xpointphi_2)[\lx \xmapsto{\tau} \xpointphi] \defeq 
(\xpointphi_1[\lx \xmapsto{\tau} \xpointphi], \xpointphi_2[\lx \xmapsto{\tau} \xpointphi]).
\end{align*}

We can state and prove a natural generalization of our key lemma from the previous subsection, Lemma~\ref{lem:bool_cor}.  The lemma
depends on \emph{pointwise iff}, denoted $\xpointphi_1 \xLeftrightarrow{\tau}
\xpointphi_2$ and defined inductively as follows:  $\varphi_1 \xLeftrightarrow{\bool} \varphi_2
\defeq \varphi_1 \Leftrightarrow \varphi_2$ and ${(\xpointphi_1,
\xpointphi_2)}\xLeftrightarrow{\tau_1 \times \tau_2}{(\xpointphi'_1, \xpointphi'_2)} \defeq \left(
{\xpointphi_1}\xLeftrightarrow{\tau_1}{\xpointphi_1'} \right) \land \left(
{\xpointphi_2}\xLeftrightarrow{\tau_2}{\xpointphi_2'} \right)$.
Then, the following correctness lemma is proved in Appendix~\ref{sec:pf_expr_correct}:
\begin{lemma}[Typed Correctness Without Functions]
  \label{lem:correct_no_proc}
 Let $\te$ be a \dice{} expression without function calls, and suppose $\{x_i :
\tau_i\} \vdash \te : \tau \rightsquigarrow (\xpointphi, \obs, w)$. Then for any
values $\{v_i : \tau_i\}$ and $v : \tau$, we have that $\dbracket{\te[x_i
\mapsto v_i]}(v) = \wmc\left(\big(({\xpointphi}\xLeftrightarrow{\tau} {v}) \land
\obs\big)[\lx_i \xmapsto{\tau_i} v_i], w\right)$.
\end{lemma}

\subsection{Functions \& Programs}
\label{sec:comp_proc}
\begin{figure}
  \centering
  \begin{mdframed}
\begin{align*}
 \frac{\Gamma \cup \{x_1: \tau_1\}, \Phi \vdash \te : \tau_2 \comp (\xpointphi, \obs, w)}
{ \Gamma, \Phi \vdash \texttt{fun}~f(x_1 : \tau_1) : \tau_2 ~ \{\te\}
  \comp (\xpointphi, \obs, w)}~\textsc{(C-Func)}
  \quad 
  \frac{\Gamma, \Phi \vdash \te : \tau \comp (\xpointphi, \obs, w)}
  {\Gamma, \Phi \vdash \bullet ~ \te : \tau \comp (\xpointphi, \obs, w)}~\textsc{(C-Prog1)}
\end{align*}
\infrule[C-Prog2]{
  \Gamma, \Phi \vdash \texttt{fun}~f(x_1: \tau_1) : \tau_2 ~ \{\te\}
  \comp (\xpointphi_f, \obs_f, w_f)\\
  \Gamma \cup \{f \mapsto \tau_1 \rightarrow \tau_2\}, \Phi \cup \{f \mapsto (\lx_1, \xpointphi_f, \obs_f, w_f)\}
  \vdash \texttt{p} : \tau \comp (\xpointphi, \obs, w)}
{\Gamma, \Phi \vdash \texttt{fun}~f(x_1 : \tau_1) : \tau_2 ~ \{\te\} ~\texttt{p} : \tau \comp (\xpointphi, \obs, w)}

\infrule[C-FuncCall]{\Gamma(f) = \tau_1 \rightarrow \tau_2 \andalso
  \Gamma(x_1) = \tau_1 \\ 
  \Phi(f) = (\lx_{arg}, \xpointphi, \obs, w) \andalso
(\xpointphi', \obs', w') = \texttt{RefreshFlips}(\lx_{arg}, \xpointphi, \obs, w)}
{\Gamma, \Phi \vdash f(x_1) : \tau_2 \rightsquigarrow (\xpointphi'[\lx_{arg} \xmapsto{\tau_1} \lx_1], \obs'[\lx_{arg}
  \xmapsto{\tau_1} \lx_1], w')}
  \end{mdframed}
  \caption{Compiling functions and programs. These assume without loss of
    generality but for simplicity that function calls are only ever given
    identifiers as arguments.}
  \label{fig:comp_func}
\end{figure}

We conclude our development of \dice{} compilation by introducing functions and
programs in Figure~\ref{fig:comp_func}. We introduce a new piece of
context $\Phi$ into our judgment, which maps function names to their compiled
function bodies. Function names are mapped to a 4-tuple
$(\lx_{arg}, \xpointphi, \obs, w)$ where $\lx_{arg}$ is the logical variable for
the function's formal argument and the other items are respectively the function
body's compiled unnormalized formula, accepting formula, and weight function.

The judgment $\Gamma,\Phi \vdash {\tt func} \rightsquigarrow (\xpointphi, \obs, w)$ compiles function definitions.  As shown in \textsc{C-Func}, we simply compile the function's body in an appropriate type environment.
The judgment $\Gamma,\Phi \vdash {\tt p} : \tau \rightsquigarrow (\xpointphi, \obs, w)$
compiles programs by compiling each function in order, followed by the ``main''
expression. The rules \textsc{C-Prog1} and \textsc{C-Prog2} perform this
compilation. After each function is compiled, its compiled WBF is added to
$\Phi$ and its name and type are added to $\Gamma$, for use in subsequent compilation.

The final judgment form for expressions is $\Gamma,\Phi \vdash \te : \tau
\rightsquigarrow (\xpointphi, \obs, w)$, and \textsc{C-FuncCall} shows the rule for
compiling function calls. The rule simply looks up the function's compiled WBF
and substitutes the actual argument for the formal argument. One subtlety is
that we must ensure that the {\tt flip}s in each call to a function are
independent of one another. Our compilation approach makes it straightforward to
do so: simply replace all of the variables in $\xpointphi$ and $\obs$, aside from
the formal argument $\lx_{arg}$, with fresh variables. We use an auxiliary
function $\texttt{RefreshFlips}(\lx_{arg}, \xpointphi, \obs, w)$ for this purpose.
We now state the full correctness theorem for \dice{} compilation:

\begin{theorem}[Compilation Correctness]
  \label{thm:compilation_correct}
  Let $\prog$ be a \dice{} program and $ \emptyset, \emptyset \vdash \prog : \tau \rightsquigarrow
(\xpointphi, \obs, w)$.  Then:
(1) $\dbracket{\prog}_A
= \wmc(\obs, w)$, 
and (2) for any value $v : \tau$, $\dbracket{\prog}_D(v)
= \wmc((\xpointphi \xLeftrightarrow{\tau} v) \land \obs, w) / \wmc(\obs, w)$.
\end{theorem}

As before, division by zero is defined to be zero, and we prove this theorem as a corollary of the
following stronger property, which is proven in Appendix~\ref{app:prog_correct}:
\begin{theorem}[Typed Program Correctness]
  \label{thm:prog_correctness}
  Let $\prog$ be a \dice{} program $\emptyset, \emptyset \vdash \prog : \tau \rightsquigarrow
(\xpointphi, \obs, w)$. Then for any $v : \tau$, we have that $\dbracket{\prog}(v)
= \wmc((\xpointphi \xLeftrightarrow{\tau} v) \land \obs, w)$.
\end{theorem}

\subsection{Binary Decision Diagrams as WBF}
\label{sec:bdd_comp}
Weighted model counting on WBFs is still \#P-hard, so our compilation above is
not necessarily advantageous. Now we reap the benefits of this translation by
representing WBF with binary decision diagrams (BDDs), a data structure that
facilitates efficient inference by exploiting the program structure to minimize
the size of the WBF. A BDD is a popular data structure for representing Boolean
formulas, and there is a rich literature of using BDDs to represent the
state space of non-probabilistic programs during model
checking~\citep{Clarke2000, Jhala2009}.

\begin{figure}
  \centering
 \begin{align*}
\dfrac{
  \dfrac{\text{fresh}~f_1}{\Lflip{0.1} \rightsquigarrow
    \left(
    \scalebox{0.7}{
    \begin{tikzpicture}[baseline=-2ex]
    \def\lvl{14pt}
    \node (f1) at (0, 0) [bddnode] {$f_1$};
    \node (true) at ($(f1) + (-20bp, -\lvl)$) [bddterminal] {$\true$};
    \node (false) at ($(f1) + (20bp, -\lvl)$) [bddterminal] {$\false$};
    \begin{scope}[on background layer]
      \draw [highedge] (f1) -- (true);
      \draw [lowedge] (f1) -- (false);
    \end{scope}
    \end{tikzpicture}}~,
    \scalebox{0.7}{
    \begin{tikzpicture}
      \node (true) [bddterminal] {$\true$};
    \end{tikzpicture}}~, w_1\right)}
  \andalso
  \dfrac{
    \dfrac{}{x \comp
      \left(
    \scalebox{0.7}{
    \begin{tikzpicture}[baseline=-2ex]
    \def\lvl{14pt}
    \node (f1) at (0, 0) [bddnode] {$\lx$};
    \node (true) at ($(f1) + (-20bp, -\lvl)$) [bddterminal] {$\true$};
    \node (false) at ($(f1) + (20bp, -\lvl)$) [bddterminal] {$\false$};
    \begin{scope}[on background layer]
      \draw [highedge] (f1) -- (true);
      \draw [lowedge] (f1) -- (false);
    \end{scope}
  \end{tikzpicture}}~,
    \scalebox{0.7}{
    \begin{tikzpicture}
      \node (true) [bddterminal] {$\true$};
    \end{tikzpicture}}~, \emptyset\right)}
    \quad
  \dfrac{\text{fresh}~f_2}{\Lflip{0.4} \comp
    \left(
    \scalebox{0.7}{
    \begin{tikzpicture}[baseline=-2ex]
    \def\lvl{14pt}
    \node (f1) at (0, 0) [bddnode] {$f_2$};
    \node (true) at ($(f1) + (-20bp, -\lvl)$) [bddterminal] {$\true$};
    \node (false) at ($(f1) + (20bp, -\lvl)$) [bddterminal] {$\false$};
    \begin{scope}[on background layer]
      \draw [highedge] (f1) -- (true);
      \draw [lowedge] (f1) -- (false);
    \end{scope}
    \end{tikzpicture}}~,
    \scalebox{0.7}{
    \begin{tikzpicture}
      \node (true) [bddterminal] {$\true$};
    \end{tikzpicture}}~, w_2\right)}
  }{\Lflip{0.4} \lor x \comp
      \left(
    \scalebox{0.7}{
    \begin{tikzpicture}[baseline=-4ex]
    \def\lvl{14pt}
    \node (f2) at (0, 0) [bddnode] {$\lx$};
    \node (x) at ($(f2) + (-20bp, -\lvl)$) [bddnode] {$f_2$};
    \node (false) at ($(x) + (-20bp, -\lvl)$) [bddterminal] {$\false$};
    \node (true) at ($(x) + (20bp, -\lvl)$) [bddterminal] {$\true$};
    \begin{scope}[on background layer]
      \draw [lowedge] (f2) -- (x);
      \draw [highedge] (f2) -- (true);
      \draw[lowedge] (x) -- (false);
      \draw[highedge] (x) -- (true);
    \end{scope}
    \end{tikzpicture}}~,
    \scalebox{0.7}{
    \begin{tikzpicture}
      \node (true) [bddterminal] {$\true$};
    \end{tikzpicture}}~, w_2\right)}
  }{\Llet{x = \Lflip{0.1}}{\Lflip{0.4} \lor x} \comp
      \left(
    \scalebox{0.7}{
    \begin{tikzpicture}[baseline=-4ex]
    \def\lvl{14pt}
    \node (f2) at (0, 0) [bddnode] {$f_1$};
    \node (x) at ($(f2) + (-20bp, -\lvl)$) [bddnode] {$f_2$};
    \node (false) at ($(x) + (-20bp, -\lvl)$) [bddterminal] {$\false$};
    \node (true) at ($(x) + (20bp, -\lvl)$) [bddterminal] {$\true$};
    \begin{scope}[on background layer]
      \draw [lowedge] (f2) -- (x);
      \draw [highedge] (f2) -- (true);
      \draw[lowedge] (x) -- (false);
      \draw[highedge] (x) -- (true);
    \end{scope}
    \end{tikzpicture}}~,
    \scalebox{0.7}{
    \begin{tikzpicture}
      \node (true) [bddterminal] {$\true$};
    \end{tikzpicture}}~, w_1 \cup w_2\right)}
\end{align*}
  \caption{A BDD derivation tree for \textsc{(ExLetCompilation)}.}
  \label{fig:bdd_deriv}
\end{figure}

The compilation rules in the previous subsections were deliberately designed to
facilitate BDD compilation. Consider the example compilation
(\textsc{ExLetCompilation}) from Section~\ref{sec:comp_bool}. Each step in this
derivation can be translated into a corresponding BDD operation, as illustrated
by the \emph{BDD derivation tree} in Figure~\ref{fig:bdd_deriv}. The final BDD
is compiled compositionally, at each step exploiting program structure to
produce a minimal, canonical representation (for the given variable ordering).
The operations necessary for constructing this derivation tree --- BDD
conjunction, disjunction, and substitution --- are all standard operations that
are available in BDD packages such as \texttt{CUDD}~\citep{Somenzi2010}. 

The
cost of \dice{} inference is dominated by the cost of constructing the
corresponding BDD derivation tree: that step is computationally hard in general, while WMC on the final BDD is linear time in the size of the BDD.
However, BDDs can exploit program structure in order to allow compilation to scale efficiently on many examples.
The remainder of this paper is devoted to showing that the BDD can
be efficient to construct for useful programs. In Section~\ref{sec:experiments}
we show this experimentally, and Section~\ref{sec:analysis} characterizes the
hardness of \dice{} inference.

\section{\dice{} Implementation \& Empirical Evaluation}
\label{sec:experiments}

Now we describe our implementation and empirical evaluation of \dice{}. \dice{} is
implemented in \texttt{OCaml} and uses \texttt{CUDD} as its backend for
compiling BDDs~\citep{Somenzi2010}. First we 
describe extensions to the core \dice{} syntax that make programming
more ergonomic and enable us to more easily implement some of the benchmark
programs. Then we describe our empirical evaluation
of \dice{}'s performance in comparison with prior PPLs on a suite of benchmarks. In
Section~\ref{sec:analysis} we give context to these experiments and discuss why
\dice{} succeeds on many benchmarks where others fail.

\subsection{\dice{} Extensions \& Ergonomics}
\label{sec:extensions}
Our implementation extends 
the core \dice{} syntax from Figure~\ref{fig:grammar} in several ways. We relax the constraint on A-normal form, allowing more arbitrary placement of expressions. We also include syntactic sugar for the usual Boolean operators 
 $\land, \lor$ and $\neg$.  Finally, we include support for bounded integers and bounded iteration, both of which are described in more detail next.  
Further details of our implementation can be found in Appendix~\ref{sec:impl_details}.

\subsubsection{Bounded Integers}

\dice{} supports probability distributions over integers with the \texttt{discrete} keyword: for
instance, the expression \texttt{discrete(0.1, 0.4, 0.5)} defines a discrete
distribution over $\{0,1,2\}$ where $0$ has probability $0.1$, $1$ has
probability $0.4$, and $2$ has probability $0.5$. 
There are a number of possible strategies for encoding integers into a WBF. The simplest --- and the one we implemented --- is a
\emph{one-hot encoding}.  Specifically, a distribution over $n$ integers is represented as tuple of $n$ Boolean variables, each representing one integer value, and {\tt flip}s are used to ensure that each variable is true with the specified probability.  For example, here is the encoding of our example distribution above:
\begin{align*}
  \texttt{discrete(0.1, 0.4. 0.5)} \rightsquigarrow
  \Bigg\{~
    \begin{aligned}
    &\texttt{let v0 = flip(0.1) in}\\
    &\texttt{let v1 = $\neg$v0 $\land$ flip(0.4/(0.4 + 0.5)) in} \\
    &\texttt{let v2 = $\neg$v0 $\land$ $\neg$v1 in (v0, (v1, v2))}
  \end{aligned}
\end{align*}
Formally, for a discrete distribution \texttt{discrete}$(\theta_1, \theta_2,
\cdots, \theta_n)$, the encoded value \texttt{v}$_i$ is true only if (1)
$\bigwedge_{k < i} \neg\texttt{v}_k$ holds and (2) a coin flipped with
probability $\theta_i / \sum_{j \ge i} \theta_j$ is true. \dice{} also supports
the standard modular arithmetic operations like $(+)$ and $(\times)$ on integers.

\subsubsection{Statically Bounded Iteration}
Iteration and loops are challenging program constructs to support in PPLs.
\dice{}, like many other PPLs, supports \emph{bounded iteration}: loops that
always terminate after a finite number of
iterations~\citep{cusumano2018incremental, gehr2016psi, Claret2013, Pfeffer2007,
goodman2014design}. It does so via the syntax \texttt{iterate(f, init, k)}, where 
\texttt{f} is a function name, {\tt init} is an initialization expression, and {\tt k} is an integer indicating the number of times to call {\tt f}:
\begin{align*}
  \texttt{iterate(f, init, k)}
  \rightsquigarrow \underbrace{f(f(\cdots f}_{k \text{ times}}(\texttt{init}))).
\end{align*}
Many useful examples --- such as the network reachability example from
Section~\ref{sec:motivation} --- can be expressed as bounded iteration.

\subsection{Empirical Performance Evaluation}
We have faithfully implemented the compilation strategy and use of BDDs as
described in Section~\ref{sec:compilation}. Section~\ref{sec:motivation}
highlights some program structure that BDD compilation exploits, and
Section~\ref{sec:analysis} explores this structure further, but the question
remains: does this structure exist in practice, and can \dice{} effectively
exploit it? We investigate these questions from three angles:
\begin{description}
\item[Q1: Comparison with Existing PPLs] How quickly can \dice{} perform exact
  inference on benchmark probabilistic programs from the literature? We evaluate
  this question in Section~\ref{sec:exp_baselines}.
\item[Q2: Exploiting Functions] What are the performance benefits of modular
  compilation for functions? We evaluate this question in
  Section~\ref{sec:exp_procedure} by comparing \dice{}'s performance with and
  without inlining function calls.
\item[Q3: Comparison with Bayesian Network Solvers] Discrete Bayesian networks are a
  special case of \dice{} programs and are a good source of challenging and
  realistic inference problems. A natural question here is: how does \dice{}
  compare against state-of-the-art Bayesian network solvers that are specialized
  for this class of programs? In Section~\ref{sec:bn} we compare \dice{}
  against \texttt{Ace}~\citep{Chavira2008}, a state-of-the-art discrete Bayesian
  network solver.
\end{description}

In our evaluation we compare \dice{} against state-of-the-art
PPLs that employ two different classes of exact inference algorithms:

\begin{description}
\item[Algebraic Methods] The first class are \emph{algebraic} inference methods
that represent the probability distribution as a symbolic expression or algebraic decision
diagram (ADD)~\cite{gehr2016psi, Claret2013, dehnert2017storm,
narayanan2016probabilistic}. We discuss this class of inference algorithms more
thoroughly in Section~\ref{sec:algebraic}. In this class, we compare
experimentally against \textsc{Psi}~\citep{gehr2016psi}.\footnote{We used
\textsc{Psi} version \texttt{2d21f9fe04cf3aac533e08ccc2df18179947baad}}
\item[Enumerative Methods] The second class of inference methods work by
exhaustively \emph{enumerating} all paths through the probabilistic program,
possibly using dynamic programming to reduce the search space~\citep{
wingate2013automated, Sankaranarayanan2013, Albarghouthi2017_1,
goodman2014design, Chistikov2015, filieri2013reliability,
geldenhuys2012probabilistic}. Both \textsc{Psi} and
\textsc{WebPPL}~\citep{goodman2014design} have a mode that supports dynamic-programming
exact inference, and we compare against them experimentally.
\end{description}

Comparing the performance of probabilistic program inference is challenging
because performance is closely tied to the intricacies of how the program is
structured: semantically equivalent programs may have vastly differing
performance. Throughout our experiments we made a best-effort attempt at
representing the programs in a way that was maximally performant in each
language. The tables in this section report the mean value over at least 5 runs
for each experiment; the appendix contains expanded tables that
include standard deviations for each result. All experiments were single-threaded
and performed on the same server with a 2.66GHz CPU and 512GB of RAM. The
timings were recorded using
\texttt{hyperfine},\footnote{\url{https://github.com/sharkdp/hyperfine}} a
utility that performs statistical timing analysis of Unix shell commands.

\subsubsection{Baselines}
\label{sec:exp_baselines}

\begin{table*}

  \centering
  \begin{tabular}{lrrr|rr}
    \toprule
    Benchmark
    & {Psi (ms)}
    & DP (ms)
    & {\dice{} (ms)}
    & {\# Paths}
    & {BDD Size}
    \\
    \midrule
    Grass & 167 & 57 & \bfseries 14 & 95 & 15\\
    Burglar Alarm & 98 & 10 & \bfseries 13 & 250 & 11\\
    Coin Bias & 94 & 23 & \bfseries 13 & 4 & 13\\
    Noisy Or & 81 & 152 & \bfseries 13 & 1640 & 35 \\
    Evidence1 & 48 & 32 & \bfseries 13 & 9 &  5\\
    Evidence2 & 59 & 28 & \bfseries 13 & 9 &  6 \\
    Murder Mystery & 193 & 75 & \bfseries 10 & 16 &  6 \\
    \bottomrule
  \end{tabular}
  \caption{\emph{Baselines.} Comparison of inference algorithms (times are
milliseconds). The total time for \dice{} is reported under the ``\dice{}''
column, and the total size of the final compiled BDD is reported in the ``BDD
Size'' column. }
  \label{tab:experiments}
\end{table*}

Table~\ref{tab:experiments} summarizes the results of our performance
experiments on well-known baselines, 
which includes all of the discrete programs that
\textsc{Psi} and \textsc{R2} were evaluated on~\citep{gehr2016psi, nori2014r2,
 borgstrom2011measure}.\footnote{See the
appendix for the source code for all the benchmarks used in this
section. One discrete benchmark, ``Digit Recognition'', was omitted
from this table because the version of Psi that we tested with does not support
this program.}
Each row is a different benchmark.
The
``Psi'', ``DP'', and ``\dice{}'' columns give the amount of time (in
milliseconds) for respectively (1) Psi's default inference
algorithm~\citep{gehr2016psi}, (2) Psi's dynamic programming inference
algorithm that is specialized for finite discrete programs, and (3) the total
time for \dice{} to compile a BDD and perform weighted model counting.
These examples are small and thus relatively easy for exact
inference, but they serve as an important sanity check. Generally these examples
are too trivial to differentiate the performance of \dice{} and Psi.

We include two other columns, ``\# Paths'' and ``BDD Size'', that give a proxy
for how hard each inference problem is. The ``\# Paths'' column gives how many
paths would be explored by a path enumeration algorithm. The ``BDD Size''
gives the final compiled BDD generated by \dice{}, which in conjunction with the
``\# Paths'' column gives a metric for how much structure \dice{} is exploiting.

\subsubsection{Modular Compilation}
\label{sec:exp_procedure}
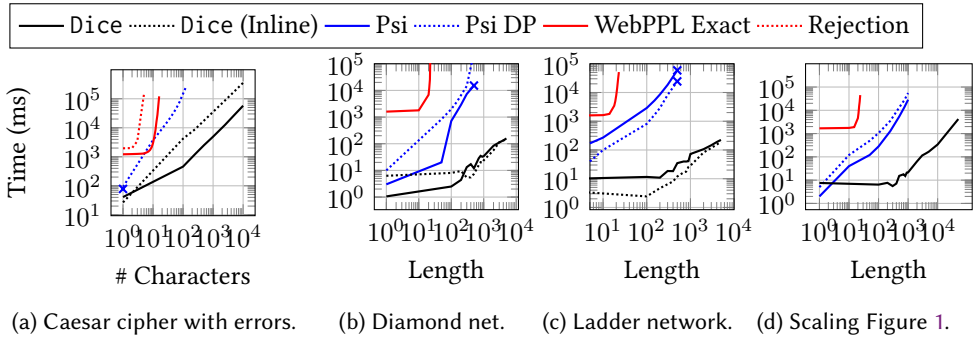
\begin{figure*}
  \centering
  \begin{subfigure}[b]{0.3\linewidth}
      \begin{tikzpicture}
    \begin{axis}[
		height=3.5cm,
    ymin=0,
    ymax=1000000,
		grid=major,
    width=3.5cm,
    xmode=log,
    ymode=log,
    xtick={1,10,100,1000,10000},
    ytick={1,10,100,1000,10000,100000},
    xlabel={\# Characters},
    ylabel={Time (ms)},
    legend columns=6,
    legend style={at={(13em,9em)},anchor=north}
    ]
    \addplot[mark=none, thick] table [x index={0}, y index={1}] {\symencerror}; %
    \addlegendentry{\dice{}};
    \addplot[mark=none, densely dotted, thick] table [x index={0}, y index={2}] {\symencerror}; %
    \addlegendentry{\dice{} (Inline)};
    \addplot[thick, blue] coordinates {
      (1,78)
	  } node[draw, fill=blue, cross out, scale=0.4] {}; %
    \addlegendentry{Psi};
    \addplot[mark=none, blue, thick, densely dotted] table [x index={0}, y index={1}] {\psiencerror}; %
    \addlegendentry{Psi DP};
    \addplot[mark=none, thick, red] table [x index={0}, y index={1}] {\webpplencerror}; %
    \addlegendentry{WebPPL Exact};
    \addplot[mark=none, thick, red, densely dotted] table [x index={0}, y index={2}] {\webpplencerror}; %
    \addlegendentry{Rejection};
  \end{axis}
\end{tikzpicture}
\caption{Caesar cipher with errors.}
\label{fig:exp_caesar_error}
\end{subfigure}~
\begin{subfigure}[b]{0.2\linewidth}
  \centering
  \begin{tikzpicture}
    \begin{axis}[
		height=3.5cm,
		width=3.5cm,
		grid=major,
    xmode=log,
    ymode=log,
    xtick={1,10,100,1000,10000},
    ytick={1,10,100,1000,10000,100000},
    xlabel={Length},
    ymin=0,
    ymax=100000
    ]
    \addplot[mark=none, thick, blue, densely dotted] table [x index={0}, y index={1}] {\psidiamond}; %
    \addplot[thick, blue] coordinates {
      (500,15200)
	  } node[draw, fill=blue, cross out, scale=0.4] {};
    \addplot[mark=none, thick, blue] table [x index={0}, y index={2}] {\psidiamond}; %
    \addplot[mark=none, thick] table [x index={0}, y index={1}] {\syminfdiamond}; %
    \addplot[mark=none, thick, densely dotted] table [x index={0}, y index={2}] {\syminfdiamond}; %
    \addplot[mark=none, thick, red] table [x index={0}, y index={1}] {\webppldiamond}; %
    \end{axis}
  \end{tikzpicture}
  \caption{Diamond net.}
  \label{fig:diamondnetscale}
\end{subfigure}~
\begin{subfigure}[b]{0.2\linewidth}
  \centering
  \begin{tikzpicture}
    \begin{axis}[
		height=3.5cm,
		width=3.5cm,
		grid=major,
    xmode=log,
    ymode=log,
    xtick={1,10,100,1000,10000},
    ytick={1,10,100,1000,10000,100000},
    xmin=5,
    xlabel={Length},
    ymin=0,
    ymax=100000,
    ]
    \addplot[mark=none, thick, blue] table [x index={0}, y index={1}] {\psiinfnetworking}; %
    \addplot[thick, blue] coordinates {
      (500,24340)
	  } node[draw, fill=blue, cross out, scale=0.4] {};
    \addplot[mark=none, thick, blue, densely dotted] table [x index={0}, y index={2}] {\psiinfnetworking}; %
    \addplot[thick, blue] coordinates {
      (500,58810)
	  } node[draw, fill=blue, cross out, scale=0.4] {};
    \addplot[mark=none, thick] table [x index={0}, y index={1}] {\syminfladder}; %
    \addplot[mark=none, thick, densely dotted] table [x index={0}, y index={2}] {\syminfladder}; %
    \addplot[mark=none, thick, red] table [x index={0}, y index={1}] {\webpplladder}; %
    \end{axis}
  \end{tikzpicture}
  \caption{Ladder network.}
  \label{fig:laddernetscale}
\end{subfigure}~
\begin{subfigure}[b]{0.2\linewidth}
  \centering
      \begin{tikzpicture}
    \begin{axis}[
		height=3.5cm,
    width=3.8cm,
    ymin=0,
    ymax=1000000,
		grid=major,
    xmode=log,
    ymode=log,
    xtick={1,10,100,1000,10000},
    ytick={1,10,100,1000,10000,100000},
    xlabel={Length},
    ]
    \addplot[mark=none, blue, densely dotted, thick] table [x index={0}, y index={1}] {\psimc};
    \addplot[mark=none, blue, thick] table [x index={0}, y index={2}] {\psimc};
    \addplot[mark=none, red, thick] table [x index={0}, y index={1}] {\webmc};
    \addplot[mark=none, thick] table [x index={0}, y index={1}] {\dicemc};
  \end{axis}
\end{tikzpicture}
\caption{Scaling Figure~\ref{fig:motiv_init}.}
\label{fig:motiv_scale}
\end{subfigure}
\caption{Log-log scaling plots illustrating the benefits of separate
compilation of functions. An ``x''-mark denotes a runtime error was encountered
at that point. The time reported for \dice{} inference includes the time
required to compile and perform WMC. The
standard deviation for the run-times are negligible.}
  \label{fig:scaling}
\end{figure*}
We return to the motivating examples from Section~\ref{sec:motivation} to see
how \dice{} compares with existing methods, and against a version of itself
where all function calls are inlined. Figure~\ref{fig:scaling} shows how
different algorithms scale as the size of the problem grows (note that all plots
are in log-log scale). Figure~\ref{fig:motiv_scale} was introduced and discussed
in Section~\ref{sec:motivation}.

\paragraph{Encryption}
Figure~\ref{fig:caesar} introduced the Caesar cipher motivating example, and
Figure~\ref{fig:exp_caesar_error} shows how exact inference on this example
scales as the number of characters being encrypted increases. \dice{} is about
an order of magnitude faster than the case when function calls are inlined, and
multiple orders of magnitude faster than WebPPL and Psi. In particular, Psi's
default algebraic inference fails to handle the encryption of even a single
character; we explore why in Section~\ref{sec:algebraic}.

Approximate inference approaches generally struggle with these kinds of
programs, due to the low probability of finding samples that satisfy the
observations. To illustrate this, we also report the time it took for rejection
sampling to draw 10 accepted samples. WebPPL supports rejection sampling, and
Figure~\ref{fig:exp_caesar_error} shows how it scales for this particular
example program. This figure shows that rejection sampling scales exponentially
in this case, and thus is not a feasible route around the state-space explosion
problem.

\paragraph{Network Reachability}
Next we examine how separate compilation helps in the network reachability task
described in Figure~\ref{fig:ex_net}. Figure~\ref{fig:diamondnetscale} shows how
exact inference scales in the number of diamond subnetworks. We see a modest
benefit over inlining: compiling the {\tt diamond} function multiple times is
not very expensive since it is so small. Note that modular function compilation
is not strictly beneficial: for this example, the inlined version is faster than
the modular version after about $10^2$ iterations. Also note that both versions
of \dice{} are multiple orders of magnitude faster than \textsc{Psi} and
\textsc{WebPPL} due to the exponential number of paths.

We expect to see overall linear scaling of \dice{} for many network
topologies due to conditional independence. To evaluate this, Figure~\ref{fig:laddernetscale}
shows a version where instead of diamonds we use a \emph{ladder network} of the
following structure:
     \begin{tikzpicture}[node distance=0.2cm, baseline=-1em]
      \node[] (init1) {$\dots$};
      \node[below = of init1] (init2) {$\dots$};
      \node[draw, circle, right=of init1] (S1) {};
      \node[draw, circle, right=of init2] (S2) {};
      \node[draw, circle, right=of S1] (S3) {};
      \node[draw, circle, right=of S2] (S4) {};
      \node[right=of S3] (final1) {$\dots$};
      \node[right=of S4] (final2) {$\dots$};
    
      \draw[->] (init1) -- (S1);
      \draw[->] (init2) -- (S2);

      \draw[->] (S1) -- (S3);
      \draw[->] (S1) -- (S4);

      \draw[->] (S2) -- (S3);
      \draw[->] (S2) -- (S4);

      \draw[->] (S3) -- (final1);
      \draw[->] (S4) -- (final2);
    \end{tikzpicture}.
The goal is to determine the probability of a packet reaching the end of a
network that consists of a chain of ladder subnetworks where each has a similar
probabilistic routing policy to the diamond network. \dice{} continues to scale
well, while this example is challenging for the other methods, in part since the
number of paths is exponential in the length of the network.

\subsubsection{Discrete Bayesian Networks}
\label{sec:bn}
There is currently a lack of challenging discrete probabilistic program
benchmarks in the literature. To more rigorously establish the relative
performance of \dice{} and existing algorithms, here we evaluate the performance
of \dice{} on discrete Bayesian networks that we translated into equivalent Psi and
\dice{} programs. These benchmarks were selected from the Bayesian Network
Repository, an online repository of well-known Bayesian
networks.\footnote{\url{https://www.bnlearn.com/bnrepository/}} These programs
are (1) \emph{realistic}: each has been used to answer scientific research
questions in various domains such as medical diagnosis, weather modeling, and
insurance modeling; and (2) \emph{challenging}: many of these examples have on
the order of thousands or tens of thousands of random variables.

First, we will compare the performance of \dice{} and Psi on this task; then we
compare \dice{} against a specialized Bayesian network tool. We will show that
\dice{} significantly outperforms Psi on all of these examples and is
competitive with the specialized Bayesian network solver.

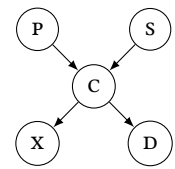
\begin{wrapfigure}{R}{0.2\linewidth}
  \centering
       \begin{tikzpicture}[node distance=0.75cm, baseline=-1em]
      \node[draw, circle] (C) {\textsc{c}};
      \node[draw, circle, right of=C, above of=C] (S) {\textsc{s}};
      \node[draw, circle, left of=C, above of=C] (P) {\textsc{p}};
      \node[draw, circle, left of=C, below of=C] (X) {\textsc{x}};
      \node[draw, circle, right of=C, below of=C] (D) {\textsc{d}};

      \draw[->] (P) -- (C);
      \draw[->] (S) -- (C);
      \draw[->] (C) -- (X);
      \draw[->] (C) -- (D);
    \end{tikzpicture}
    \caption{The ``Cancer'' Bayesian network.}
    \label{fig:cancer}
\end{wrapfigure}

\paragraph{Comparison with Psi}
\label{sec:singlemarg}
\begin{table*}

  \centering
  \begin{tabular}{l
    rrr
    |r
    S[table-format=2.1e2]
    S[table-format=2.1e2]}
    \toprule
    Benchmark
    & {Psi (ms)}
    & DP (ms)
    & {\dice{} (ms)}
    & {\# Parameters}
    & {\# Paths}
    & {BDD Size}
    \\
    \midrule
    Cancer & 772 & 46 & \bfseries 13 & 10 & 1.1e3 & 28 \\
    Survey & 2477 & 152 & \bfseries 13 & 21 & 1.3e4 & 73 \\
    Alarm & \xmark & \xmark & \bfseries 25 & 509 & 1.0 e 36 & 1.3e3 \\
    Insurance & \xmark & \xmark & \bfseries 212 & 984  & 1.2 e 40 & 1.0e5 \\
    Hepar2 & \xmark & \xmark & \bfseries 54 & 48 & 2.9e69 &  1.3e3 \\
    Hailfinder & \xmark & \xmark & \bfseries 618 & 2656  & 2.0e76 & 6.5e4 \\
    Pigs & \xmark & \xmark & \bfseries 72 & 5618 & 7.3e492 & 35\\
    Water & \xmark & \xmark & \bfseries 2590 & $1.0\times 10^4$ & 3.2e54 & 5.1e4\\
    Munin & \xmark & \xmark & \bfseries 1866 & $8.1 \times 10^5$  & 2.1e1622 & 1.1e4\\
    \bottomrule
  \end{tabular}
  \caption{\emph{Single Marginal Inference}. Comparison of inference algorithms (times are milliseconds). A
``\xmark'' denotes a timeout at 2 hours of running. The total time for \dice{}
is reported under the ``\dice{}'' column, and the total size of the final
compiled BDD is reported in the ``BDD Size'' column.}
  \label{tab:bnexperiments}
\end{table*}
Table~\ref{tab:bnexperiments} compares Psi against \dice{} on the task of
computing a \emph{single marginal} of a leaf node of a Bayesian network, a
standard Bayesian network query. As an example of this task,
Figure~\ref{fig:cancer} shows the ``Cancer'' Bayesian
network~\citep{korb2010bayesian}, a simple 5-node network for modeling the
probability that a patient has cancer (the \textcircled{c} node) given a
collection of symptoms (\textcircled{\textsc{x}} and \textcircled{\textsc{d}})
and causes (\textcircled{\textsc{p}} and \textcircled{\textsc{s}}). The
single-marginal task for this example is to compute the marginal probability
of the leaf node $\Pr(\textcircled{\textsc{x}})$.

Table~\ref{tab:bnexperiments} compares the performance of \dice{} and Psi on the
single-marginal inference task for a variety of Bayesian networks. The size of
the network --- a proxy for the difficulty of the inference task --- is
given by the number of parameters (the ``\# Parameters'' column in the table).
Psi fails to complete the inference within the allotted two hours on any of the
medium or larger sized Bayesian networks.

\paragraph{Comparison with a Bayesian Network Solver}
\begin{table}
  \centering
  \sisetup{detect-weight=true,detect-inline-weight=math}
  \begin{tabular}{l
    r
    r
    |S[table-format=2.1e1]}
    \toprule
    Benchmark
    & {\dice{} (ms)}
    & {\texttt{Ace} (ms)}
    & {BDD Size}
    \\
    \midrule
    Alarm & \bfseries 159 & 422 & 4.3e5\\
    Hailfinder & 1280 & \bfseries 522 & 2.1e5 \\
    Insurance & \bfseries 222 & 492 & 2.3e5\\
    Hepar2& \bfseries 163 & 495 & 5.4e5\\
    Pigs & 11243 & \bfseries 985 & 2.6e5\\
    Water & 3320  & \bfseries 605 & 6.8e4\\
    Munin & 4021194 & \bfseries 3500 & 2.2e7\\
    \bottomrule
  \end{tabular}
    \caption{\emph{All marginals}. A comparison between \dice{} and \texttt{Ace}
    on the all-marginal discrete Bayesian network inference task.}
    \label{tbl:full-joint}
\end{table}
As a final test of the \dice{}'s performance, in Table~\ref{tbl:full-joint} we
compare against \texttt{Ace}, a state-of-the-art Bayesian network
solver~\citep{Chavira2008}. The task here is to compute \emph{all marginal
probabilities}, a strictly harder task than the single-marginal task considered
earlier. We note that Psi fails to complete even a single marginal inference
task on any of these examples within 2 hours, so it is omitted from this table.

Part of what makes the all-marginals inference task challenging is that it
requires the computation of many queries: one for each node in the Bayesian
network. One of the benefits of our compilation is that a single (potentially
expensive) compilation, once completed, can be efficiently reused to perform
many marginal probability queries. We highlight this capability in
Table~\ref{tbl:full-joint}, which shows the cost of compiling the full joint
distribution of the example discrete Bayesian networks.
  These compilations take on the order of
several seconds; however, once compiled, computing each marginal probability ---
or any other query with a small BDD, such as disjoining together several
variables --- takes milliseconds. For comparison, Psi cannot compute a single
marginal on any of these examples within two hours.

\texttt{Ace}, similar to \dice{}, reduces the Bayesian network probabilistic
inference task to weighted model counting (with a very different encoding
scheme). This gives \texttt{Ace} an inherent advantage over \dice{} on this
task: \texttt{Ace} does not support arbitrary program constructs --- such as
conditional branching, procedures, and \texttt{observe} statements --- and hence can
specialize directly for Bayesian networks, a limited subclass of \dice{}
programs.

Despite these inherent advantages, Table~\ref{tbl:full-joint} shows that \dice{} is
competitive with \texttt{Ace} on a number of challenging Bayesian network
inference tasks. \texttt{Ace} significantly outperforms \dice{} only on the very
largest network, ``Munin''. These results suggest that even though \dice{} is a
general-purpose PPL, it is still a competitive
exact inference algorithm for medium-sized Bayesian networks. %

\section{Discussion \& Analysis}
\label{sec:analysis}

The previous section demonstrates empirically that \dice{} can perform exact
inference orders of magnitude faster than existing inference algorithms on a range of benchmarks.  In this section we provide discussion and analysis that provide context for these results.
First we ask in Section~\ref{sec:hardness}: how hard is exact inference in
\dice{}? We show that inference is \textsf{PSPACE}-hard, which means that it is likely harder than inference on discrete Bayesian networks.  This begs the
question: why do the experiments in Section~\ref{sec:experiments} succeed at
all? We explore this question in Section~\ref{sec:inference_fast} by identifying different forms of program structure that \dice{} exploits in order to scale.
Finally, Section~\ref{sec:algebraic} considers algebraic representations as an
alternative compilation target for probabilistic programs and discusses the forms of structure that they
are and are not capable of exploiting.

\subsection{Computational Hardness of Exact \dice{} Inference}
\label{sec:hardness}
The experiments in Section~\ref{sec:experiments} raise a natural question: how
hard is the exact inference challenge for \dice{} programs? The complexity of
exact inference has been well-studied in the context of discrete Bayesian
networks. In particular, the decision problem of determining whether or not the
probability of an event in a Bayesian network exceeds a certain threshold is
\textsf{PP}-complete~\citep{kwisthout2009computational,
littman1998computational}. The canonical \textsf{PP}-complete problem is
\textsc{MajSat}, the problem of deciding whether or not the majority of truth
assignments satisfy a logical formula. It is clear that exact \dice{} is
\textsf{PP}-hard: indeed, some of our experiments in
Section~\ref{sec:experiments} utilize a polynomial-time reduction from discrete
Bayesian networks to \dice{} programs. However, in fact exact inference for
\dice{} is \textsf{PSPACE}-hard, and therefore likely harder than discrete
Bayesian network inference as \textsf{PP$\subseteq$PSPACE}:

\begin{theorem}
  Exact inference in \dice{} is \textsf{PSPACE}-hard.
  \label{thm:hard}
\end{theorem}
A proof sketch is in Appendix~\ref{sec:hard}. This result depends on the
expressiveness of functions, which Bayesian networks lack. We leave for future
work the investigation of tighter complexity bounds for \dice{} inference.

\subsection{When is \dice{} Inference Fast?}
\label{sec:inference_fast}
\dice{} inference, in the worst case, is extremely hard. Why, then, do the
experiments in Section~\ref{sec:experiments} succeed? Put another way: when can
we guarantee that the BDD derivation tree is efficient to construct (i.e.,
polynomial in the size of the program)? In this section we explore two sources
of tractability in \dice{} inference, both of which are structural properties
that a programmer can consciously exploit while designing \dice{} programs. The
first source of structure is \emph{independence}, which implies the existence of
factorizations. The second is a more subtle property called \emph{local
structure} that implies that, even in some cases without independence, it can
still be efficient to construct the BDD derivation
tree~\citep{boutilier1996context, chavira2005compiling}.
These forms of structure were first introduced in the context of
graphical models for capturing conditional probability tables with various forms
of structure. We show that these insights can be generalized to \dice{} programs.

\subsubsection{Independence}
\label{sec:indep}
The independence property implies that two program parts communicate only over a
limited interface. It is the key reason why \dice{} performs so well in many of
the benchmarks (Section~\ref{sec:exp_baselines}). Programs naturally have
conditional independence, implied by their control flow, function boundaries,
etc. In the motivating example in Figure~\ref{fig:cond_bdd}, variable \texttt{z}
does not depend on \texttt{x} given an assignment to \texttt{y}. This is
commonly called \emph{conditional independence} of \texttt{x} and \texttt{z}
given \texttt{y}, and it partially explains why \dice{} scales to thousands of
conditionally independent layers in Figure~\ref{fig:motiv_scale}.

\dice{} naturally exploits conditional independence. We can formalize this by
giving bounds on the cost of composing BDDs that are conditionally independent.
In general, the operation $B_1 \land B_2$ on two BDDs $B_1$ and $B_2$ has time
and space complexity $\bigO(|B_1| \times |B_2|)$, and similarly for $B_1 \lor
B_2$~\citep{meinel1998algorithms}. This implies a worst-case exponential blowup
as BDDs are composed. However, \dice{} can exploit conditional independence ---
among other properties --- to avoid this exponential blowup in practice:

\begin{proposition}
Let $B_1$ and $B_2$ be BDDs that share no variables other than some variable
$z$, and let $|B|$ be the size of the BDD $B$. Then we say $B_1$ and $B_2$ are
\emph{conditionally independent given $z$}, and computing $B_1 \land B_2$ and
$B_1 \lor B_2$ has time and space complexity $\bigO(|B_1| + |B_2|)$ for a
variable order that orders the variables in $B_1$ before $z$ and $z$ before the
variables in $B_2$.
\label{prop:composition}
\end{proposition}
A proof sketch can be found in Appendix~\ref{sec:comp_proof}.
Proposition~\ref{prop:composition} implies that compositional rules that utilize
conjunction and disjunction to compose \dice{} programs --- like \textsc{C-Let}
--- can be efficient in the presence of conditional independence.
One useful source of conditional independence is function calls: they
are conditionally independent from all other expressions given their arguments
and return value. The motivating example in Figure~\ref{fig:ex_net} illustrates
an example of this form of conditional independence. Each call to the
\texttt{diamond} procedure is independent of all prior calls given only the
immediately previous call. It follows that the size of the BDD for the example
in Figure~\ref{fig:ex_net_d} grows as $\bigO(|\texttt{diamond}| \times c)$,
where $c$ is the number of calls to the \texttt{diamond} procedure and
$|\texttt{diamond}|$ is the size of the compiled BDD for the procedure.

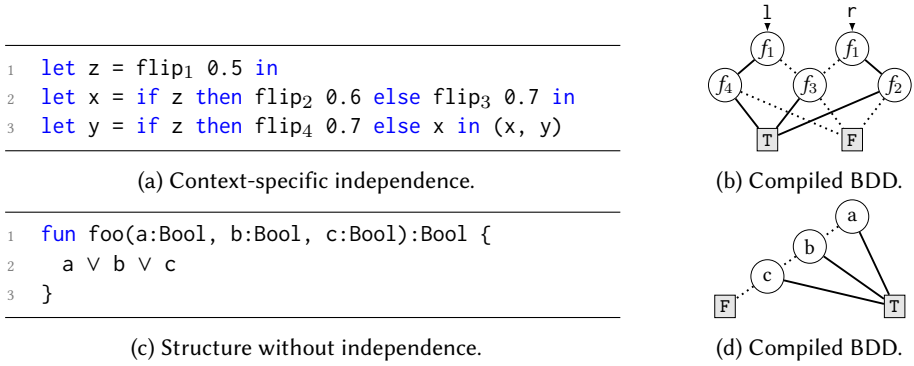
\begin{figure}
\begin{subfigure}[b]{0.6\linewidth}
 \begin{lstlisting}[mathescape=true, basicstyle=\ttfamily\small]
let z = flip$_1$ 0.5 in
let x = if z then flip$_2$ 0.6 else flip$_3$ 0.7 in
let y = if z then flip$_4$ 0.7 else x in (x, y)
\end{lstlisting}
     \caption{Context-specific independence.}
     \label{fig:sens_ex}
   \end{subfigure}~
   \begin{subfigure}[b]{0.35\linewidth}
     \centering
     \scalebox{0.8}{
    \begin{tikzpicture}
      \def\lvl{17pt}
    \node (f1a) at (0, 0) [bddnode] {$f_1$};
    \node (f1b) at ($(f1a) + (40bp, 0)$) [bddnode] {$f_1$};

      \node (x) at ($(f1a) + (0, \lvl)$)  {\texttt{l}};
      \node (y) at ($(f1b) + (0, \lvl)$)  {\texttt{r}};
    \node (f4) at ($(f1a) + (-20bp, -\lvl)$) [bddnode] {$f_4$};
    \node (f3) at ($(f1a) + (20bp, -\lvl)$) [bddnode] {$f_3$};
    \node (f2) at ($(f1b) + (20bp, -\lvl)$) [bddnode] {$f_2$};

    \node (true) at ($(f1a) + (0bp, -2.5*\lvl)$) [bddterminal] {$\true$};
    \node (false) at ($(f1b) + (0bp, -2.5*\lvl)$) [bddterminal] {$\false$};

    \begin{scope}[on background layer]
      \draw [highedge] (f1a) -- (f4);
      \draw [lowedge] (f1a) -- (f3);
      \draw [lowedge] (f1b) -- (f3);
      \draw [highedge] (f1b) -- (f2);

      \draw [highedge] (f4) -- (true);
      \draw [highedge] (f3) -- (true);
      \draw [highedge] (f2) -- (true);
      \draw [lowedge] (f4) -- (false);
      \draw [lowedge] (f3) -- (false);
      \draw [lowedge] (f2) -- (false);
      
      \draw[->] (x) -- (f1a);
      \draw[->] (y) -- (f1b);
    \end{scope}
    \end{tikzpicture}
    }

    \caption{Compiled BDD.}
    \label{fig:sens_bdd}
 \end{subfigure}\\
\begin{subfigure}[b]{0.6\linewidth}
 \begin{lstlisting}[mathescape=true, basicstyle=\ttfamily\small]
fun foo(a:Bool, b:Bool, c:Bool):Bool {
  a $\lor$ b $\lor$ c
}
\end{lstlisting}
     \caption{Structure without independence.}
     \label{fig:xor_ex}
   \end{subfigure}~
   \begin{subfigure}[b]{0.35\linewidth}
     \centering
     \scalebox{0.8}{
    \begin{tikzpicture}
      \def\lvl{14pt}
    \node (f1) at (0, 0) [bddnode] {a};
    \node (f2a) at ($(f1) + (-20bp, -\lvl)$) [bddnode] {b};
    \node (f3a) at ($(f2a) + (-20bp, -\lvl)$) [bddnode] {c};
    \node (true) at ($(f3a) + (60bp, -\lvl)$) [bddterminal] {$\true$};
    \node (false) at ($(f3a) + (-20bp, -\lvl)$) [bddterminal] {$\false$};
    \begin{scope}[on background layer]

      \draw [highedge] (f1) -- (true);
      \draw [lowedge] (f1) -- (f2a);
      \draw [highedge] (f2a) -- (true);
      \draw [lowedge] (f2a) -- (f3a);
      \draw [highedge] (f3a) -- (true);
      \draw [lowedge] (f3a) -- (false);
    \end{scope}
    \end{tikzpicture}
    }

    \caption{Compiled BDD.}
    \label{fig:xor_bdd}
 \end{subfigure}

   \caption{\dice{} programs and their compiled BDDs illustrating different
degrees of structure. }
   \label{fig:context_specific}
\end{figure}

\dice{} exploits another, more fine-grained form of independence called {\em
context-specific independence}. Historically, context-specific independence has
led to significant speedups in graphical model
inference~\citep{boutilier1996context}. We briefly sketch its benefits here. Two BDDs $B_1$ and $B_2$ are
\emph{contextually independent given $z=v$}, for some variable $z$ and value
$v$, if $B_1[z \mapsto v]$ and $B_2[z \mapsto v]$ share no
variables~\cite{boutilier1996context}. As for conditional independence,
composing contextually independent BDDs can often be efficient.

An example program that exhibits context-specific independence is shown in
Figure~\ref{fig:sens_ex}. The variables $x$ and $y$ are correlated if $z =
\false$ or if $z$ is unknown, but they are independent if $z= \true$. Thus, $x$
is independent of $y$ given $z = \true$. Figure~\ref{fig:sens_bdd} shows how our
compilation strategy exploits this independence. Since the program evaluates to
a tuple, it is compiled to a tuple of two BDDs. However, in our implementation
these BDDs share nodes wherever possible, so they can be equivalently viewed as
a single, \emph{multi-rooted BDD}. The left and right element of the tuple are represented
by the  \texttt{l} and \texttt{r} roots respectively. The program's context-specific independence implies that there will
be no shared sub-BDD between \texttt{l} and \texttt{r} if $f_1$ is true. We
refer to \citet{boutilier1996context} for more on the performance benefits of exploiting
context-specific independence in probabilistic graphical models.

\subsubsection{Local Structure}
Finally, it is possible for the BDD compilation process to be efficient even in
the absence of independence if the program has structure that is amenable to
efficient BDD compilation. \citet{chavira2005compiling} showed that exploiting
local structure led to significant speedups in Bayesian network inference, and
this performance was one of the primary motivations for developing \texttt{Ace}.
Local structure is a broad category of structural properties that can make
performance more efficient, including determinism, context-specific
independence, and other properties~\citep{boutilier1996context,
gogate2011samplesearch, sang2005performing, Chavira2008}.

At its core, local structure is a property that makes compiling a BDD more
efficient than naively using a conditional probability table to represent a
probability distribution. Figure~\ref{fig:xor_ex} gives an example \dice{}
function that computes the disjunction of three arguments.
Figure~\ref{fig:xor_bdd} shows the compiled BDD for this function. It is compact
and hence exploiting the program structure. Note that, if the number of
variables disjoined together were to increase, the size of the BDD --- and the
cost of compiling it --- would increase only linearly with the number of
variables. This stands in stark contrast to an approach to inference that is
agnostic to local structure (such as simple variable elimination), which would
not identify that this or-function is a compact way of representing the
distribution.

\dice{} implicitly exploits local structure during
inference. For instance, the Bayesian network ``Hepar2'' has many examples of
determinism, sparse probability tables, and context-specific independence;
\dice{} exploits these properties to be competitive with the performance
of \texttt{Ace} on this example and others in Table~\ref{tbl:full-joint}.

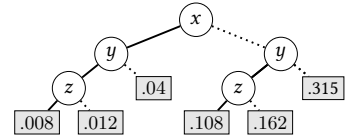
\begin{wrapfigure}{r}{0.3\linewidth}
  \centering
       \scalebox{0.8}{
    \begin{tikzpicture}
      \def\lvl{16pt}
    \node (f1) at (0, 0) [bddnode] {$x$};
    \node (f2) at ($(f1) + (-40bp, -\lvl)$) [bddnode] {$y$};
    \node (f2b) at ($(f1) + (40bp, -\lvl)$) [bddnode] {$y$};
    \node (f3) at ($(f2) + (-20bp, -\lvl)$) [bddnode] {$z$};
    \node (tft) at ($(f2) + (20bp, -\lvl)$) [bddterminal] {$.04$};
    \node (f3c) at ($(f2b) + (-20bp, -\lvl)$) [bddnode] {$z$};
    \node (f3d) at ($(f2b) + (20bp, -\lvl)$) [bddterminal] {.315};

    \node (ttt) at ($(f3) + (-15bp, -\lvl)$) [bddterminal] {$.008$};
    \node (ttf) at ($(f3) + (15bp, -\lvl)$) [bddterminal] {$.012$};
    \node (ftt) at ($(f3c) + (-15bp, -\lvl)$) [bddterminal] {$.108$};
    \node (ftf) at ($(f3c) + (15bp, -\lvl)$) [bddterminal] {$.162$};

    \begin{scope}[on background layer]
      \draw [highedge] (f1) -- (f2);
      \draw [lowedge] (f1) -- (f2b);
      \draw [highedge] (f2) -- (f3);
      \draw [highedge] (f2b) -- (f3c);
      \draw [lowedge] (f2b) -- (f3d);
      \draw [highedge] (f2b) -- (f3c);
      \draw [lowedge] (f2b) -- (f3d);

      \draw [highedge] (f3) -- (ttt);
      \draw [lowedge] (f3) -- (ttf);
      \draw [lowedge] (f2) -- (tft);
      \draw [highedge] (f3c) -- (ftt);
      \draw [lowedge] (f3c) -- (ftf);
    \end{scope}
    \end{tikzpicture}
    }
    \caption{An ADD representation of the distribution in Equation~\ref{eq:exhaustive}.}
    \label{fig:add}
  \end{wrapfigure}
\subsection{Algebraic Representations}
\label{sec:algebraic}
Previous sections have shown that BDDs naturally capture and exploit
factorization and procedure reuse. While these are common and useful
program properties, they are not the only possible ones, and different
compilation targets will naturally exploit others. In this section we consider
\emph{algebraic compilation targets} as a foil to our approach, to highlight their relative
strengths and weaknesses.

In contrast to our WMC approach that explicitly separates
the logical representation from probabilities, algebraic approaches integrate
probabilities directly into the compilation target. A common algebraic target
are \emph{algebraic decision diagrams} (ADDs)~\citep{bahar1997algebric}, which
are similar to binary decision diagrams except that they have numeric values as
leaves.  This makes them a natural choice for compactly encoding
probability distributions in the probabilistic
programming and probabilistic model checking communities, with different
encoding strategies from \dice{}~\citep{Claret2013,
dehnert2017storm, Kwiatkowska2011}. As an example, Figure~\ref{fig:add}
shows an ADD for the program in
Figure~\ref{fig:cond_ex} if it returned a tuple of $x$, $y$, and $z$. ADDs encode probabilities of total
assignments of variables: in this example, a probability of $0.008$ is given to the assignment $x =
y = z = \true$.

ADDs have several similarities with BDDs. First, they support composition
operations and so can offer a compositional compilation target~\citep{Claret2013}, albeit very different from the one
described by our compilation rules. Second, they support
efficient inference once the ADD is constructed.
Despite these similarities, ADDs have strikingly different scaling properties
from BDDs because they exploit different underlying structure of the program.
The key difference is that BDDs are agnostic to the \texttt{flip} parameters:
they naturally exploit logical program structure such as independence and local structure in order
to scale without needing to know what any probabilities are. As the previous
subsections have argued, BDDs excel at this task. In contrast, ADDs naturally exploit \emph{global repetitious
probabilities}: repeated probabilities of possible worlds in the entire
distribution. This is shown in Figure~\ref{fig:add}, which collapses states
with the same probability --- for example, if $x = y = \false,$ then the ADD
terminates with a node that does not depend on $z$'s value:
\begin{tikzpicture}
    \node (ftf) [bddterminal] {$.315$};
  \end{tikzpicture}.

Global repetitious probabilities are an orthogonal property to independence.
ADDs do not exploit independence in the same way as \dice{}. ADDs must explicitly
represent the probability of each total instantiation of the variables of interest, corresponding to each possible value of the returned tuple. In our example, this means that
the ADD cannot exploit the conditional independence of $z$ and $x$ given $y$, and instead needs to enumerate their joint probabilities.

Hence, unlike \dice{}'s BDD representation, the size of a compiled ADD is
sensitive to the precise parameters chosen for \texttt{flip}s in the program.
If these parameters are chosen such that the probability of each total
assignment is distinct, and we are interested in a tuple of all the random variables, then the number of leaves in the ADD will equal the
number of of paths in the probabilistic program. As shown in
Table~\ref{tab:experiments}, this can be prohibitively large for many examples;
the BDD size is typically many orders of magnitude smaller than the
number of paths on these real-world programs.

\section{Related Work}
\label{sec:related_work}
There is a large literature on probabilistic programming languages and inference
algorithms. At a high level, \dice{} is distinguished from existing PPLs by being the first to use weighted model counting to perform exact inference for a PPL that includes traditional programming language constructs, functions, and first-class observations.  In this section we survey the existing literature on
probabilistic program inference and provide context for how each relates to
\dice{}.

\paragraph{Path-based inference algorithms}
The most common class of probabilistic program inference algorithms today are
\emph{operational}, meaning that they work by executing the probabilistic
program on concrete values. Common examples include sampling
algorithms~\citep{carpenter2016stan, Hur2015, pfeffer2007general,
chaganty2013efficiently, wood2014new, van2015particle,
mansinghka2013approximate, goodman:uai08, saad2016, Mansinghka2018} and
variational approximations~\citep{bingham2019pyro, dillon2017tensorflow,
wingate2013automated, kucukelbir2015automatic, InferNET14}. Other approaches use
symbolic techniques to perform inference but are similar in spirit, in the sense that they
separately enumerate paths through the program~\citep{Sankaranarayanan2013,
Albarghouthi2017_1, geldenhuys2012probabilistic, filieri2013reliability}.
These approaches do not factorize the program: they consider entire execution
paths as a whole.
\citet{Chistikov2015} proposes performing \emph{weighted model integration} --- a
generalization of weighted model counting to the continuous
domain~\citep{BelleIJCAI15, zeng2020efficient, dos2019exact} --- to perform
inference by integrating along paths through a probabilistic program. 

Additionally, sampling and variational algorithms are distinguished from our
approach by being approximate rather than exact inference algorithms. In
general, these techniques can be applied to both discrete and continous
distributions, though they often rely on program continuity or differentiation
to be effective~\citep{carpenter2016stan, hoffman2014no, gram2018hamiltonian,
wingate2013automated, kucukelbir2015automatic, InferNET14}. In contrast to all
of these approaches, \dice{} performs factorized, exact inference on non-smooth,
non-differentiable, discrete programs.

\paragraph{Algebraic inference algorithms}
A number of PPL inference algorithms work by translating the probabilistic
program into an algebraic expression that encodes its probability distribution,
and then using symbolic algebra tools in order to manipulate that expression and
perform probabilistic inference. Examples include
Psi~\citep{gehr2016psi}, Hakaru~\citep{narayanan2016probabilistic}, and
approaches that employ algebraic decision diagrams~\citep{Claret2013,
dehnert2017storm}. Algebraic representations exploit fundamentally different
program structure from our approach based on weighted model counting; see
Section~\ref{sec:algebraic} for a discussion.

\paragraph{Graphical model compilation}
There exists a large number of PPLs that perform inference by converting the
program into a probabilistic graphical model~\cite{ pfeffer2009figaro,
McCallum2009, InferNET14, bornholt2014uncertain}. These compilation strategies
are limited by the semantics of graphical models: key program structure --- such
as functions, conditional branching, \emph{etc}. --- is usually lost during compilation
and so cannot be exploited during inference. Further, graphical
models can express conditional independence via the graphical structure, but
typical inference algorithms such as variable elimination cannot exploit more
subtle, context-specific forms of independence that our
approach exploits, as shown in Section~\ref{sec:indep}~\citep{Darwiche09}.

\paragraph{Probabilistic Logic Programs}
Closest to our approach are techniques for exact inference in probabilistic
logic programs~\cite{de2007problog, Riguzzi2011TPLP, Fierens2015,
Vlasselaer2015}. Similar to our work, these techniques reduce probabilistic
inference to weighted model counting and employ representations that support
efficient WMC, such as BDDs~\cite{Bryant86} or sentential decision
diagrams~\cite{darwiche2011sdd}. Unlike that work, \dice{} supports traditional
programming language constructs, including functions, and it supports
first-class observations rather than only observations at the very end of the
program. We show how to exploit functional abstraction for modular compilation,
and first-class observations require us to explicitly account for an {\em
accepting} probability in both the semantics and the compilation strategy.

\paragraph{Programmer-Guided Inference Decomposition}
Several PPLs provide a sublanguage that allows the programmer to provide information that can be used to decompose program inference into multiple separate parts~
\citep{pfeffer2018structured,Mansinghka2018,holtzen2018sound}. Hence the goal is similar in spirit to our goal of automated program factorization. These
approaches are complementary to ours: \dice{} automatically finds and exploits
program factorizations and local structure, while these approaches can perform sophisticated decompositions through explicit programmer guidance.

\paragraph{Static Analysis \& Model Checking}
Forms of symbolic model checking often represent the reachable state space of a
program as a BDD~\citep{Jhala2009, BiereModelChecking}. Our work can be thought
of as enriching this representation with probabilities: we track the possible
assignments to each \texttt{flip} and the accepting formula in order to do exact
Bayesian inference via WMC. Static analysis techniques have
also been generalized to analyze probabilistic programs. For example,
probabilistic abstract interpretation~\citep{Cousot2012} provides a general
framework for static analysis of probabilistic programs. However, these
techniques seek to acquire lower or upper bounds on probabilities, while we
target exact inference. Probabilistic model checking (PMC) is a mature
generalization of traditional model checking with multiple high-quality
implementations~\citep{dehnert2017storm, Kwiatkowska2011}. The goal of PMC is
typically to verify that a system meets a given probabilistic temporal logic
formula. They can also be used to perform probablistic inference, but they
have not used weighted model counting for inference and instead typically rely
on ADDs, which gives them different scaling properties than \dice{} as we discussed
earlier. 
\citet{vazquez2020_mce} recently described an approach to
learn Boolean task specifications on Markov decision processes. This work
shares some core technical machinery with our approach but differs markedly
in its goals and encoding strategy.

\section{Conclusion}
\label{sec:conclusion}
We presented a new approach to exact inference for discrete probabilistic
programs and implement it in the \dice{} probabilistic programming language. We (1) showed how to reduce exact
inference for \dice{} to weighted model counting, (2) proved this translation
correct, (3) demonstrated the performance of this inference strategy over
existing methods, and (4) characterized the efficiency of compiling \dice{} in
key scenarios.

In the future we hope to extend \dice{} in several ways. First, we believe that
the insights of \dice{} can be cleanly integrated into many existing
probabilistic programming systems, even those with approximate inference that
can handle continuous random variables. We see this as
an exciting avenue for extending the reach of approximate inference algorithms,
which currently struggle with discreteness. Second, we believe that \dice{} can
be extended to handle more powerful data structures and programming constructs,
notably forms of unbounded loops and recursion. And finally, we hope to further explore
the landscape of weighted model counting approaches.

\section*{Acknowledgments}
This work is partially supported by NSF grants \#IIS-1943641, \#IIS-1956441,
\#CCF-1837129, DARPA grant \#N66001-17-2-4032, a Sloan Fellowship, and gifts by
Intel and Facebook research. The authors would like to thank Jon Aytac and
Philip Johnson-Freyd for feedback on paper drafts.

\bibliography{bib}


\begin{thebibliography}{86}


\ifx \showCODEN    \undefined \def \showCODEN     #1{\unskip}     \fi
\ifx \showDOI      \undefined \def \showDOI       #1{#1}\fi
\ifx \showISBNx    \undefined \def \showISBNx     #1{\unskip}     \fi
\ifx \showISBNxiii \undefined \def \showISBNxiii  #1{\unskip}     \fi
\ifx \showISSN     \undefined \def \showISSN      #1{\unskip}     \fi
\ifx \showLCCN     \undefined \def \showLCCN      #1{\unskip}     \fi
\ifx \shownote     \undefined \def \shownote      #1{#1}          \fi
\ifx \showarticletitle \undefined \def \showarticletitle #1{#1}   \fi
\ifx \showURL      \undefined \def \showURL       {\relax}        \fi
\providecommand\bibfield[2]{#2}
\providecommand\bibinfo[2]{#2}
\providecommand\natexlab[1]{#1}
\providecommand\showeprint[2][]{arXiv:#2}

\bibitem[\protect\citeauthoryear{Abramson, Brown, Edwards, Murphy, and
  Winkler}{Abramson et~al\mbox{.}}{1996}]%
        {abramson1996hailfinder}
\bibfield{author}{\bibinfo{person}{Bruce Abramson}, \bibinfo{person}{John
  Brown}, \bibinfo{person}{Ward Edwards}, \bibinfo{person}{Allan Murphy}, {and}
  \bibinfo{person}{Robert~L Winkler}.} \bibinfo{year}{1996}\natexlab{}.
\newblock \showarticletitle{Hailfinder: A Bayesian system for forecasting
  severe weather}.
\newblock \bibinfo{journal}{\emph{International Journal of Forecasting}}
  \bibinfo{volume}{12}, \bibinfo{number}{1} (\bibinfo{year}{1996}),
  \bibinfo{pages}{57--71}.
\newblock
\urldef\tempurl%
\url{https://doi.org/10.1016/0169-2070(95)00664-8}
\showDOI{\tempurl}


\bibitem[\protect\citeauthoryear{Albarghouthi, D'Antoni, Drews, and
  Nori}{Albarghouthi et~al\mbox{.}}{2017}]%
        {Albarghouthi2017_1}
\bibfield{author}{\bibinfo{person}{Aws Albarghouthi}, \bibinfo{person}{Loris
  D'Antoni}, \bibinfo{person}{Samuel Drews}, {and} \bibinfo{person}{Aditya~V.
  Nori}.} \bibinfo{year}{2017}\natexlab{}.
\newblock \showarticletitle{FairSquare: Probabilistic Verification of Program
  Fairness}.
\newblock \bibinfo{journal}{\emph{Proc. ACM Program. Lang.}}
  \bibinfo{volume}{1}, \bibinfo{number}{OOPSLA}, Article
  \bibinfo{articleno}{80} (\bibinfo{date}{Oct.} \bibinfo{year}{2017}),
  \bibinfo{numpages}{30}~pages.
\newblock
\showISSN{2475-1421}
\urldef\tempurl%
\url{https://doi.org/10.1145/3133904}
\showDOI{\tempurl}


\bibitem[\protect\citeauthoryear{Andreassen, Jensen, Andersen, Falck,
  Kj{\ae}rulff, Woldbye, S{\o}rensen, Rosenfalck, and Jensen}{Andreassen
  et~al\mbox{.}}{1989}]%
        {andreassen1989munin}
\bibfield{author}{\bibinfo{person}{Steen Andreassen}, \bibinfo{person}{Finn~V
  Jensen}, \bibinfo{person}{Stig~Kj{\ae}r Andersen}, \bibinfo{person}{B Falck},
  \bibinfo{person}{U Kj{\ae}rulff}, \bibinfo{person}{M Woldbye},
  \bibinfo{person}{AR S{\o}rensen}, \bibinfo{person}{A Rosenfalck}, {and}
  \bibinfo{person}{F Jensen}.} \bibinfo{year}{1989}\natexlab{}.
\newblock \showarticletitle{MUNIN: an expert EMG Assistant}.
\newblock In \bibinfo{booktitle}{\emph{Computer-aided electromyography and
  expert systems}}. \bibinfo{publisher}{Pergamon Press},
  \bibinfo{pages}{255--277}.
\newblock
\urldef\tempurl%
\url{https://doi.org/10.1016/0924-980x(95)00252-g}
\showDOI{\tempurl}


\bibitem[\protect\citeauthoryear{Bahar, Frohm, Gaona, Hachtel, Macii, Pardo,
  and Somenzi}{Bahar et~al\mbox{.}}{1997}]%
        {bahar1997algebric}
\bibfield{author}{\bibinfo{person}{R~Iris Bahar}, \bibinfo{person}{Erica~A
  Frohm}, \bibinfo{person}{Charles~M Gaona}, \bibinfo{person}{Gary~D Hachtel},
  \bibinfo{person}{Enrico Macii}, \bibinfo{person}{Abelardo Pardo}, {and}
  \bibinfo{person}{Fabio Somenzi}.} \bibinfo{year}{1997}\natexlab{}.
\newblock \showarticletitle{Algebric decision diagrams and their applications}.
\newblock \bibinfo{journal}{\emph{Formal methods in system design}}
  \bibinfo{volume}{10}, \bibinfo{number}{2-3} (\bibinfo{year}{1997}),
  \bibinfo{pages}{171--206}.
\newblock


\bibitem[\protect\citeauthoryear{Beinlich, Suermondt, Chavez, and
  Cooper}{Beinlich et~al\mbox{.}}{1989}]%
        {beinlich1989alarm}
\bibfield{author}{\bibinfo{person}{Ingo~A Beinlich},
  \bibinfo{person}{Henri~Jacques Suermondt}, \bibinfo{person}{R~Martin Chavez},
  {and} \bibinfo{person}{Gregory~F Cooper}.} \bibinfo{year}{1989}\natexlab{}.
\newblock \showarticletitle{The ALARM monitoring system: A case study with two
  probabilistic inference techniques for belief networks}.
\newblock In \bibinfo{booktitle}{\emph{AIME 89}}.
  \bibinfo{publisher}{Springer}, \bibinfo{pages}{247--256}.
\newblock
\urldef\tempurl%
\url{https://doi.org/10.1007/978-3-642-93437-7_28}
\showDOI{\tempurl}


\bibitem[\protect\citeauthoryear{Belle, Passerini, and Van~den Broeck}{Belle
  et~al\mbox{.}}{2015}]%
        {BelleIJCAI15}
\bibfield{author}{\bibinfo{person}{Vaishak Belle}, \bibinfo{person}{Andrea
  Passerini}, {and} \bibinfo{person}{Guy Van~den Broeck}.}
  \bibinfo{year}{2015}\natexlab{}.
\newblock \showarticletitle{Probabilistic Inference in Hybrid Domains by
  Weighted Model Integration}. In \bibinfo{booktitle}{\emph{Proc.~of IJCAI}}.
  \bibinfo{pages}{2770--2776}.
\newblock


\bibitem[\protect\citeauthoryear{Biere}{Biere}{2009}]%
        {BiereModelChecking}
\bibfield{author}{\bibinfo{person}{Armin Biere}.}
  \bibinfo{year}{2009}\natexlab{}.
\newblock \showarticletitle{Bounded Model Checking}.
\newblock In \bibinfo{booktitle}{\emph{Handbook of Satisfiability}},
  \bibfield{editor}{\bibinfo{person}{Armin Biere}, \bibinfo{person}{Marijn
  J.~H. Heule}, \bibinfo{person}{Hans van Maaren}, {and} \bibinfo{person}{Toby
  Walsh}} (Eds.). \bibinfo{series}{Frontiers in Artificial Intelligence and
  Applications}, Vol.~\bibinfo{volume}{185}. \bibinfo{publisher}{IOS Press},
  Chapter~14.
\newblock


\bibitem[\protect\citeauthoryear{Binder, Koller, Russell, and Kanazawa}{Binder
  et~al\mbox{.}}{1997}]%
        {binder1997adaptive}
\bibfield{author}{\bibinfo{person}{John Binder}, \bibinfo{person}{Daphne
  Koller}, \bibinfo{person}{Stuart Russell}, {and} \bibinfo{person}{Keiji
  Kanazawa}.} \bibinfo{year}{1997}\natexlab{}.
\newblock \showarticletitle{Adaptive probabilistic networks with hidden
  variables}.
\newblock \bibinfo{journal}{\emph{Machine Learning}} \bibinfo{volume}{29},
  \bibinfo{number}{2-3} (\bibinfo{year}{1997}), \bibinfo{pages}{213--244}.
\newblock
\urldef\tempurl%
\url{https://doi.org/10.1023/A:1007421730016}
\showDOI{\tempurl}


\bibitem[\protect\citeauthoryear{Bingham, Chen, Jankowiak, Obermeyer, Pradhan,
  Karaletsos, Singh, Szerlip, Horsfall, and Goodman}{Bingham
  et~al\mbox{.}}{2019}]%
        {bingham2019pyro}
\bibfield{author}{\bibinfo{person}{Eli Bingham}, \bibinfo{person}{Jonathan~P
  Chen}, \bibinfo{person}{Martin Jankowiak}, \bibinfo{person}{Fritz Obermeyer},
  \bibinfo{person}{Neeraj Pradhan}, \bibinfo{person}{Theofanis Karaletsos},
  \bibinfo{person}{Rohit Singh}, \bibinfo{person}{Paul Szerlip},
  \bibinfo{person}{Paul Horsfall}, {and} \bibinfo{person}{Noah~D Goodman}.}
  \bibinfo{year}{2019}\natexlab{}.
\newblock \showarticletitle{Pyro: Deep universal probabilistic programming}.
\newblock \bibinfo{journal}{\emph{The Journal of Machine Learning Research}}
  \bibinfo{volume}{20}, \bibinfo{number}{1} (\bibinfo{year}{2019}),
  \bibinfo{pages}{973--978}.
\newblock


\bibitem[\protect\citeauthoryear{Borgstr{\"o}m, Gordon, Greenberg, Margetson,
  and Van~Gael}{Borgstr{\"o}m et~al\mbox{.}}{2011}]%
        {borgstrom2011measure}
\bibfield{author}{\bibinfo{person}{Johannes Borgstr{\"o}m},
  \bibinfo{person}{Andrew~D Gordon}, \bibinfo{person}{Michael Greenberg},
  \bibinfo{person}{James Margetson}, {and} \bibinfo{person}{Jurgen Van~Gael}.}
  \bibinfo{year}{2011}\natexlab{}.
\newblock \showarticletitle{Measure transformer semantics for Bayesian machine
  learning}. In \bibinfo{booktitle}{\emph{European symposium on programming}}.
  Springer, \bibinfo{pages}{77--96}.
\newblock


\bibitem[\protect\citeauthoryear{Bornholt, Mytkowicz, and McKinley}{Bornholt
  et~al\mbox{.}}{2014}]%
        {bornholt2014uncertain}
\bibfield{author}{\bibinfo{person}{James Bornholt}, \bibinfo{person}{Todd
  Mytkowicz}, {and} \bibinfo{person}{Kathryn~S McKinley}.}
  \bibinfo{year}{2014}\natexlab{}.
\newblock \showarticletitle{Uncertain<T>: A first-order type for uncertain
  data}. In \bibinfo{booktitle}{\emph{ACM SIGPLAN Notices}},
  Vol.~\bibinfo{volume}{49}. ACM, \bibinfo{pages}{51--66}.
\newblock
\urldef\tempurl%
\url{https://doi.org/10.1145/2654822.2541958}
\showDOI{\tempurl}


\bibitem[\protect\citeauthoryear{Boutilier, Friedman, Goldszmidt, and
  Koller}{Boutilier et~al\mbox{.}}{1996}]%
        {boutilier1996context}
\bibfield{author}{\bibinfo{person}{Craig Boutilier}, \bibinfo{person}{Nir
  Friedman}, \bibinfo{person}{Moises Goldszmidt}, {and} \bibinfo{person}{Daphne
  Koller}.} \bibinfo{year}{1996}\natexlab{}.
\newblock \showarticletitle{Context-specific independence in Bayesian
  networks}. In \bibinfo{booktitle}{\emph{Proceedings of the Twelfth
  international conference on Uncertainty in artificial intelligence}}. Morgan
  Kaufmann Publishers Inc., \bibinfo{pages}{115--123}.
\newblock


\bibitem[\protect\citeauthoryear{Bryant}{Bryant}{1986}]%
        {Bryant86}
\bibfield{author}{\bibinfo{person}{R. Bryant}.}
  \bibinfo{year}{1986}\natexlab{}.
\newblock \showarticletitle{Graph-based algorithms for {Boolean} function
  manipulation}.
\newblock \bibinfo{journal}{\emph{IEEE TC}}  \bibinfo{volume}{C-35}
  (\bibinfo{year}{1986}), \bibinfo{pages}{677--691}.
\newblock
\urldef\tempurl%
\url{https://doi.org/10.1109/TC.1986.1676819}
\showDOI{\tempurl}


\bibitem[\protect\citeauthoryear{Carpenter, Gelman, Hoffman, Lee, Goodrich,
  Betancourt, Brubaker, Guo, Li, and Riddell}{Carpenter et~al\mbox{.}}{2016}]%
        {carpenter2016stan}
\bibfield{author}{\bibinfo{person}{Bob Carpenter}, \bibinfo{person}{Andrew
  Gelman}, \bibinfo{person}{Matt Hoffman}, \bibinfo{person}{Daniel Lee},
  \bibinfo{person}{Ben Goodrich}, \bibinfo{person}{Michael Betancourt},
  \bibinfo{person}{Michael~A Brubaker}, \bibinfo{person}{Jiqiang Guo},
  \bibinfo{person}{Peter Li}, {and} \bibinfo{person}{Allen Riddell}.}
  \bibinfo{year}{2016}\natexlab{}.
\newblock \showarticletitle{Stan: A probabilistic programming language}.
\newblock \bibinfo{journal}{\emph{Journal of Statistical Software}}
  (\bibinfo{year}{2016}).
\newblock


\bibitem[\protect\citeauthoryear{Chaganty, Nori, and Rajamani}{Chaganty
  et~al\mbox{.}}{2013}]%
        {chaganty2013efficiently}
\bibfield{author}{\bibinfo{person}{Arun Chaganty}, \bibinfo{person}{Aditya
  Nori}, {and} \bibinfo{person}{Sriram Rajamani}.}
  \bibinfo{year}{2013}\natexlab{}.
\newblock \showarticletitle{Efficiently sampling probabilistic programs via
  program analysis}. In \bibinfo{booktitle}{\emph{Artificial Intelligence and
  Statistics}}. \bibinfo{pages}{153--160}.
\newblock


\bibitem[\protect\citeauthoryear{Chavira and Darwiche}{Chavira and
  Darwiche}{2005}]%
        {chavira2005compiling}
\bibfield{author}{\bibinfo{person}{Mark Chavira} {and} \bibinfo{person}{Adnan
  Darwiche}.} \bibinfo{year}{2005}\natexlab{}.
\newblock \showarticletitle{Compiling {Bayesian} networks with local
  structure}. In \bibinfo{booktitle}{\emph{IJCAI}}.
  \bibinfo{pages}{1306--1312}.
\newblock


\bibitem[\protect\citeauthoryear{Chavira and Darwiche}{Chavira and
  Darwiche}{2008}]%
        {Chavira2008}
\bibfield{author}{\bibinfo{person}{Mark Chavira} {and} \bibinfo{person}{Adnan
  Darwiche}.} \bibinfo{year}{2008}\natexlab{}.
\newblock \showarticletitle{On Probabilistic Inference by Weighted Model
  Counting}.
\newblock \bibinfo{journal}{\emph{J. Artificial Intelligence}}
  \bibinfo{volume}{172}, \bibinfo{number}{6-7} (\bibinfo{date}{April}
  \bibinfo{year}{2008}), \bibinfo{pages}{772--799}.
\newblock
\showISSN{0004-3702}
\urldef\tempurl%
\url{https://doi.org/10.1016/j.artint.2007.11.002}
\showDOI{\tempurl}


\bibitem[\protect\citeauthoryear{Chavira, Darwiche, and Jaeger}{Chavira
  et~al\mbox{.}}{2006}]%
        {chavira2006compiling}
\bibfield{author}{\bibinfo{person}{Mark Chavira}, \bibinfo{person}{Adnan
  Darwiche}, {and} \bibinfo{person}{Manfred Jaeger}.}
  \bibinfo{year}{2006}\natexlab{}.
\newblock \showarticletitle{{Compiling relational Bayesian networks for exact
  inference}}.
\newblock \bibinfo{journal}{\emph{International Journal of Approximate
  Reasoning}} \bibinfo{volume}{42}, \bibinfo{number}{1} (\bibinfo{year}{2006}),
  \bibinfo{pages}{4--20}.
\newblock


\bibitem[\protect\citeauthoryear{Chistikov, Dimitrova, and Majumdar}{Chistikov
  et~al\mbox{.}}{2015}]%
        {Chistikov2015}
\bibfield{author}{\bibinfo{person}{Dmitry Chistikov}, \bibinfo{person}{Rayna
  Dimitrova}, {and} \bibinfo{person}{Rupak Majumdar}.}
  \bibinfo{year}{2015}\natexlab{}.
\newblock \showarticletitle{Approximate Counting in SMT and Value Estimation
  for Probabilistic Programs}. In \bibinfo{booktitle}{\emph{Proc.~of TACAS}}.
  \bibinfo{publisher}{Springer-Verlag New York, Inc.}, \bibinfo{address}{New
  York, NY, USA}, \bibinfo{pages}{320--334}.
\newblock
\showISBNx{978-3-662-46680-3}
\urldef\tempurl%
\url{https://doi.org/10.1007/978-3-662-46681-0_26}
\showDOI{\tempurl}


\bibitem[\protect\citeauthoryear{Claret, Rajamani, Nori, Gordon, and
  Borgstr{\"{o}}m}{Claret et~al\mbox{.}}{2013}]%
        {Claret2013}
\bibfield{author}{\bibinfo{person}{Guillaume Claret},
  \bibinfo{person}{Sriram~K. Rajamani}, \bibinfo{person}{Aditya~V. Nori},
  \bibinfo{person}{Andrew~D. Gordon}, {and} \bibinfo{person}{Johannes
  Borgstr{\"{o}}m}.} \bibinfo{year}{2013}\natexlab{}.
\newblock \showarticletitle{{Bayesian inference using data flow analysis}}.
\newblock \bibinfo{journal}{\emph{Proceedings of the 2013 9th Joint Meeting on
  Foundations of Software Engineering - ESEC/FSE 2013}} (\bibinfo{year}{2013}),
  \bibinfo{pages}{92}.
\newblock
\showISBNx{9781450322379}
\urldef\tempurl%
\url{https://doi.org/10.1145/2491411.2491423}
\showDOI{\tempurl}


\bibitem[\protect\citeauthoryear{Clarke, Grumberg, and Peled}{Clarke
  et~al\mbox{.}}{1999}]%
        {Clarke2000}
\bibfield{author}{\bibinfo{person}{Edmund~M. Clarke, Jr.},
  \bibinfo{person}{Orna Grumberg}, {and} \bibinfo{person}{Doron~A. Peled}.}
  \bibinfo{year}{1999}\natexlab{}.
\newblock \bibinfo{booktitle}{\emph{Model Checking}}.
\newblock \bibinfo{publisher}{MIT Press}, \bibinfo{address}{Cambridge, MA,
  USA}.
\newblock
\showISBNx{0-262-03270-8}


\bibitem[\protect\citeauthoryear{Cousot and Monerau}{Cousot and
  Monerau}{2012}]%
        {Cousot2012}
\bibfield{author}{\bibinfo{person}{Patrick Cousot} {and}
  \bibinfo{person}{Michael Monerau}.} \bibinfo{year}{2012}\natexlab{}.
\newblock \showarticletitle{{Probabilistic abstract interpretation}}. In
  \bibinfo{booktitle}{\emph{Proc.~of ESOP}}. \bibinfo{pages}{169--193}.
\newblock
\showISBNx{9783642288685}
\showISSN{03029743}
\urldef\tempurl%
\url{https://doi.org/10.1007/978-3-642-28869-2\_9}
\showDOI{\tempurl}


\bibitem[\protect\citeauthoryear{Cusumano-Towner, Bichsel, Gehr, Vechev, and
  Mansinghka}{Cusumano-Towner et~al\mbox{.}}{2018}]%
        {cusumano2018incremental}
\bibfield{author}{\bibinfo{person}{Marco Cusumano-Towner},
  \bibinfo{person}{Benjamin Bichsel}, \bibinfo{person}{Timon Gehr},
  \bibinfo{person}{Martin Vechev}, {and} \bibinfo{person}{Vikash~K
  Mansinghka}.} \bibinfo{year}{2018}\natexlab{}.
\newblock \showarticletitle{Incremental inference for probabilistic programs}.
  In \bibinfo{booktitle}{\emph{ACM SIGPLAN Notices}},
  Vol.~\bibinfo{volume}{53}. ACM, \bibinfo{pages}{571--585}.
\newblock
\urldef\tempurl%
\url{https://doi.org/10.1145/3296979.3192399}
\showDOI{\tempurl}


\bibitem[\protect\citeauthoryear{Darwiche}{Darwiche}{2009}]%
        {Darwiche09}
\bibfield{author}{\bibinfo{person}{Adnan Darwiche}.}
  \bibinfo{year}{2009}\natexlab{}.
\newblock \bibinfo{booktitle}{\emph{Modeling and Reasoning with {B}ayesian
  Networks}}.
\newblock \bibinfo{publisher}{Cambridge University Press}.
\newblock
\urldef\tempurl%
\url{https://doi.org/10.1017/CBO9780511811357}
\showDOI{\tempurl}


\bibitem[\protect\citeauthoryear{Darwiche}{Darwiche}{2011}]%
        {darwiche2011sdd}
\bibfield{author}{\bibinfo{person}{Adnan Darwiche}.}
  \bibinfo{year}{2011}\natexlab{}.
\newblock \showarticletitle{SDD: A new canonical representation of
  propositional knowledge bases}. In \bibinfo{booktitle}{\emph{IJCAI
  Proceedings-International Joint Conference on Artificial Intelligence}}.
  \bibinfo{pages}{819}.
\newblock


\bibitem[\protect\citeauthoryear{Darwiche and Marquis}{Darwiche and
  Marquis}{2002}]%
        {Darwiche2002}
\bibfield{author}{\bibinfo{person}{A. Darwiche} {and} \bibinfo{person}{P.
  Marquis}.} \bibinfo{year}{2002}\natexlab{}.
\newblock \showarticletitle{A Knowledge Compilation Map}.
\newblock \bibinfo{journal}{\emph{Journal of Artificial Intelligence Research}}
   \bibinfo{volume}{17} (\bibinfo{year}{2002}), \bibinfo{pages}{229--264}.
\newblock


\bibitem[\protect\citeauthoryear{De~Raedt, Kimmig, and Toivonen}{De~Raedt
  et~al\mbox{.}}{2007}]%
        {de2007problog}
\bibfield{author}{\bibinfo{person}{Luc De~Raedt}, \bibinfo{person}{Angelika
  Kimmig}, {and} \bibinfo{person}{Hannu Toivonen}.}
  \bibinfo{year}{2007}\natexlab{}.
\newblock \showarticletitle{ProbLog: A Probabilistic Prolog and Its Application
  in Link Discovery}. In \bibinfo{booktitle}{\emph{Proceedings of IJCAI}},
  Vol.~\bibinfo{volume}{7}. \bibinfo{pages}{2462--2467}.
\newblock


\bibitem[\protect\citeauthoryear{Dehnert, Junges, Katoen, and Volk}{Dehnert
  et~al\mbox{.}}{2017}]%
        {dehnert2017storm}
\bibfield{author}{\bibinfo{person}{Christian Dehnert},
  \bibinfo{person}{Sebastian Junges}, \bibinfo{person}{Joost-Pieter Katoen},
  {and} \bibinfo{person}{Matthias Volk}.} \bibinfo{year}{2017}\natexlab{}.
\newblock \showarticletitle{A storm is coming: A modern probabilistic model
  checker}. In \bibinfo{booktitle}{\emph{International Conference on Computer
  Aided Verification}}. Springer, \bibinfo{pages}{592--600}.
\newblock


\bibitem[\protect\citeauthoryear{Dillon, Langmore, Tran, Brevdo, Vasudevan,
  Moore, Patton, Alemi, Hoffman, and Saurous}{Dillon et~al\mbox{.}}{2017}]%
        {dillon2017tensorflow}
\bibfield{author}{\bibinfo{person}{Joshua~V Dillon}, \bibinfo{person}{Ian
  Langmore}, \bibinfo{person}{Dustin Tran}, \bibinfo{person}{Eugene Brevdo},
  \bibinfo{person}{Srinivas Vasudevan}, \bibinfo{person}{Dave Moore},
  \bibinfo{person}{Brian Patton}, \bibinfo{person}{Alex Alemi},
  \bibinfo{person}{Matt Hoffman}, {and} \bibinfo{person}{Rif~A Saurous}.}
  \bibinfo{year}{2017}\natexlab{}.
\newblock \showarticletitle{TensorFlow Distributions}.
\newblock \bibinfo{journal}{\emph{arXiv preprint arXiv:1711.10604}}
  (\bibinfo{year}{2017}).
\newblock


\bibitem[\protect\citeauthoryear{Dos~Martires, Dries, and
  De~Raedt}{Dos~Martires et~al\mbox{.}}{2019}]%
        {dos2019exact}
\bibfield{author}{\bibinfo{person}{Pedro~Zuidberg Dos~Martires},
  \bibinfo{person}{Anton Dries}, {and} \bibinfo{person}{Luc De~Raedt}.}
  \bibinfo{year}{2019}\natexlab{}.
\newblock \showarticletitle{Exact and Approximate Weighted Model Integration
  with Probability Density Functions Using Knowledge Compilation}. In
  \bibinfo{booktitle}{\emph{Proceedings of the AAAI Conference on Artificial
  Intelligence}}, Vol.~\bibinfo{volume}{33}. \bibinfo{pages}{7825--7833}.
\newblock


\bibitem[\protect\citeauthoryear{Fierens, Van~den Broeck, Renkens, Shterionov,
  Gutmann, Thon, Janssens, and De~Raedt}{Fierens et~al\mbox{.}}{2015}]%
        {Fierens2015}
\bibfield{author}{\bibinfo{person}{Daan Fierens}, \bibinfo{person}{Guy Van~den
  Broeck}, \bibinfo{person}{Joris Renkens}, \bibinfo{person}{Dimitar
  Shterionov}, \bibinfo{person}{Bernd Gutmann}, \bibinfo{person}{Ingo Thon},
  \bibinfo{person}{Gerda Janssens}, {and} \bibinfo{person}{Luc De~Raedt}.}
  \bibinfo{year}{2015}\natexlab{}.
\newblock \showarticletitle{Inference and learning in probabilistic logic
  programs using weighted Boolean formulas}.
\newblock \bibinfo{journal}{\emph{J. Theory and Practice of Logic Programming}}
   \bibinfo{volume}{15(3)} (\bibinfo{year}{2015}), \bibinfo{pages}{358 -- 401}.
\newblock
\urldef\tempurl%
\url{https://doi.org/10.1017/S1471068414000076}
\showDOI{\tempurl}


\bibitem[\protect\citeauthoryear{Filieri, P{\u{a}}s{\u{a}}reanu, and
  Visser}{Filieri et~al\mbox{.}}{2013}]%
        {filieri2013reliability}
\bibfield{author}{\bibinfo{person}{Antonio Filieri}, \bibinfo{person}{Corina~S
  P{\u{a}}s{\u{a}}reanu}, {and} \bibinfo{person}{Willem Visser}.}
  \bibinfo{year}{2013}\natexlab{}.
\newblock \showarticletitle{Reliability analysis in symbolic pathfinder}. In
  \bibinfo{booktitle}{\emph{2013 35th International Conference on Software
  Engineering (ICSE)}}. IEEE, \bibinfo{pages}{622--631}.
\newblock


\bibitem[\protect\citeauthoryear{Flanagan, Sabry, Duba, and Felleisen}{Flanagan
  et~al\mbox{.}}{1993}]%
        {DBLP:conf/pldi/FlanaganSDF93}
\bibfield{author}{\bibinfo{person}{Cormac Flanagan}, \bibinfo{person}{Amr
  Sabry}, \bibinfo{person}{Bruce~F. Duba}, {and} \bibinfo{person}{Matthias
  Felleisen}.} \bibinfo{year}{1993}\natexlab{}.
\newblock \showarticletitle{The Essence of Compiling with Continuations}. In
  \bibinfo{booktitle}{\emph{Proceedings of the {ACM} SIGPLAN'93 Conference on
  Programming Language Design and Implementation (PLDI), Albuquerque, New
  Mexico, USA, June 23-25, 1993}}, \bibfield{editor}{\bibinfo{person}{Robert
  Cartwright}} (Ed.). \bibinfo{publisher}{{ACM}}, \bibinfo{pages}{237--247}.
\newblock
\urldef\tempurl%
\url{https://doi.org/10.1145/155090.155113}
\showDOI{\tempurl}


\bibitem[\protect\citeauthoryear{Gehr, Misailovic, Tsankov, Vanbever, Wiesmann,
  and Vechev}{Gehr et~al\mbox{.}}{2018}]%
        {gehr2018bayonet}
\bibfield{author}{\bibinfo{person}{Timon Gehr}, \bibinfo{person}{Sasa
  Misailovic}, \bibinfo{person}{Petar Tsankov}, \bibinfo{person}{Laurent
  Vanbever}, \bibinfo{person}{Pascal Wiesmann}, {and} \bibinfo{person}{Martin
  Vechev}.} \bibinfo{year}{2018}\natexlab{}.
\newblock \showarticletitle{Bayonet: probabilistic inference for networks}. In
  \bibinfo{booktitle}{\emph{ACM SIGPLAN Notices}}, Vol.~\bibinfo{volume}{53}.
  ACM, \bibinfo{pages}{586--602}.
\newblock
\urldef\tempurl%
\url{https://doi.org/10.1145/3296979.3192400}
\showDOI{\tempurl}


\bibitem[\protect\citeauthoryear{Gehr, Misailovic, and Vechev}{Gehr
  et~al\mbox{.}}{2016}]%
        {gehr2016psi}
\bibfield{author}{\bibinfo{person}{Timon Gehr}, \bibinfo{person}{Sasa
  Misailovic}, {and} \bibinfo{person}{Martin Vechev}.}
  \bibinfo{year}{2016}\natexlab{}.
\newblock \showarticletitle{Psi: Exact symbolic inference for probabilistic
  programs}. In \bibinfo{booktitle}{\emph{International Conference on Computer
  Aided Verification}}. Springer, \bibinfo{pages}{62--83}.
\newblock


\bibitem[\protect\citeauthoryear{Geldenhuys, Dwyer, and Visser}{Geldenhuys
  et~al\mbox{.}}{2012}]%
        {geldenhuys2012probabilistic}
\bibfield{author}{\bibinfo{person}{Jaco Geldenhuys}, \bibinfo{person}{Matthew~B
  Dwyer}, {and} \bibinfo{person}{Willem Visser}.}
  \bibinfo{year}{2012}\natexlab{}.
\newblock \showarticletitle{Probabilistic symbolic execution}. In
  \bibinfo{booktitle}{\emph{Proceedings of the 2012 International Symposium on
  Software Testing and Analysis}}. ACM, \bibinfo{pages}{166--176}.
\newblock
\urldef\tempurl%
\url{https://doi.org/10.1145/2338965.2336773}
\showDOI{\tempurl}


\bibitem[\protect\citeauthoryear{Gogate and Dechter}{Gogate and
  Dechter}{2011}]%
        {gogate2011samplesearch}
\bibfield{author}{\bibinfo{person}{V. Gogate} {and} \bibinfo{person}{R.
  Dechter}.} \bibinfo{year}{2011}\natexlab{}.
\newblock \showarticletitle{SampleSearch: Importance sampling in presence of
  determinism}.
\newblock \bibinfo{journal}{\emph{Artificial Intelligence}}
  \bibinfo{volume}{175}, \bibinfo{number}{2} (\bibinfo{year}{2011}),
  \bibinfo{pages}{694--729}.
\newblock


\bibitem[\protect\citeauthoryear{Goodman, Mansinghka, Roy, Bonawitz, and
  Tenenbaum}{Goodman et~al\mbox{.}}{2008}]%
        {goodman:uai08}
\bibfield{author}{\bibinfo{person}{Noah~D. Goodman}, \bibinfo{person}{Vikash~K.
  Mansinghka}, \bibinfo{person}{Daniel~M. Roy}, \bibinfo{person}{Keith
  Bonawitz}, {and} \bibinfo{person}{Joshua~B. Tenenbaum}.}
  \bibinfo{year}{2008}\natexlab{}.
\newblock \showarticletitle{{Church: a language for generative models}}. In
  \bibinfo{booktitle}{\emph{{Proceedings of the 24th Conference in Uncertainty
  in Artificial Intelligence (UAI)}}}.
\newblock


\bibitem[\protect\citeauthoryear{Goodman and Stuhlm{\"u}ller}{Goodman and
  Stuhlm{\"u}ller}{2014}]%
        {goodman2014design}
\bibfield{author}{\bibinfo{person}{Noah~D Goodman} {and}
  \bibinfo{person}{Andreas Stuhlm{\"u}ller}.} \bibinfo{year}{2014}\natexlab{}.
\newblock \bibinfo{title}{The design and implementation of probabilistic
  programming languages}.
\newblock
\newblock


\bibitem[\protect\citeauthoryear{Gorinova, Moore, and Hoffman}{Gorinova
  et~al\mbox{.}}{2020}]%
        {gorinova2019automatic}
\bibfield{author}{\bibinfo{person}{Maria~I Gorinova}, \bibinfo{person}{Dave
  Moore}, {and} \bibinfo{person}{Matthew~D Hoffman}.}
  \bibinfo{year}{2020}\natexlab{}.
\newblock \showarticletitle{Automatic Reparameterisation of Probabilistic
  Programs}.
\newblock \bibinfo{journal}{\emph{International Conference on Machine Learning
  (ICML)}} (\bibinfo{year}{2020}).
\newblock


\bibitem[\protect\citeauthoryear{Gram-Hansen, Zhou, Kohn, Rainforth, Yang, and
  Wood}{Gram-Hansen et~al\mbox{.}}{2018}]%
        {gram2018hamiltonian}
\bibfield{author}{\bibinfo{person}{Bradley Gram-Hansen}, \bibinfo{person}{Yuan
  Zhou}, \bibinfo{person}{Tobias Kohn}, \bibinfo{person}{Tom Rainforth},
  \bibinfo{person}{Hongseok Yang}, {and} \bibinfo{person}{Frank Wood}.}
  \bibinfo{year}{2018}\natexlab{}.
\newblock \showarticletitle{Hamiltonian Monte Carlo for Probabilistic Programs
  with Discontinuities}.
\newblock \bibinfo{journal}{\emph{arXiv preprint arXiv:1804.03523}}
  (\bibinfo{year}{2018}).
\newblock


\bibitem[\protect\citeauthoryear{Hoffman and Gelman}{Hoffman and
  Gelman}{2014}]%
        {hoffman2014no}
\bibfield{author}{\bibinfo{person}{Matthew~D Hoffman} {and}
  \bibinfo{person}{Andrew Gelman}.} \bibinfo{year}{2014}\natexlab{}.
\newblock \showarticletitle{The No-U-turn sampler: adaptively setting path
  lengths in Hamiltonian Monte Carlo.}
\newblock \bibinfo{journal}{\emph{Journal of Machine Learning Research}}
  \bibinfo{volume}{15}, \bibinfo{number}{1} (\bibinfo{year}{2014}),
  \bibinfo{pages}{1593--1623}.
\newblock


\bibitem[\protect\citeauthoryear{Holtzen, Van~den Broeck, and
  Millstein}{Holtzen et~al\mbox{.}}{2018}]%
        {holtzen2018sound}
\bibfield{author}{\bibinfo{person}{Steven Holtzen}, \bibinfo{person}{Guy
  Van~den Broeck}, {and} \bibinfo{person}{Todd Millstein}.}
  \bibinfo{year}{2018}\natexlab{}.
\newblock \showarticletitle{Sound abstraction and decomposition of
  probabilistic programs}. In \bibinfo{booktitle}{\emph{Proceedings of the 35th
  International Conference on Machine Learning (ICML)}}.
\newblock


\bibitem[\protect\citeauthoryear{Huang and Morrisett}{Huang and
  Morrisett}{2016}]%
        {Huang2016}
\bibfield{author}{\bibinfo{person}{Daniel Huang} {and} \bibinfo{person}{Greg
  Morrisett}.} \bibinfo{year}{2016}\natexlab{}.
\newblock \showarticletitle{An Application of Computable Distributions to the
  Semantics of Probabilistic Programming Languages}. In
  \bibinfo{booktitle}{\emph{Proceedings of the 25th European Symposium on
  Programming Languages and Systems - Volume 9632}}.
  \bibinfo{publisher}{Springer-Verlag New York, Inc.}, \bibinfo{address}{New
  York, NY, USA}, \bibinfo{pages}{337--363}.
\newblock
\showISBNx{978-3-662-49497-4}
\urldef\tempurl%
\url{https://doi.org/10.1007/978-3-662-49498-1_14}
\showDOI{\tempurl}


\bibitem[\protect\citeauthoryear{Hur, Nori, Rajamani, and Sammuel}{Hur
  et~al\mbox{.}}{2015}]%
        {Hur2015}
\bibfield{author}{\bibinfo{person}{Chung-kil Hur}, \bibinfo{person}{Aditya~V.
  Nori}, \bibinfo{person}{Sriram~K. Rajamani}, {and} \bibinfo{person}{Selva
  Sammuel}.} \bibinfo{year}{2015}\natexlab{}.
\newblock \showarticletitle{{A Provably Correct Sampler for Probabilistic
  Programs}}.
\newblock \bibinfo{journal}{\emph{FSTTCS}} \bibinfo{number}{FSTTCS}
  (\bibinfo{year}{2015}), \bibinfo{pages}{1--14}.
\newblock
\showISBNx{9783939897972}
\showISSN{18688969}
\urldef\tempurl%
\url{https://doi.org/10.4230/LIPIcs.FSTTCS.2015.475}
\showDOI{\tempurl}


\bibitem[\protect\citeauthoryear{Jensen, Kj{\ae}rulff, Olesen, and
  Pedersen}{Jensen et~al\mbox{.}}{1989}]%
        {jensen1989expert}
\bibfield{author}{\bibinfo{person}{FV Jensen}, \bibinfo{person}{U
  Kj{\ae}rulff}, \bibinfo{person}{KG Olesen}, {and} \bibinfo{person}{J
  Pedersen}.} \bibinfo{year}{1989}\natexlab{}.
\newblock \bibinfo{booktitle}{\emph{An expert system for control of waste water
  treatment—a pilot project}}.
\newblock \bibinfo{type}{{T}echnical {R}eport}. \bibinfo{institution}{Technical
  report, Judex Datasystemer A/S, Aalborg, 1989. In Danish}.
\newblock


\bibitem[\protect\citeauthoryear{Jhala and Majumdar}{Jhala and
  Majumdar}{2009}]%
        {Jhala2009}
\bibfield{author}{\bibinfo{person}{Ranjit Jhala} {and} \bibinfo{person}{Rupak
  Majumdar}.} \bibinfo{year}{2009}\natexlab{}.
\newblock \showarticletitle{{Software model checking}}.
\newblock \bibinfo{journal}{\emph{Comput. Surveys}} \bibinfo{volume}{41},
  \bibinfo{number}{4} (\bibinfo{year}{2009}), \bibinfo{pages}{1--54}.
\newblock
\showISBNx{9783540331025}
\showISSN{03600300}
\urldef\tempurl%
\url{https://doi.org/10.1145/1592434.1592438}
\showDOI{\tempurl}


\bibitem[\protect\citeauthoryear{Jordan, Ghahramani, Jaakkola, and Saul}{Jordan
  et~al\mbox{.}}{1999}]%
        {jordan1999introduction}
\bibfield{author}{\bibinfo{person}{M.I. Jordan}, \bibinfo{person}{Z.
  Ghahramani}, \bibinfo{person}{T.S. Jaakkola}, {and} \bibinfo{person}{L.K.
  Saul}.} \bibinfo{year}{1999}\natexlab{}.
\newblock \showarticletitle{An introduction to variational methods for
  graphical models}.
\newblock \bibinfo{journal}{\emph{Machine learning}} \bibinfo{volume}{37},
  \bibinfo{number}{2} (\bibinfo{year}{1999}), \bibinfo{pages}{183--233}.
\newblock
\urldef\tempurl%
\url{https://doi.org/10.1023/A:1007665907178}
\showDOI{\tempurl}


\bibitem[\protect\citeauthoryear{Katz, Menezes, Van~Oorschot, and
  Vanstone}{Katz et~al\mbox{.}}{1996}]%
        {katz1996handbook}
\bibfield{author}{\bibinfo{person}{Jonathan Katz}, \bibinfo{person}{Alfred~J
  Menezes}, \bibinfo{person}{Paul~C Van~Oorschot}, {and}
  \bibinfo{person}{Scott~A Vanstone}.} \bibinfo{year}{1996}\natexlab{}.
\newblock \bibinfo{booktitle}{\emph{Handbook of applied cryptography}}.
\newblock \bibinfo{publisher}{CRC press}.
\newblock


\bibitem[\protect\citeauthoryear{Koller and Friedman}{Koller and
  Friedman}{2009}]%
        {koller2009probabilistic}
\bibfield{author}{\bibinfo{person}{D. Koller} {and} \bibinfo{person}{N.
  Friedman}.} \bibinfo{year}{2009}\natexlab{}.
\newblock \bibinfo{booktitle}{\emph{Probabilistic graphical models: principles
  and techniques}}.
\newblock \bibinfo{publisher}{MIT press}.
\newblock


\bibitem[\protect\citeauthoryear{Korb and Nicholson}{Korb and
  Nicholson}{2010}]%
        {korb2010bayesian}
\bibfield{author}{\bibinfo{person}{Kevin~B Korb} {and} \bibinfo{person}{Ann~E
  Nicholson}.} \bibinfo{year}{2010}\natexlab{}.
\newblock \bibinfo{booktitle}{\emph{Bayesian artificial intelligence}}.
\newblock \bibinfo{publisher}{CRC press}.
\newblock
\urldef\tempurl%
\url{https://doi.org/10.1201/b10391}
\showDOI{\tempurl}


\bibitem[\protect\citeauthoryear{Kozen}{Kozen}{1979}]%
        {Kozen1979}
\bibfield{author}{\bibinfo{person}{Dexter Kozen}.}
  \bibinfo{year}{1979}\natexlab{}.
\newblock \showarticletitle{Semantics of Probabilistic Programs}. In
  \bibinfo{booktitle}{\emph{Proceedings of the 20th Annual Symposium on
  Foundations of Computer Science}} \emph{(\bibinfo{series}{SFCS '79})}.
  \bibinfo{publisher}{IEEE Computer Society}, \bibinfo{address}{Washington, DC,
  USA}, \bibinfo{pages}{101--114}.
\newblock
\urldef\tempurl%
\url{https://doi.org/10.1109/SFCS.1979.38}
\showDOI{\tempurl}


\bibitem[\protect\citeauthoryear{Kucukelbir, Ranganath, Gelman, and
  Blei}{Kucukelbir et~al\mbox{.}}{2015}]%
        {kucukelbir2015automatic}
\bibfield{author}{\bibinfo{person}{Alp Kucukelbir}, \bibinfo{person}{Rajesh
  Ranganath}, \bibinfo{person}{Andrew Gelman}, {and} \bibinfo{person}{David
  Blei}.} \bibinfo{year}{2015}\natexlab{}.
\newblock \showarticletitle{Automatic variational inference in Stan}. In
  \bibinfo{booktitle}{\emph{Advances in neural information processing
  systems}}. \bibinfo{pages}{568--576}.
\newblock


\bibitem[\protect\citeauthoryear{Kucukelbir, Tran, Ranganath, Gelman, and
  Blei}{Kucukelbir et~al\mbox{.}}{2017}]%
        {kucukelbir2017automatic}
\bibfield{author}{\bibinfo{person}{Alp Kucukelbir}, \bibinfo{person}{Dustin
  Tran}, \bibinfo{person}{Rajesh Ranganath}, \bibinfo{person}{Andrew Gelman},
  {and} \bibinfo{person}{David~M Blei}.} \bibinfo{year}{2017}\natexlab{}.
\newblock \showarticletitle{Automatic differentiation variational inference}.
\newblock \bibinfo{journal}{\emph{The Journal of Machine Learning Research}}
  \bibinfo{volume}{18}, \bibinfo{number}{1} (\bibinfo{year}{2017}),
  \bibinfo{pages}{430--474}.
\newblock


\bibitem[\protect\citeauthoryear{Kwiatkowska, Norman, and Parker}{Kwiatkowska
  et~al\mbox{.}}{2011}]%
        {Kwiatkowska2011}
\bibfield{author}{\bibinfo{person}{Marta Kwiatkowska}, \bibinfo{person}{Gethin
  Norman}, {and} \bibinfo{person}{David Parker}.}
  \bibinfo{year}{2011}\natexlab{}.
\newblock \showarticletitle{PRISM 4.0: Verification of Probabilistic Real-time
  Systems}. In \bibinfo{booktitle}{\emph{Proceedings of the 23rd International
  Conference on Computer Aided Verification}} (Snowbird, UT)
  \emph{(\bibinfo{series}{CAV'11})}. \bibinfo{publisher}{Springer-Verlag},
  \bibinfo{address}{Berlin, Heidelberg}, \bibinfo{pages}{585--591}.
\newblock
\showISBNx{978-3-642-22109-5}
\urldef\tempurl%
\url{https://doi.org/10.1007/978-3-642-22110-1_47}
\showDOI{\tempurl}


\bibitem[\protect\citeauthoryear{Kwisthout}{Kwisthout}{2009}]%
        {kwisthout2009computational}
\bibfield{author}{\bibinfo{person}{Johan Henri~Petrus Kwisthout}.}
  \bibinfo{year}{2009}\natexlab{}.
\newblock \bibinfo{booktitle}{\emph{The computational complexity of
  probabilistic networks}}.
\newblock \bibinfo{publisher}{Utrecht University}.
\newblock


\bibitem[\protect\citeauthoryear{Littman, Goldsmith, and Mundhenk}{Littman
  et~al\mbox{.}}{1998}]%
        {littman1998computational}
\bibfield{author}{\bibinfo{person}{Michael~L Littman}, \bibinfo{person}{Judy
  Goldsmith}, {and} \bibinfo{person}{Martin Mundhenk}.}
  \bibinfo{year}{1998}\natexlab{}.
\newblock \showarticletitle{The computational complexity of probabilistic
  planning}.
\newblock \bibinfo{journal}{\emph{Journal of Artificial Intelligence Research}}
   \bibinfo{volume}{9} (\bibinfo{year}{1998}), \bibinfo{pages}{1--36}.
\newblock
\urldef\tempurl%
\url{https://doi.org/10.1613/jair.505}
\showDOI{\tempurl}


\bibitem[\protect\citeauthoryear{Mansinghka, Kulkarni, Perov, and
  Tenenbaum}{Mansinghka et~al\mbox{.}}{2013}]%
        {mansinghka2013approximate}
\bibfield{author}{\bibinfo{person}{Vikash Mansinghka}, \bibinfo{person}{Tejas~D
  Kulkarni}, \bibinfo{person}{Yura~N Perov}, {and} \bibinfo{person}{Josh
  Tenenbaum}.} \bibinfo{year}{2013}\natexlab{}.
\newblock \showarticletitle{Approximate bayesian image interpretation using
  generative probabilistic graphics programs}. In
  \bibinfo{booktitle}{\emph{Advances in Neural Information Processing
  Systems}}. \bibinfo{pages}{1520--1528}.
\newblock


\bibitem[\protect\citeauthoryear{Mansinghka, Schaechtle, Handa, Radul, Chen,
  and Rinard}{Mansinghka et~al\mbox{.}}{2018}]%
        {Mansinghka2018}
\bibfield{author}{\bibinfo{person}{Vikash~K. Mansinghka},
  \bibinfo{person}{Ulrich Schaechtle}, \bibinfo{person}{Shivam Handa},
  \bibinfo{person}{Alexey Radul}, \bibinfo{person}{Yutian Chen}, {and}
  \bibinfo{person}{Martin Rinard}.} \bibinfo{year}{2018}\natexlab{}.
\newblock \showarticletitle{Probabilistic Programming with Programmable
  Inference}. In \bibinfo{booktitle}{\emph{Proceedings of the 39th ACM SIGPLAN
  Conference on Programming Language Design and Implementation}} (Philadelphia,
  PA, USA) \emph{(\bibinfo{series}{PLDI 2018})}. \bibinfo{publisher}{ACM},
  \bibinfo{address}{New York, NY, USA}, \bibinfo{pages}{603--616}.
\newblock
\showISBNx{978-1-4503-5698-5}
\urldef\tempurl%
\url{https://doi.org/10.1145/3192366.3192409}
\showDOI{\tempurl}


\bibitem[\protect\citeauthoryear{McCallum, Schultz, and Singh}{McCallum
  et~al\mbox{.}}{2009}]%
        {McCallum2009}
\bibfield{author}{\bibinfo{person}{A McCallum}, \bibinfo{person}{K Schultz},
  {and} \bibinfo{person}{S Singh}.} \bibinfo{year}{2009}\natexlab{}.
\newblock \showarticletitle{{Factorie: Probabilistic programming via
  imperatively defined factor graphs}}.
\newblock \bibinfo{journal}{\emph{Proc.~of NIPS}}  \bibinfo{volume}{22}
  (\bibinfo{year}{2009}), \bibinfo{pages}{1249--1257}.
\newblock
\showISBNx{9781615679119}
\showISSN{03643417}


\bibitem[\protect\citeauthoryear{Meinel and Theobald}{Meinel and
  Theobald}{1998}]%
        {meinel1998algorithms}
\bibfield{author}{\bibinfo{person}{Christoph Meinel} {and}
  \bibinfo{person}{Thorsten Theobald}.} \bibinfo{year}{1998}\natexlab{}.
\newblock \bibinfo{booktitle}{\emph{Algorithms and Data Structures in {VLSI}
  Design: {OBDD}-foundations and applications}}.
\newblock \bibinfo{publisher}{Springer Verlag}.
\newblock
\urldef\tempurl%
\url{https://doi.org/10.1007/978-3-642-58940-9}
\showDOI{\tempurl}


\bibitem[\protect\citeauthoryear{Minka, Winn, Guiver, Webster, Zaykov, Yangel,
  Spengler, and Bronskill}{Minka et~al\mbox{.}}{2014}]%
        {InferNET14}
\bibfield{author}{\bibinfo{person}{T. Minka}, \bibinfo{person}{J.M. Winn},
  \bibinfo{person}{J.P. Guiver}, \bibinfo{person}{S. Webster},
  \bibinfo{person}{Y. Zaykov}, \bibinfo{person}{B. Yangel}, \bibinfo{person}{A.
  Spengler}, {and} \bibinfo{person}{J. Bronskill}.}
  \bibinfo{year}{2014}\natexlab{}.
\newblock \bibinfo{title}{{Infer.NET 2.6}}.
\newblock
\newblock
\newblock
\shownote{Microsoft Research Cambridge.
  http://research.microsoft.com/infernet.}


\bibitem[\protect\citeauthoryear{Narayanan, Carette, Romano, Shan, and
  Zinkov}{Narayanan et~al\mbox{.}}{2016}]%
        {narayanan2016probabilistic}
\bibfield{author}{\bibinfo{person}{Praveen Narayanan}, \bibinfo{person}{Jacques
  Carette}, \bibinfo{person}{Wren Romano}, \bibinfo{person}{Chung{-}chieh
  Shan}, {and} \bibinfo{person}{Robert Zinkov}.}
  \bibinfo{year}{2016}\natexlab{}.
\newblock \showarticletitle{Probabilistic inference by program transformation
  in Hakaru (system description)}. In \bibinfo{booktitle}{\emph{International
  Symposium on Functional and Logic Programming - 13th International Symposium,
  {FLOPS} 2016, Kochi, Japan, March 4-6, 2016, Proceedings}}. Springer,
  \bibinfo{pages}{62--79}.
\newblock
\urldef\tempurl%
\url{https://doi.org/10.1007/978-3-319-29604-3_5}
\showDOI{\tempurl}


\bibitem[\protect\citeauthoryear{Nori, Hur, Rajamani, and Samuel}{Nori
  et~al\mbox{.}}{2014}]%
        {nori2014r2}
\bibfield{author}{\bibinfo{person}{Aditya~V Nori}, \bibinfo{person}{Chung-Kil
  Hur}, \bibinfo{person}{Sriram~K Rajamani}, {and} \bibinfo{person}{Selva
  Samuel}.} \bibinfo{year}{2014}\natexlab{}.
\newblock \showarticletitle{R2: An Efficient MCMC Sampler for Probabilistic
  Programs}. In \bibinfo{booktitle}{\emph{AAAI}}. \bibinfo{pages}{2476--2482}.
\newblock


\bibitem[\protect\citeauthoryear{Obermeyer, Bingham, Jankowiak, Pradhan, Chiu,
  Rush, and Goodman}{Obermeyer et~al\mbox{.}}{2019}]%
        {obermeyer2019tensor}
\bibfield{author}{\bibinfo{person}{Fritz Obermeyer}, \bibinfo{person}{Eli
  Bingham}, \bibinfo{person}{Martin Jankowiak}, \bibinfo{person}{Neeraj
  Pradhan}, \bibinfo{person}{Justin Chiu}, \bibinfo{person}{Alexander Rush},
  {and} \bibinfo{person}{Noah Goodman}.} \bibinfo{year}{2019}\natexlab{}.
\newblock \showarticletitle{Tensor variable elimination for plated factor
  graphs}.
\newblock  (\bibinfo{year}{2019}), \bibinfo{pages}{4871--4880}.
\newblock


\bibitem[\protect\citeauthoryear{Onisko}{Onisko}{2003}]%
        {onisko2003probabilistic}
\bibfield{author}{\bibinfo{person}{Agnieszka Onisko}.}
  \bibinfo{year}{2003}\natexlab{}.
\newblock \showarticletitle{Probabilistic causal models in medicine:
  Application to diagnosis of liver disorders}. In
  \bibinfo{booktitle}{\emph{Ph. D. dissertation, Inst. Biocybern. Biomed. Eng.,
  Polish Academy Sci., Warsaw, Poland}}.
\newblock


\bibitem[\protect\citeauthoryear{Pearl}{Pearl}{1988}]%
        {Pearl88b}
\bibfield{author}{\bibinfo{person}{Judea Pearl}.}
  \bibinfo{year}{1988}\natexlab{}.
\newblock \bibinfo{booktitle}{\emph{Probabilistic Reasoning in Intelligent
  Systems: Networks of Plausible Inference}}.
\newblock \bibinfo{publisher}{Morgan Kaufmann}.
\newblock


\bibitem[\protect\citeauthoryear{Pfeffer}{Pfeffer}{2007a}]%
        {pfeffer2007general}
\bibfield{author}{\bibinfo{person}{Avi Pfeffer}.}
  \bibinfo{year}{2007}\natexlab{a}.
\newblock \showarticletitle{A general importance sampling algorithm for
  probabilistic programs}.
\newblock  (\bibinfo{year}{2007}).
\newblock
\urldef\tempurl%
\url{http://nrs.harvard.edu/urn-3:HUL.InstRepos:25235125}
\showURL{%
\tempurl}


\bibitem[\protect\citeauthoryear{Pfeffer}{Pfeffer}{2007b}]%
        {Pfeffer2007}
\bibfield{author}{\bibinfo{person}{Avi Pfeffer}.}
  \bibinfo{year}{2007}\natexlab{b}.
\newblock \showarticletitle{{The Design and Implementation of IBAL: A
  General-Purpose Probabilistic Language}}.
\newblock \bibinfo{journal}{\emph{Introduction to statistical relational
  learning}} \bibinfo{number}{1993} (\bibinfo{year}{2007}),
  \bibinfo{pages}{399}.
\newblock
\showISBNx{0262072882}


\bibitem[\protect\citeauthoryear{Pfeffer}{Pfeffer}{2009}]%
        {pfeffer2009figaro}
\bibfield{author}{\bibinfo{person}{Avi Pfeffer}.}
  \bibinfo{year}{2009}\natexlab{}.
\newblock \showarticletitle{Figaro: An object-oriented probabilistic
  programming language}.
\newblock \bibinfo{journal}{\emph{Charles River Analytics Technical Report}}
  \bibinfo{volume}{137} (\bibinfo{year}{2009}).
\newblock


\bibitem[\protect\citeauthoryear{Pfeffer, Ruttenberg, Kretschmer, and
  OConnor}{Pfeffer et~al\mbox{.}}{2018}]%
        {pfeffer2018structured}
\bibfield{author}{\bibinfo{person}{Avi Pfeffer}, \bibinfo{person}{Brian
  Ruttenberg}, \bibinfo{person}{William Kretschmer}, {and}
  \bibinfo{person}{Alison OConnor}.} \bibinfo{year}{2018}\natexlab{}.
\newblock \showarticletitle{Structured Factored Inference for Probabilistic
  Programming}. In \bibinfo{booktitle}{\emph{International Conference on
  Artificial Intelligence and Statistics}}. \bibinfo{pages}{1224--1232}.
\newblock


\bibitem[\protect\citeauthoryear{Riguzzi and Swift}{Riguzzi and Swift}{2011}]%
        {Riguzzi2011TPLP}
\bibfield{author}{\bibinfo{person}{Fabrizio Riguzzi} {and}
  \bibinfo{person}{Terrance Swift}.} \bibinfo{year}{2011}\natexlab{}.
\newblock \showarticletitle{{The {PITA} System: Tabling and Answer Subsumption
  for Reasoning under Uncertainty}}.
\newblock \bibinfo{journal}{\emph{Theory and Practice of Logic Programming}}
  \bibinfo{volume}{11}, \bibinfo{number}{4--5} (\bibinfo{year}{2011}),
  \bibinfo{pages}{433--449}.
\newblock
\urldef\tempurl%
\url{https://doi.org/10.1017/S147106841100010X}
\showDOI{\tempurl}


\bibitem[\protect\citeauthoryear{Saad and Mansinghka}{Saad and
  Mansinghka}{2016}]%
        {saad2016}
\bibfield{author}{\bibinfo{person}{Feras Saad} {and} \bibinfo{person}{Vikash
  Mansinghka}.} \bibinfo{year}{2016}\natexlab{}.
\newblock \showarticletitle{A Probabilistic Programming Approach To
  Probabilistic Data Analysis}.
\newblock In \bibinfo{booktitle}{\emph{Advances in Neural Information
  Processing Systems (NIPS)}}.
\newblock


\bibitem[\protect\citeauthoryear{Sang, Beame, and Kautz}{Sang
  et~al\mbox{.}}{2005}]%
        {sang2005performing}
\bibfield{author}{\bibinfo{person}{Tian Sang}, \bibinfo{person}{Paul Beame},
  {and} \bibinfo{person}{Henry~A Kautz}.} \bibinfo{year}{2005}\natexlab{}.
\newblock \showarticletitle{Performing Bayesian inference by weighted model
  counting}. In \bibinfo{booktitle}{\emph{AAAI}}, Vol.~\bibinfo{volume}{5}.
  \bibinfo{pages}{475--481}.
\newblock


\bibitem[\protect\citeauthoryear{Sankaranarayanan, Chakarov, and
  Gulwani}{Sankaranarayanan et~al\mbox{.}}{2013}]%
        {Sankaranarayanan2013}
\bibfield{author}{\bibinfo{person}{Sriram Sankaranarayanan},
  \bibinfo{person}{Aleksandar Chakarov}, {and} \bibinfo{person}{Sumit
  Gulwani}.} \bibinfo{year}{2013}\natexlab{}.
\newblock \showarticletitle{Static Analysis for Probabilistic Programs:
  Inferring Whole Program Properties from Finitely Many Paths}.
\newblock \bibinfo{journal}{\emph{SIGPLAN Not.}} \bibinfo{volume}{48},
  \bibinfo{number}{6} (\bibinfo{date}{June} \bibinfo{year}{2013}),
  \bibinfo{pages}{447--458}.
\newblock
\showISSN{0362-1340}
\urldef\tempurl%
\url{https://doi.org/10.1145/2499370.2462179}
\showDOI{\tempurl}


\bibitem[\protect\citeauthoryear{Scutari and Denis}{Scutari and Denis}{2014}]%
        {scutari2014bayesian}
\bibfield{author}{\bibinfo{person}{Marco Scutari} {and}
  \bibinfo{person}{Jean-Baptiste Denis}.} \bibinfo{year}{2014}\natexlab{}.
\newblock \bibinfo{booktitle}{\emph{Bayesian networks: with examples in R}}.
\newblock \bibinfo{publisher}{CRC press}.
\newblock
\urldef\tempurl%
\url{https://doi.org/10.1111/biom.12369}
\showDOI{\tempurl}


\bibitem[\protect\citeauthoryear{Somenzi}{Somenzi}{[n.d.]}]%
        {Somenzi2010}
\bibfield{author}{\bibinfo{person}{Fabio Somenzi}.}
  \bibinfo{year}{[n.d.]}\natexlab{}.
\newblock \bibinfo{title}{{CUDD}: {BDD} package, {U}niversity of {C}olorado,
  {B}oulder.}
\newblock
\newblock


\bibitem[\protect\citeauthoryear{van~de Meent, Yang, Mansinghka, and
  Wood}{van~de Meent et~al\mbox{.}}{2015}]%
        {van2015particle}
\bibfield{author}{\bibinfo{person}{Jan-Willem van~de Meent},
  \bibinfo{person}{Hongseok Yang}, \bibinfo{person}{Vikash Mansinghka}, {and}
  \bibinfo{person}{Frank Wood}.} \bibinfo{year}{2015}\natexlab{}.
\newblock \showarticletitle{Particle Gibbs with Ancestor Sampling for
  Probabilistic Programs}. In \bibinfo{booktitle}{\emph{AISTATS}}.
\newblock


\bibitem[\protect\citeauthoryear{Van~den Broeck and Suciu}{Van~den Broeck and
  Suciu}{2017}]%
        {VdBFTDB17}
\bibfield{author}{\bibinfo{person}{Guy Van~den Broeck} {and}
  \bibinfo{person}{Dan Suciu}.} \bibinfo{year}{2017}\natexlab{}.
\newblock \bibinfo{booktitle}{\emph{Query Processing on Probabilistic Data: A
  Survey}}.
\newblock \bibinfo{publisher}{Now Publishers}.
\newblock
\urldef\tempurl%
\url{https://doi.org/10.1561/1900000052}
\showDOI{\tempurl}


\bibitem[\protect\citeauthoryear{Vazquez-Chanlatte and
  Seshia}{Vazquez-Chanlatte and Seshia}{2020}]%
        {vazquez2020_mce}
\bibfield{author}{\bibinfo{person}{Marcell Vazquez-Chanlatte} {and}
  \bibinfo{person}{Sanjit~A Seshia}.} \bibinfo{year}{2020}\natexlab{}.
\newblock \showarticletitle{Maximum Causal Entropy Specification Inference from
  Demonstrations}. In \bibinfo{booktitle}{\emph{International Conference on
  Computer Aided Verification}}. Springer.
\newblock


\bibitem[\protect\citeauthoryear{Vlasselaer, {Van den Broeck}, Kimmig, Meert,
  and {De Raedt}}{Vlasselaer et~al\mbox{.}}{2015}]%
        {Vlasselaer2015}
\bibfield{author}{\bibinfo{person}{Jonas Vlasselaer}, \bibinfo{person}{Guy {Van
  den Broeck}}, \bibinfo{person}{Angelika Kimmig}, \bibinfo{person}{Wannes
  Meert}, {and} \bibinfo{person}{Luc {De Raedt}}.}
  \bibinfo{year}{2015}\natexlab{}.
\newblock \showarticletitle{{Anytime inference in probabilistic logic programs
  with {T}p-compilation}}. In \bibinfo{booktitle}{\emph{{Proceedings of 24th
  International Joint Conference on Artificial Intelligence (IJCAI)}}}.
\newblock
\urldef\tempurl%
\url{https://doi.org/10.1016/j.ijar.2016.06.009}
\showDOI{\tempurl}


\bibitem[\protect\citeauthoryear{Wang, Hoffmann, and Reps}{Wang
  et~al\mbox{.}}{2018}]%
        {wang2018pmaf}
\bibfield{author}{\bibinfo{person}{Di Wang}, \bibinfo{person}{Jan Hoffmann},
  {and} \bibinfo{person}{Thomas Reps}.} \bibinfo{year}{2018}\natexlab{}.
\newblock \showarticletitle{PMAF: An Algebraic Framework for Static Analysis of
  Probabilistic Programs}.
\newblock \bibinfo{journal}{\emph{SIGPLAN Not.}} \bibinfo{volume}{53},
  \bibinfo{number}{4} (\bibinfo{date}{June} \bibinfo{year}{2018}),
  \bibinfo{pages}{513–528}.
\newblock
\showISSN{0362-1340}
\urldef\tempurl%
\url{https://doi.org/10.1145/3296979.3192408}
\showDOI{\tempurl}


\bibitem[\protect\citeauthoryear{Wingate and Weber}{Wingate and Weber}{2013}]%
        {wingate2013automated}
\bibfield{author}{\bibinfo{person}{David Wingate} {and}
  \bibinfo{person}{Theophane Weber}.} \bibinfo{year}{2013}\natexlab{}.
\newblock \showarticletitle{Automated variational inference in probabilistic
  programming}.
\newblock \bibinfo{journal}{\emph{arXiv preprint arXiv:1301.1299}}
  (\bibinfo{year}{2013}).
\newblock


\bibitem[\protect\citeauthoryear{Wood, Meent, and Mansinghka}{Wood
  et~al\mbox{.}}{2014}]%
        {wood2014new}
\bibfield{author}{\bibinfo{person}{Frank Wood}, \bibinfo{person}{Jan~Willem
  Meent}, {and} \bibinfo{person}{Vikash Mansinghka}.}
  \bibinfo{year}{2014}\natexlab{}.
\newblock \showarticletitle{A new approach to probabilistic programming
  inference}. In \bibinfo{booktitle}{\emph{Artificial Intelligence and
  Statistics}}. \bibinfo{pages}{1024--1032}.
\newblock


\bibitem[\protect\citeauthoryear{Zeng and Van~den Broeck}{Zeng and Van~den
  Broeck}{2020}]%
        {zeng2020efficient}
\bibfield{author}{\bibinfo{person}{Zhe Zeng} {and} \bibinfo{person}{Guy Van~den
  Broeck}.} \bibinfo{year}{2020}\natexlab{}.
\newblock \showarticletitle{Efficient search-based weighted model integration}.
  In \bibinfo{booktitle}{\emph{Uncertainty in Artificial Intelligence}}. PMLR,
  \bibinfo{pages}{175--185}.
\newblock


\bibitem[\protect\citeauthoryear{Zhou, Yang, Teh, and Rainforth}{Zhou
  et~al\mbox{.}}{2020}]%
        {zhou2019divide}
\bibfield{author}{\bibinfo{person}{Yuan Zhou}, \bibinfo{person}{Hongseok Yang},
  \bibinfo{person}{Yee~Whye Teh}, {and} \bibinfo{person}{Tom Rainforth}.}
  \bibinfo{year}{2020}\natexlab{}.
\newblock \showarticletitle{Divide, Conquer, and Combine: a New Inference
  Strategy for Probabilistic Programs with Stochastic Support}.
\newblock \bibinfo{journal}{\emph{International Conference on Machine
  Learning}} (\bibinfo{year}{2020}).
\newblock


\end{thebibliography}

\clearpage
\appendix
\section{Notation}
\label{app:notation}
\begin{description}
\item[Form] Denoted $\form_\tau(x)$, converts a syntactic variable $x$ into
  a tuple of type $\tau$:
  \begin{itemize}
  \item $\form_\bool(x) = \lx.$
  \item $\form_{\tau_1 \times \tau_2}(x) = (\form_{\tau_1}(x_l), \form_{\tau_2}(x_r))$.
  \end{itemize}
\item[Typed substitution] $\varphi_1[\lx \xmapsto{\tau} \tup_2]$, where $\Gamma(x)
  = \tau$ and $\Gamma \vdash \varphi_2 : \tau$. Intuitively, remaps each
  identifier in $\lx$ to its corresponding Boolean formula in $\tup_2.$
\begin{align*}
  \varphi_2[\lx \xmapsto{\bool} \varphi_1] \defeq \varphi_2[\lx \mapsto
\varphi_1],
\quad
\tup_2[\lx \xmapsto{\tau_{a} \times \tau_{b}} (\xpointphi_a,
\xpointphi_b)] \defeq \tup_2[\lx_l \xmapsto{\tau_{a}} \xpointphi_a][\lx_r \xmapsto{\tau_{b}}
\xpointphi_b],
\end{align*}
\begin{align*}
(\xpointphi_1, \xpointphi_2)[\lx \xmapsto{\tau} \xpointphi] \defeq 
(\xpointphi_1[\lx \xmapsto{\tau} \xpointphi], \xpointphi_2[\lx \xmapsto{\tau} \xpointphi]).
\end{align*}
\item[Broadcasted conjunction] $\varphi_a \xbroadand{\tau} \tup_b$, where
    $\Gamma(\varphi_a) = \bool$ and $\Gamma \vdash \tup_b : \tau$: conjoins a
    Boolean expression $\varphi_a$ with a tuple $\tup_b$. Intuitively, conjoins each
    element in $\tup_b$ with the Boolean expression $\varphi_a$:
    \begin{itemize}
    \item $\varphi_a \xbroadand{\bool} \tup_b \defeq \varphi_a \land \tup_b$.
    \item $\varphi_a \xbroadand{\tau_1 \times \tau_2} (\tup_b^l, \tup_b^r)
      \defeq \left(\varphi_a \xbroadand{\tau_1} \tup_b^l,~ \varphi_a \xbroadand{\tau_2} \tup_b^r\right)$
    \end{itemize}
  \item[Point-wise disjunction] $\tup_1 \xpointor{\tau} \tup_2$ where
    $\Gamma \vdash \tup_1 : \tau,  \Gamma \vdash \tup_2 : \tau$:
    \begin{itemize}
    \item $\tup_1 \xpointor{\bool} \tup_2 \defeq \tup_1 \lor \tup_2$.
    \item $(\tup_1^l, \tup_1^r) \xpointor{\tau_1 \times \tau_2}
      (\tup_2^l, \tup_2^r) \defeq (\tup_1^l \xpointor{\tau_1}
      \tup_2^l, \tup_1^r \xpointor{\tau_2} \tup_2^r)$.
    \end{itemize}
  \item[Pointwise iff] $\tup_1 \xLeftrightarrow{\tau} \tup_2$, where $\tup_1$ and $\tup_2$ are of
    type $\tau$.
    \begin{itemize}
    \item $\varphi_1 \xLeftrightarrow{\bool} \varphi_2 \defeq \varphi_1 \Leftrightarrow \varphi_2$.
    \item ${(\tup_1, \tup_2)}\xLeftrightarrow{\tau_1 \times \tau_2}{(\tup'_1, \tup'_2)} \defeq
      \left( {\tup_1}\xLeftrightarrow{\tau_1}{\tup_1'} \right) \land \left( {\tup_2}\xLeftrightarrow{\tau_2}{\tup_2'} \right)$.
    \end{itemize}
\end{description}

\section{Supplemental Experimental Results}
\begin{table}
 \begin{tabular}{lllll}
  \toprule
  <20 Nodes & 20---50 Nodes & 50---100 Nodes & 100---1000 Nodes & >1000 Nodes \\
   \midrule
   Survey (920B) & Alarm (18KB) & Hailfinder (91KB) & Pigs (213KB) & Munin (1.9MB) \\
   Cancer (725B) & Insurance (33KB) & Hepar2 (54KB) \\
              & Water (302KB) \\
                \bottomrule
\end{tabular}
\caption{Sizes of each discrete Bayesian network benchmark. Each network is
  annotated with a size in bytes (B) that is the size of the generated \dice{} program.}
\label{tbl:bnsize}
\end{table}

\begin{table*}

  \centering
  \begin{tabular}{l
    S[table-format=2.0, separate-uncertainty = true, table-figures-uncertainty=1]
    S[table-format=2.0, separate-uncertainty = true, table-figures-uncertainty=1]
    S[table-format=2.1, separate-uncertainty = true, table-figures-uncertainty=1]}
    \toprule
    Benchmark
    & {Psi (ms)}
    & {DP (ms)}
    & {\dice{} (ms)}
    \\
    \midrule
    Grass & 167 \pm 2 & 58 \pm 2 & 14.0 \pm 1.0 \\
    Burglar Alarm & 98 \pm 14 & 30 \pm 2 & 13 \pm 0.1 \\
    Coin Bias & 94 \pm 19 & 23 \pm 13 & 13.0 \pm 1.5 \\
    Noisy Or & 81 \pm 38 & 152 \pm 10 & 13.0 \pm 2.0 \\
    Evidence1 & 70 \pm 34 & 43 \pm 23 & 12.9 \pm 1.3 \\
    Evidence2 & 67 \pm 40 & 46 \pm 23 & 13.2 \pm 2.3  \\
    Murder Mystery & 193 \pm 33 & 75 \pm 10 & 13.6 \pm 1.6   \\
    \bottomrule
  \end{tabular}
  \caption{Comparison of inference algorithms on standard baselines (times are milliseconds). The
    reported time is the mean plus or minus a single standard deviation over 5 runs. }
  \label{tab:experiments-full}
\end{table*}

\begin{table*}
  \centering
  \begin{tabular}{l
    S[table-format=4.0, separate-uncertainty = true, table-figures-uncertainty=1]
    S[table-format=4.0, separate-uncertainty = true, table-figures-uncertainty=1]
    S[table-format=4.0, separate-uncertainty = true, table-figures-uncertainty=1]}
    \toprule
    Benchmark
    & {Psi (ms)}
    & {DP (ms)}
    & {\dice{} (ms)}
    \\
    \midrule
    Cancer~\citep{korb2010bayesian} & 772 \pm 60 & 46 \pm 2 & 13 \pm 3  \\
    Survey~\citep{scutari2014bayesian} & 2477 \pm 569 & 152 \pm 58 & 13 \pm 1 \\
    Alarm~\citep{beinlich1989alarm} & \xmark & \xmark & 25 \pm 3 \\
    Insurance~\citep{binder1997adaptive} & \xmark & \xmark & 212 \pm 12  \\
    Water~\citep{jensen1989expert} & \xmark & \xmark & 2590 \pm 21 \\
    Hailfinder~\citep{abramson1996hailfinder} &  \xmark & \xmark & 618 \pm 8 \\
    Hepar2~\citep{onisko2003probabilistic} & \xmark & \xmark & 48 \pm 6 \\
    Pigs & \xmark & \xmark & 72 \pm 2 \\
    Munin~\citep{andreassen1989munin} & \xmark & \xmark & 1866 \pm 27 \\
    \bottomrule
  \end{tabular}
  \caption{Comparison of inference algorithms on the single-marginal inference
    task (times are milliseconds). The reported time is the mean plus or minus a
    single standard deviation over 5 runs. A single standard deviation and the
    mean are reported.}
  \label{tab:bn-full}
\end{table*}

\begin{table}
  \centering
  \begin{tabular}{l
    S[table-format=10.0, separate-uncertainty = true, table-figures-uncertainty=1]
    S[table-format=10.0, separate-uncertainty = true, table-figures-uncertainty=1]}
    \toprule
    Benchmark
    & {\dice{} (ms)}
    & {\texttt{Ace} (ms)}
    \\
    \midrule
    Alarm & 159 \pm 12 & 422 \pm 32 \\
    Hailfinder & 1280 \pm 16 & 522 \pm 37 \\
    Insurance & 222 \pm 1 & 492 \pm 34 \\
    Hepar2& 163 \pm 3 & 495 \pm 17 \\
    Pigs & 11243 \pm 79 & 985 \pm 76 \\
    Water & 3320 \pm 118 & 605 \pm 10\\
    Munin & 4021194 \pm 2123290 & 3500 \pm 575 \\
    \bottomrule
  \end{tabular}
    \caption{\emph{All marginals}. A comparison between \dice{} and \texttt{Ace}
on the all-marginal discrete Bayesian network inference task. A single standard
deviation and the mean are reported.}
    \label{tbl:full-joint-full}
\end{table}

This section extends the experimental results by showing the mean and standard
deviation over at least 5 runs for all of the tables in the main body of the
paper. Table~\ref{tab:experiments-full} extends Table~\ref{tab:experiments},
Table~\ref{tab:bn-full} extends Table~\ref{tab:bnexperiments}, and
Table~\ref{tbl:full-joint-full} extends Table~\ref{tbl:full-joint}.

\section{Proofs}
\subsection{Key Lemmas}

\begin{lemma}[Independent Conjunction]
  \label{lem:wmcindconj}
  Let $\alpha$ and $\beta$ be Boolean sentences which share no variables; we
call such sentences \emph{independent}. Then, for any weight function $w$,
$\wmc(\alpha \land \beta, w) = \wmc(\alpha, w) \times \wmc(\beta, w)$.
\end{lemma}
\begin{proof}
  The proof relies on the fact that, if two sentences $\alpha$ and $\beta$ share
no variables, then any model $\omega$ of $\alpha \land \beta$ can be split into
two components, $\omega_\alpha$ and $\omega_\beta$, such that $\omega =
\omega_\alpha \land \omega_\beta$, $\omega_\alpha \Rightarrow \alpha$, and
$\omega_\beta \Rightarrow \beta$, and $\omega_\alpha$ and $\omega_\beta$ share
no variables. Then:
    $\wmc(\alpha \land \beta, w)
    = \textstyle \sum_{\omega \in \mods(\alpha \land \beta)} \textstyle \prod_{l \in \omega} w(l)
    =\left[ \textstyle \sum_{\omega_\alpha \in \mods(\alpha)} \textstyle\prod_{a \in \omega_\alpha} w(a) \right] \times \left[  \textstyle \sum_{\omega_\beta \in \mods(\beta)} 
        \textstyle\prod_{b \in \omega_\beta } w(b)\right]
    = \wmc(\alpha, w) \times \wmc(\beta, w).$
\end{proof}
\begin{proposition}[Inclusion-Exclusion]
For any two formulas $\varphi_1$ and $\varphi_2$ and weight function $w$,
$\wmc(\varphi_1 \lor \varphi_2, w) = \wmc(\varphi_1, w) + \wmc(\varphi_2, w) -
\wmc(\varphi_1 \land \varphi_2, w)$. Note the important \emph{mutual exclusion}
case when $\varphi_1 \land \varphi_2 = \false$.
\end{proposition}

\subsection{Correctness of Expression Compilation}
\label{sec:pf_expr_correct}
\begin{lemma}[Value Correctness]
  \label{lem:value_correct}
  For any values $v$ and $v'$ of type $\tau$, $\dbracket{v}(v') = \wmc(v
  \xLeftrightarrow{\tau} v', \emptyset)$.
\end{lemma}
\begin{proof}
  By induction on $\tau$:
  \begin{itemize}
  \item $\tau = \bool$. Then case analysis:
    \begin{itemize}
    \item $\dbracket{\true}(\true) = 1 = \wmc(\true \Leftrightarrow \true, \emptyset)$
    \item $\dbracket{\true}(\false) = 0 =\wmc(\true \Leftrightarrow \false, \emptyset)$
    \item $\dbracket{\false}(\false) = 1 = \wmc(\false\Leftrightarrow \false, \emptyset)$
    \item $\dbracket{\false}(\true) = 0 = \wmc(\false\Leftrightarrow \true, \emptyset)$
    \end{itemize}
  \item Inductive step: $\tau = \tau_1 \times \tau_2$. Then,
    \begin{align*}
      \dbracket{(v_1,v_2)}((v_1', v_2'))
      &= \dbracket{v_1}(v_1') \times \dbracket{v_2}(v_2') \\
      &= \wmc(v_1 \xLeftrightarrow{\tau_1} v_1', \emptyset) \times
        \wmc(v_2 \xLeftrightarrow{\tau_1} v_2', \emptyset) & \text{Induction Hyp.} \\
      &= \wmc(v_1 \xLeftrightarrow{\tau_1} v_1' \land v_2 \xLeftrightarrow{\tau_1} v_2', \emptyset)
      & \text{Independent Conj.} \\
      &= \wmc((v_1, v_2) \xLeftrightarrow{\tau_1 \times \tau_2} (v_1', v_2'), \emptyset).
    \end{align*}
  \end{itemize}
\end{proof}

\begin{lemma}[Typed Substitution]
  \label{lem:subst}
  For any values $v, v_x : \tau$, it holds that $(v \xLeftrightarrow{\tau} v_x) =
  (\form_{\tau}(x) \xLeftrightarrow{\tau} v)[\lx \xmapsto{\tau} v_x]$.
\end{lemma}
\begin{proof}
  By induction on $\tau$:
  \begin{itemize}
  \item $\tau = \bool$. Then, $(v \Leftrightarrow v_x) = (v \Leftrightarrow \lx)[\lx
    \mapsto v_x] = (v \Leftrightarrow \form_\bool(x))[\lx \mapsto v_x]$.
  \item $\tau = \tau_1 \times \tau_2$. Then, let $v = (v^l, v^r)$ and $v_x =
    (v_x^l, v_x^r)$. Then,
    \begin{align*}
      (v^l, v^r) \xLeftrightarrow{\tau_1 \times \tau_2} (v_x^l, v_x^r)
      &= (v^l \xLeftrightarrow{\tau_1} v_x^l) \land (v^r \xLeftrightarrow{\tau_2} v_x^r) \\
      &= (v^l \xLeftrightarrow{\tau_1} \form_{\tau_1}(x_l))[\lx_l \xmapsto{\tau_1} v_x^l] \land
        (v^r \xLeftrightarrow{\tau_2} \form_{\tau_2}(x_r))[\lx_r \xmapsto{\tau_1} v_x^r]  & \text{Ind. Hyp.}\\
      &=(v^l \xLeftrightarrow{\tau_1} \form_{\tau_1}(x_l) \land (v^r \xLeftrightarrow{\tau_2} \form_{\tau_2}(x_r)))
        [\lx_l \xmapsto{\tau_1} v_x^l][\lx_r \xmapsto{\tau_1} v_x^r]  & \\
      &= ((v^l, v^r) \xLeftrightarrow{\tau_1 \times \tau_2} \form_{\tau_1 \times \tau_2}(x))
        [\lx \xmapsto{\tau_1 \times \tau_2} (v_x^l, v_x^r)].
    \end{align*}
  \end{itemize}
\end{proof}

\begin{lemma}[Typed Correctness Without Procedures]
  Let $\te$ be a \dice{} expression without procedure calls.
  Let $\{x_i : \tau_i\} \vdash \te : \tau \rightsquigarrow (\tup, \obs, w)$. Then for
  any values $\{v_i : \tau_i\}$ and $v : \tau$, we have that $\dbracket{\te[x_i \mapsto v_i]}(v)
  = \wmc\left(\big(({v}\xLeftrightarrow{\tau} {\varphi}) \land \obs\big)[\lx_i \xmapsto{\tau_i} v_i], w\right)$.
\end{lemma}
\begin{proof}
The proof is by structural induction on the syntax of Boolean \dice{} programs.
First, we prove that the theorem holds for the non-inductive terms:
  \begin{itemize}[leftmargin=*]
  \item $\te = \true$ and $\te = \false$ follow directly from Lemma~\ref{lem:value_correct}.
  \item $\te = \Lflip{\theta}.$  Then, $\Gamma \vdash \Lflip{\theta} : \bool \rightsquigarrow
    (\mathbf{f}, \true, w)$ for a fresh $\mathbf{f}$. Then, $\wmc(\mathbf{f}
    \land \true, w) = \theta = \dbracket{\Lflip{\theta}}(\true)$ and
    $\wmc(\overline{f}, w) = 1-\theta = \dbracket{\Lflip{\theta}}(\false)$.
  \item $\te = x$. Then, $\Gamma \vdash x : \tau \rightsquigarrow (\tup, \true,
    \emptyset)$, and let $v_x : \tau$ be the value substituted for $x$.
    \begin{align*}
      \dbracket{x[x \xmapsto{\tau} v_x]}(v)
      &= \dbracket{v_x}(v)\\
      &=  \wmc(({v_x} \xLeftrightarrow{\tau} {v}) \land \true, \emptyset) & \text{Lemma~\ref{lem:value_correct}}\\
      &= \wmc\Big(((\form_\tau(x) \xLeftrightarrow{\tau} v) \land \true)[\lx \xmapsto{\tau} v_x], \emptyset\Big) & \text{Lemma~\ref{lem:subst}}
    \end{align*}
  \item $\te = \Lfst{x}$. Assume $\Gamma(x) = \tau_1 \times \tau_2$. Then, $\Gamma \vdash \Lfst{x} : \tau_1 \rightsquigarrow
    (\form_{\tau_1}(x_l), \true, \emptyset)$. Let $v_x = (v_x^l, v_x^r) : \tau_1
    \times \tau_2$ be the value substituted for $x$. Then,
    \begin{align*}
      \dbracket{\Lfst{x}[x \xmapsto{\tau \times \tau'} v_x]}(v)
      &= \dbracket{v_x^l}(v)\\
      &=\wmc((v_x^l \xLeftrightarrow{\tau} v) \land \true, \emptyset) &\text{Lemma~\ref{lem:value_correct}}\\
      &= \wmc\Big(\big((\form_{\tau}(x_l) \xLeftrightarrow{\tau} v)  \land \true\big)[\lx \xmapsto{\tau \times \tau'} v_x], \emptyset\Big)
        & \text{Lemma~\ref{lem:subst}}
    \end{align*}
    An analogous argument holds for $\Lsnd{x}$.
  \item $\te = (x_1, x_2)$. Then, $\Gamma \vdash (x_1, x_2) : \tau_1 \times
    \tau_2 \rightsquigarrow ((\form_{\tau_1}(x_1), \form_{\tau_2}(x_2)), \true, \emptyset)$. Let $v_1 : \tau_1$
    and $v_2: \tau_2$ be the value substituted for $x_1$ and $x_2$ respectively,
    and let $v = (v^l, v^r)$.
    Then,
    \begin{align*}
      &\dbracket{(x_1, x_2)[x_1 \xmapsto{\tau_1} v_1, x_2 \xmapsto{\tau_2} v_2]}((v^l, v^r))\\
      &= \dbracket{(v_1, v_2)}(v^l, v^r)\\
      &= \wmc\Big(({v_1} \xLeftrightarrow{\tau_1} {v^l}) \land (v_2 \xLeftrightarrow{\tau_2} v^r) \land \true, \emptyset\Big) 
      & \text{Lemma~\ref{lem:value_correct}} \\
      &= \wmc\Big(\big(\form_{\tau_1}(x_1) \xLeftrightarrow{\tau_1} v^l \big) \land
        \big(\form_{\tau_2}(x_2) \xLeftrightarrow{\tau_2} v^r \big) \land \true
        [x_1 \xmapsto{\tau_1} v_1, x_2 \xmapsto{\tau_2} v_2], \emptyset\Big) & \text{Lemma~\ref{lem:subst}}
    \end{align*}
  \end{itemize}

  \noindent
  Now for the inductive terms:

  \begin{itemize}[leftmargin=*]
  \item $\te = \Llet{\te_1}{\te_2}$. Assume $\Gamma \vdash \te_1 : \tau_1 \rightsquigarrow (\tup_1,
    \obs_1, w_1)$ and $\Gamma \cup \{x : \tau_1\} \vdash \te_2 : \tau_2 \rightsquigarrow (\tup_2, \obs_2, w_2)$.
    For notational simplicity, assume that the substitution $[x_i \xmapsto{\tau_i} v_i]$
    has been applied to $\tup_1, \obs_1, \tup_2, \obs_2$, and that all
    weighted model counts are performed with the weight $w_1 \cup w_2$. Then,
    \begin{align*}
      &\dbracket{(\Llet{x=\te_1}{\te_2})[x_i \mapsto v_i]}(\true)\\
      &= \sum_{v}\dbracket{\te_1[x_i \mapsto v_i]}(v) \times \dbracket{\te_2[x_i \mapsto v_i, x \mapsto v]}(\true)\\
      &= \sum_{v_x \in \tau_1}\wmc\big((\tup_1 \xLeftrightarrow{\tau_1} v_x) \land \obs_1 \big)
        \times \wmc\big(((\tup_2 \xLeftrightarrow{\tau_2} v)\land \obs_2)[\lx \xmapsto{\tau_1} v_x]\big)
      & \text{Ind. Hyp.} \\
      &= \sum_{v_x \in \tau_1}\wmc\Big((\tup_1 \xLeftrightarrow{\tau_1} v_x) \land \obs_1 
        \land \big((\tup_2 \xLeftrightarrow{\tau_2} v)\land \obs_2\big)[\lx \xmapsto{\tau_1} v_x] \Big)
      &\text{Indep. Conj.}\\
      &= \wmc\Big(\bigvee_{{v_x \in \tau_1}}(\tup_1 \xLeftrightarrow{\tau_1} v_x) \land \obs_1 
        \land \big((\tup_2 \xLeftrightarrow{\tau_2} v)\land \obs_2\big)[\lx \xmapsto{\tau_1} v_x] \Big)
      &\text{Mut. Excl.}\\
      &= \wmc\left( ((\tup_2 \xLeftrightarrow{\tau_2} v_2) \land \obs_1 \land \obs_2)[\lx \xmapsto{\tau_1} \tup_1] \right)
    \end{align*}

  \item $\te = \Lobs{g}$. Assume $\Gamma \vdash g : \bool \rightsquigarrow (\varphi, \true, w)$.
    This case relies on interpreting the semantics of $\dbracket{\Lobs{g}[x_i
      \mapsto v_i]}(v)$ as $\dbracket{g[x_i \mapsto v_i]}(\true) \times \dbracket{\true}(v)$.
    Then,
    \begin{align*}
      \dbracket{\Lobs{g}[x_i \mapsto v_i]}(v)
      &= \dbracket{g[x_i \mapsto v_i]}(\true) \times \dbracket{\true}(v) \\
      &= \wmc(\varphi \land \true, w) \times \wmc(v \land \true). & \text{Ind. Hyp.}\\ 
      &= \wmc(\varphi \land v, w). & \text{Indep. Conj.}
    \end{align*}
  \item 
    $\te = \Lite{g}{\te_T}{\te_E}$. Assume $\Gamma \vdash g : \bool \rightsquigarrow(\varphi_g, \true,
    w_g)$, $\Gamma \vdash \te_T : \tau \rightsquigarrow (\tup_T, \obs_T, w_T)$, 
    $\Gamma \vdash \te_E : \tau \rightsquigarrow (\tup_E, \obs_E, w_E)$. Again assume for
    notational simplicity that all weighted model counts are performed with the
    weight function $w_g \cup w_2 \cup w_g$ and that the substitutions $[x_i
    \xmapsto{\tau_i} v_i]$ have been performed on the compiled formulas. Then,
    \begin{align*}
      &\dbracket{\Lite{g}{\te_T}{\te_E}}(v)\\
      &= \dbracket{g}(\true) \times \dbracket{\te_T}(v) +
         \dbracket{g}(\false) \times \dbracket{\te_E}(v) \\
      &= \wmc(\varphi_g \land \true) \times \wmc((\tup_T \xLeftrightarrow{\tau} v) \land \obs_T) +
        \wmc(\overline{\varphi}_g \land \true) \times \wmc((\tup_E \xLeftrightarrow{\tau} v) \land \obs_E) & \text{Ind. Hyp.}\\
      &= \wmc(\varphi_g \land (\tup_T \xLeftrightarrow{\tau} v) \land \obs_T) +
        \wmc(\overline{\varphi}_g \land (\tup_E \xLeftrightarrow{\tau} v) \land \obs_E) & \text{Indep. Conj.}\\
      &= \wmc((\varphi_g \land (\tup_T \xLeftrightarrow{\tau} v) \land \obs_T) \lor
        (\overline{\varphi}_g \land (\tup_E \xLeftrightarrow{\tau} v) \land \obs_E)) & \text{Mut. Excl.}\\
      &= \wmc\left(\Big((\varphi_g \xbroadand{\tau} \tup_T) \xpointor{\tau}
        (\overline{\varphi}_g \xbroadand{\tau} \tup_E)\Big) \xLeftrightarrow{\tau} v \land
        \Big((\varphi_g \land \obs_T) \lor (\overline{\varphi}_g\land\obs_E)\Big)\right)
    \end{align*}
  \end{itemize}
\end{proof}

\subsection{Theorem~\ref{thm:prog_correctness}}
\label{app:prog_correct}
First we extend Lemma~\ref{lem:correct_no_proc} to show that Boolean function call
compilation is correct. First we need some preliminaries. The semantics and
compilation of an expression can only be compared if the function context they
are compiled in is \emph{compatible}:
\begin{definition}[Table Compatibility]
  Let $\Phi$ be a compiled function table, $T$ be a function table, and $\Gamma$
be a type environment. Then we say $T$ and $\Phi$ are \emph{compatible} if for
any function identifier $x$, where $\Gamma(x) = \tau_1 \rightarrow \tau_2$ and
$\Phi(x) = (\lx, \tup, \obs, w)$, it holds for any argument value $v^x:
\tau_1$ and value $v:\tau_2$, $T(x)(v^x)(v) = \wmc\big(((\tup
\xLeftrightarrow{\tau_2} v) \land \obs)[\lx \xmapsto{\tau_1} v^x], w\big)$.
\end{definition}

Then, we can extend Lemma~\ref{lem:correct_no_proc} to assume compatible tables:
\begin{theorem}[Boolean Correctness with Procedure Calls]
  \label{thm:correctness_fcall}
  Let $\te$ be a \dice{} expression with function calls, $T$ and $\Phi$ be
compatible tables, let $\{x_i: \tau_i\}, \Phi \vdash \te : \tau \comp (\varphi, \obs,
w)$. Then, for any values $\{v_i : \tau_i\}$ and $v:\tau$, we have that 
$\dbracket{\te[x_i \mapsto v_i]}(v) = \wmc\Big(\big((\varphi \xLeftrightarrow{\tau}
v) \land \obs\big) [\lx_i \xmapsto{\tau_i} v_i]\Big)$.
\end{theorem}
\begin{proof}
  The proof is identical to the proof of Lemma~\ref{lem:correct_no_proc} except
  for the addition of the function call syntax, which we prove here.

  Assume $\te = x_1(x_2)$ and assume $\Phi(x_1) = (\lx_{arg}, \tup, \obs,
  w)$. Assume $(\tup', \obs', w) = \texttt{RefreshFlips}(\tup, \obs, w)$
  Then, $x_1(x_2) \rightsquigarrow (\tup[\lx_{arg} \mapsto \lx_2],
  \obs[\lx_{arg} \mapsto \lx_2], w)$. Then the result follows directly from
  table compatibility:
  \begin{align*}
    \dbracket{x(v^x)}(\true)
    &= T(x)(v^x)(\true) \\
    &= \wmc\big(((\tup \xLeftrightarrow{\tau_2} v) \land \obs)[\lx \xmapsto{\tau_1} v^x], w\big) & \text{Table Compatibility}\\
    &= \wmc\big(((\tup' \xLeftrightarrow{\tau_2} v) \land \obs')[\lx \xmapsto{\tau_1} v^x], w\big)
    & \text{Defn. of \texttt{RefreshFlips}}
  \end{align*}
\end{proof}

Now we are ready for the main theorem:
\begin{theorem}[Typed Program Correctness]
  Let $\prog$ be a \dice{} program $\Gamma \vdash \prog : \tau
\rightsquigarrow (\tup, \obs, w)$. Then for any $v : \tau$, we have that $\dbracket{\prog}(v) =
\wmc((\tup \xLeftrightarrow{\tau} v) \land \obs, w)$.
\end{theorem}
\begin{proof}
  
  \begin{itemize}[leftmargin=*]
  \item \emph{Base case}: $\prog = \te$. Assume $\Gamma, \Phi ~~\bullet \te :
    \tau \rightsquigarrow (\tup, \obs, w)$. Then,
      $\dbracket{\bullet \te}(v)
      = \dbracket{\te}(v)
      = \wmc((\tup \xLeftrightarrow{\tau} v) \land \obs, w)$, by Theorem~\ref{thm:correctness_fcall}.
  \item \emph{Inductive step}: The program is of the form
    $\prog_1 = \texttt{fun}~x_1(x_2)  ~ \{\te\} ~\prog_2$.

Assume that $\Gamma, \Phi \vdash \texttt{fun}~x_1(x_2) ~ \{\te\} : \tau_1
\rightarrow \tau_2 \rightsquigarrow (\tup_f, \obs_f, w_f)$. Let $T' = T \cup
\{x_1 \mapsto \dbracket{\texttt{func}}\} $ and $\Phi' = \Phi \cup \big\{x_1
\mapsto (\lx_2, \tup_f, \obs_f, w_f)\big\}$. Then,
Theorem~\ref{thm:correctness_fcall} guarantees that $T'$ and $\Phi'$ are
compatible tables. Let $\Gamma \cup \{x_1 \mapsto \tau_1 \rightarrow \tau_2\},
\Phi' \vdash \prog_2 : \tau \rightsquigarrow (\tup, \obs, w)$. Then,
\begin{align*}
  \dbracket{\texttt{fun}~x_1(x_2)  ~ \{\te\} ~\prog_2}^T(v)
  =&\dbracket{\prog_2}^{T'}(v) \\
  =& \wmc\Big(\big(\tup \xLeftrightarrow{\tau} v \big) \land \obs, w\Big) & \text{By Ind. Hyp.}
\end{align*}
  \end{itemize}

\end{proof}

Finally we prove Theorem~\ref{thm:compilation_correct}, restated here for
convenience:
\begin{theorem}[Compilation Correctness]
  Let $\prog$ be a \dice{} program and $ \emptyset, \emptyset \vdash \prog : \tau \rightsquigarrow
(\xpointphi, \obs, w)$.  Then:
\begin{itemize}
\item $\dbracket{\prog}_A
= \wmc(\obs, w)$

\item for any value $v : \tau$, $\dbracket{\prog}_D(v)
= \wmc((\xpointphi \xLeftrightarrow{\tau} v) \land \obs, w) / \wmc(\obs, w)$.
\end{itemize}
\end{theorem}
\begin{proof}
  Let $\{\}, \{\} \vdash \prog : \tau \rightsquigarrow (\tup, \obs, w)$. Then,
  \begin{align*}
    \dbracket{\prog}_A &= \sum_{v} \wmc((\tup \xLeftrightarrow{\tau} v) \land \obs, w) & \text{Theorem~\ref{thm:prog_correctness}} \\
                       &= \wmc\left( \bigvee_v ((\tup \xLeftrightarrow{\tau} v) \land \obs), w \right) & \text{Mut. Excl.} \\
                       &= \wmc(\gamma, w).
  \end{align*}
  Then, $\dbracket{\prog}_D(v) = \dbracket{\prog}(v) / \sum_{v'}
  \dbracket{\prog}(v') = \wmc((\tup \xLeftrightarrow{\tau} v) \land \obs,
  w)/\wmc(\gamma, w)$ by Theorem~\ref{thm:prog_correctness} and the above argument.
\end{proof}

\section{Implementation Details}
\label{sec:impl_details}
\subsection{Variable Ordering} The \emph{variable ordering} --- the order in
which variables are branched on in a BDD --- is a critical parameter that
determines how compactly a BDD can represent a particular logical
formula~\citep{meinel1998algorithms, Bryant86}. Finding the optimal order --- the
one that minimizes the size of the BDD --- is NP-hard, so one must typically
resort to heuristics for choosing orderings that work well in practice. 
\dice{} orders variables according to the syntactic order in which they occur
in the program, mirroring the topological variable ordering heuristic from
Bayesian networks~\citep{Darwiche09}. We anticipate future work in deriving more
sophisticated variable ordering heuristics from static program analyses. 

\subsection{Multi-rooted BDDs}
\begin{wrapfigure}{R}{0.1\textwidth}
     \begin{tikzpicture}
    \def\lvl{22pt}
    \node (f2) at (0bp,0bp) [bddnode] {$f_2$};
    \node (f1) at ($(f2) + (-20bp, -\lvl)$) [bddnode] {$f_1$};

    \node (l) at ($(f1) + (0bp, \lvl)$) {$\mathtt{fst}$};
    \node (r) at ($(f2) + (0bp, \lvl)$) {$\mathtt{snd}$};
    \node (true) at ($(f1) + (0bp, -\lvl)$) [bddterminal] {$\true$};
    \node (false) at ($(f2) + (0bp, -2*\lvl)$) [bddterminal] {$\false$};
    
    \begin{scope}[on background layer]
      \draw [highedge] (f1) -- (true);
      \draw [lowedge] (f1) -- (false);
      \draw [highedge] (f2) -- (f1);
      \draw [lowedge] (f2) -- (false);
      \draw [->] (l) -- (f1);
      \draw [->] (r) -- (f2);
    \end{scope}
  \end{tikzpicture}
  \caption{}
  \label{fig:mult_root}
  \vspace{-1cm}
\end{wrapfigure}
\dice{} typically needs to represent many BDDs at the same time that share
structure. The accepting and unnormalized formulas may share
sub-formulas, or tuples may compile to formulas that share some substructure. 
\emph{Multi-rooted BDDs} naturally exploit this repeated substructure to
compactly represent multiple Boolean formulas in a single data structure. For
instance, the following example program that returns a tuple is compiled into
the multi-rooted BDD in Figure~\ref{fig:mult_root}:
\begin{center}
\texttt{let x = flip$_1$ 0.6 in let y = x $\land$ flip$_2$
  0.4 in (x, y)}
\end{center}

\subsection{Proposition~\ref{prop:composition}}

\label{sec:comp_proof}
\begin{proof}[Proof Sketch]
  The proof is by construction. For instance, for conjunction, BDDs for $B_1$ and
  $B_2$ are of the form:
  \begin{align*}
      B_1 =
  \scalebox{0.7}
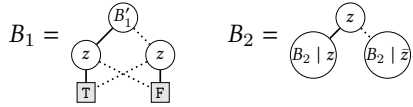
{
  \begin{tikzpicture}[baseline=-3ex]
      \def\lvl{20pt}
    \node (b) at (0, 0) [bddnode] {$B_1'$};
    \node (z1) at ($(b) + (-20bp, -\lvl)$) [bddnode] {$z$};
    \node (z2) at ($(b) + (20bp, -\lvl)$) [bddnode] {$z$};
    \node (true) at ($(z1) + (0bp, -\lvl)$) [bddterminal] {$\true$};
    \node (false) at ($(z2) + (0bp, -\lvl)$) [bddterminal] {$\false$};
    \begin{scope}[on background layer]
      \draw [highedge] (b) -- (z1);
      \draw [lowedge] (b) -- (z2);
      \draw [highedge] (z1) -- (true);
      \draw [lowedge] (z1) -- (false);
      \draw [highedge] (z2) -- (false);
      \draw [lowedge] (z2) -- (true);
    \end{scope}
  \end{tikzpicture}}
    \qquad
    B_2=
  \scalebox{0.7}{
  \begin{tikzpicture}[baseline=-3ex]
      \def\lvl{20pt}
    \node (z) at (0, 0) [bddnode] {$z$};
    \node (b2z) at ($(z) + (-20bp, -\lvl)$) [bddnode] {$B_2 \mid z$};
    \node (b2nz) at ($(z) + (20bp, -\lvl)$) [bddnode] {$B_2 \mid \bar{z}$};
    \begin{scope}[on background layer]
      \draw [highedge] (z) -- (b2z);
      \draw [lowedge] (z) -- (b2nz);
    \end{scope}
  \end{tikzpicture}}
  \end{align*}
  where $B_1'$ is the BDD for $B_1$ with $z$ separated out and $B_2 \mid z$ is
the BDD for $B_2$ with $z=\true$. The BDD for $B_1 \land B_2$ can be constructed
in linear time by traversing $B_1'$ and rerouting all high edges coming from $z$
to that end in $\true$ to $B_2 \mid z$, and all low edges from $z$ that end in
$\true$ to $B_2 \mid \bar{z}$.
\begin{tikzpicture}
  
\end{tikzpicture}
\end{proof}

\subsection{Theorem~\ref{thm:hard}}
\label{sec:hard}
\begin{proof}[Proof Sketch]
  The \textsf{PSPACE}-hardness of \dice{} inference follows directly from the
expressiveness of non-recursive Boolean programs. In particular, there is a
polynomial-time reduction from the {\em quantified Boolean formula} (QBF)
problem, which is \textsf{PSPACE}-complete, to such a program. This reduction
can also be used to reduce QBF to the problem of determining the probability
that a \dice{} program outputs true.
In particular, the construction relies on the expressiveness of nested function calls.
Each nested function call corresponds to either a universal or existential
quantifier, and the innermost call can be such that it evaluates a
fully-quantified CNF.
\end{proof}

\end{document}